\documentclass[11pt]{article}

\usepackage[utf8]{inputenc}
\usepackage[T1]{fontenc}
\usepackage{lmodern}
\usepackage{amsmath,amssymb,amsthm}
\usepackage{geometry}
\usepackage{hyperref}
\usepackage{booktabs}
\usepackage{tabularx}
\usepackage{float}
\usepackage{xcolor}
\usepackage{graphicx}
\usepackage{enumitem}
\usepackage{tikz}
\usepackage{placeins}
\usepackage{hyperref}
\usepackage{xurl}
\hypersetup{breaklinks=true}
\usetikzlibrary{positioning,arrows.meta}

\hypersetup{
  colorlinks=true,
  linkcolor=blue,
  citecolor=blue,
  urlcolor=blue
}

\newtheorem{definition}{Definition}
\newtheorem{theorem}{Theorem}
\newtheorem{lemma}{Lemma}
\newtheorem{assumption}{Assumption}
\newtheorem{proposition}{Proposition}

	\title{The Treasury Proof Ledger: A Cryptographic Framework for Accountable Bitcoin Treasuries}

\author{
\begin{tabular}[t]{@{}c@{}}
Jose E. Puente\\
\textit{402T Labs, SL}\\
\texttt{jepuente@402T.com}
\end{tabular}
\and
\begin{tabular}[t]{@{}c@{}}
Carlos Puente\\
\textit{402T Labs, SL}\\
\texttt{carlospuenteg@402T.com}
\end{tabular}
}

\date{}
\begin{document}
\maketitle

\begin{abstract}

Public companies and institutional investors that hold Bitcoin face a tension between growing demands for transparency and the need to protect trading strategies, counterparties, and internal controls. Existing proof-of-reserves (PoR) deployments are typically point solutions: they certify assets at a single custodian or exchange at a single point in time, and they rarely model liabilities, encumbrances, or exposures that span multiple on-chain and off-chain domains or families of related entities.

We propose the \emph{Treasury Proof Ledger} (TPL), a Bitcoin-anchored accountable logging framework for multi-domain Bitcoin treasuries. A TPL instance treats the treasury as a collection of heterogeneous exposure domains (for example spot holdings, lending books, derivatives, and encumbered positions) linked by explicit balance-update rules and an explicit fee sink. The log records PoR snapshots, proof-of-transit (PoT) receipts for transfers between domains, and policy metadata, and supports policy-based views for different stakeholders.

Formally, we fix an idealised TPL abstraction and give a state-machine model for a multi-domain Bitcoin treasury that tracks
exposures across on-chain and off-chain domains. Our main contribution is a deployment-level, definitional and compositional
framework: we identify the core treasury state machine and a family of security notions (exposure soundness, policy completeness,
non-equivocation, and privacy-compatible policy views) and show how restricted variants of these notions can be realised using
standard building blocks such as proof-of-reserves, proof-of-transit logging, and hash-based commitment primitives. Our results
are existence-style: they isolate the cryptographic guarantees available once the economic and governance modelling assumptions
are fixed, without claiming that any particular existing deployment already meets them.

Finally, we illustrate the modelling power and practicality of TPL on a stylised corporate-treasury deployment and discuss how TPL can act as a cryptographic substrate for future cross-institution and macroprudential supply-consistency checks against Bitcoin's fixed monetary cap.

On the privacy side, our formal guarantees are deliberately restricted: we prove a simulation-based theorem only for a specific public-investor leakage profile, and leave privacy for richer policies and observers as open work.
\end{abstract}

\section{Introduction}

Bitcoin treasury companies, publicly traded firms holding significant amounts of Bitcoin, have emerged as a macro-financial phenomenon. Meister~\cite{Meister2025} shows that these firms behave like highly levered Bitcoin portfolios whose valuations incorporate equity premia and structural leverage effects. Wen~\cite{Wen2025} formally models the survival of digital asset treasury (DAT) firms during bear markets, highlighting treasury policies and liquidity management as crucial determinants of long-term viability.

However, neither line of research resolves a central governance challenge: \emph{how should public treasuries provide transparency about their Bitcoin holdings?}

Conventional proof-of-reserve (PoR) systems used by exchanges typically require public disclosure of wallet addresses and signed messages proving coin control. Treasury operators, most notably Michael Saylor, have publicly argued that for public companies this model is \emph{insecure}, \emph{destabilising}, and ``a bad idea''~\cite{Saylor2025} in its naive form. Specifically, publishing real-time operational wallet maps:

\begin{itemize}[noitemsep]
\item Increases physical and cyberattack surface,
\item Facilitates targeted surveillance and front-running,
\item Creates misinterpretation risks when routine treasury movements appear as ``unusual'' flows,
\item Conflicts with regulated internal controls.
\end{itemize}

Despite this, markets, regulators, and auditors increasingly expect Bitcoin treasuries to provide stronger, cryptographically grounded transparency.

At the policy and supervisory level, international standard-setters and regional regulators have begun to formalise expectations for crypto-asset risk management and transparency. The Financial Stability Board's recommendations on the regulation, supervision and oversight of crypto-asset activities, together with the European Union's Markets in Crypto-Assets (MiCA) framework and related ESMA guidance, call for robust, auditable processes for asset segregation, reserve management and disclosure of material risks~\cite{FSB2023CryptoAssets,ESMA2023MiCA}. At the same time, a growing number of listed companies, funds and asset managers now hold Bitcoin in their treasuries, and there is an active debate in both practitioner and academic circles about proof-of-reserves practices, the auditability of exchange and custodian claims, and the role of cryptographic proofs in corporate reporting~\cite{Fidelity2023BitcoinTreasury,Dagher2015,Lazirko2025,VidalTomas2024}. In this environment, TPL is intended as a candidate cryptographic foundation for the type of treasury-level disclosures and replayable audit trails that these regulatory frameworks appear to be moving towards.

This paper introduces the \textbf{Treasury Proof Ledger (TPL)}, a Bitcoin-anchored transparency and state-machine abstraction for Bitcoin treasury transparency.
Our main technical novelty is a treasury-level cryptographic abstraction that composes heterogeneous proof-of-reserves (PoR) and proof-of-transit (PoT) mechanisms into a Bitcoin-anchored, multi-domain exposure state with deployment-level guarantees (exposure soundness, policy completeness, and privacy-compatible policy views).

We make four contributions:

\begin{itemize}[noitemsep]
\item We fix a treasury-level, multi-domain transparency abstraction by formalising a Bitcoin treasury as a multi-domain exposure vector with an explicit fee sink and a balance-update rule, and we prove a conservation theorem (Theorem~\ref{thm:conservation}) for value flows within this state machine.
\item We define deployment-level security notions: exposure soundness, policy completeness, non-equivocation and view correctness, and we show that an idealised TPL built from standard PoR, PoT and hash-chained commitments realises three core integrity properties (Theorems~\ref{thm:conservation},~\ref{thm:non-equivocation} and~\ref{thm:view-correctness}).
\item We prove restricted deployment-level guarantees for real treasuries: a restricted exposure-soundness theorem for closed sets of domains under faithful policies (Theorem~\ref{lem:restricted-exp-soundness}) and a restricted privacy theorem for a simple public-investor policy (Theorem~\ref{thm:pub-privacy}). Table~\ref{tab:assumption-map} summarises which cryptographic and organisational assumptions each result relies on.
\item We illustrate the modelling power and practicality of TPL by sketching a minimal deployment architecture and cost model, and by explaining how families of TPL instances can serve as a cryptographic substrate for future cross-institution supply-consistency checks against Bitcoin's fixed monetary cap.
\end{itemize}

\paragraph{Contributions in context.}
Existing proof-of-reserves (PoR) schemes for exchanges are typically single-domain, point-in-time checks that do not model encumbrances or cross-domain flows, and they do not expose deployment-level notions such as exposure soundness or policy completeness for treasuries.
Transparency-log and accountable-logging frameworks, including certificate-transparency style systems, provide append-only event logs but treat entries as generic statements, with no conservation law over a multi-domain monetary state or policy-based stakeholder views.
TPL differs in that the logged state is an explicit multi-domain exposure vector with a conservation invariant and fee sink, together with deployment-level security notions and policy-based views tailored to Bitcoin treasury governance and to aggregate-supply-consistency checks across families of treasuries.

Formally, our cryptographic contribution is to fix an idealised TPL abstraction and analyse deployment-level security notions for Bitcoin treasuries in that model. On the privacy side, we emphasise that all formal results in this paper concern the restricted public-investor setting and do not yet provide general privacy guarantees for arbitrary policies.
Our main cryptographic contribution is definitional and compositional: we identify deployment-level notions of exposure soundness, policy completeness, and privacy for multi-domain Bitcoin treasuries, and show that restricted versions of these notions can be realised by a generic TPL state machine built from standard proof-of-reserves, proof-of-transit, and hash-based commitment primitives.

From a cryptographic standpoint, TPL can be viewed as a domain-specific accountable logging and transparency system. Like generic secure logging and certificate-transparency-style constructions, it uses append-only, hash-chained commitments and non-equivocation guarantees, but it differs in two key respects. First, the state space is explicitly structured as a multi-domain balance vector and exposure flows that satisfy a conservation law. Second, the primary outputs are policy-based views for heterogeneous stakeholders, rather than raw append-only event logs. The formal treatment of conservation across domains, domain-level non-equivocation and view correctness in combination with PoR/PoT soundness, specialised to the Bitcoin treasury setting, is, to our knowledge, not captured by existing frameworks.

We prove three core integrity properties: a conservation law for Bitcoin exposures across domains (Theorem~\ref{thm:conservation}); a forward-integrity and non-repudiation guarantee based on non-equivocation of the hash-chained, anchored ledger (Theorem~\ref{thm:non-equivocation}); and correctness of policy-based views produced by honest executions of \textsf{GenView} for deterministic, stateless policies in our simple policy language (Theorem~\ref{thm:view-correctness}).
We then introduce stronger target notions of exposure soundness, policy completeness, and privacy and show that standard PoR/PoT-style primitives suffice to realise restricted versions of these goals: a restricted exposure-soundness theorem for closed sets of domains under a faithful policy (Theorem~\ref{lem:restricted-exp-soundness}) and a restricted privacy theorem for a public-investor policy (Theorem~\ref{thm:pub-privacy}).
Full exposure soundness and policy completeness for richer policies, multi-institution deployments, and more expressive leakage functions are left to future work; throughout we spell out the organisational and cryptographic assumptions under which these existence-style results hold. As in other accountable logging and transparency systems, no cryptographic mechanism can rule out full collusion between all internal and external actors; our guarantees are instead conditioned on a minimal non-collusion assumption, formalised in Section~\ref{sec:threat}, that for each domain at least one of the PoR/PoT providers, anchoring keys, or external auditors behaves honestly.

Beyond the internal governance of any single treasury, there is an emerging \emph{systemic} question for
Bitcoin markets: as regulated vehicles (exchange-traded funds, listed corporates, custodians and similar
intermediaries) accumulate BTC, can independent observers check that the aggregate, on-chain-anchored claims
about holdings never exceed Bitcoin's fixed monetary supply?
The Bitcoin protocol hard-codes a maximum supply of $21{,}000{,}000$~BTC, and as of late~2025
public market data providers report a circulating supply of around $20$~million~BTC
(roughly $95\%$ of the cap).
In such a regime, even modest levels of hidden rehypothecation or double-counting across ETFs, custodians and
corporate treasuries could in principle generate total outstanding claims that exceed the available supply.

One target \emph{motivating application} for TPL is to act as a building block for cross-institutional
supply-consistency checks.
Given a family of independently operated TPL instances for multiple treasuries, an external verifier can,
in principle, compute aggregate lower bounds on BTC under management and compare them against the
protocol-level supply limit, subject to clearly specified policies about which exposures are in-scope
for each ledger and to assumptions about coverage of participating institutions.
We do not formalise a separate security notion for such macroprudential supply guarantees,
nor do we attempt to develop a full macroprudential model.
Rather, our formalisation of exposure soundness and our requirement that all published views be
re-verifiable against Bitcoin anchors are intended to provide the cryptographic substrate on top of
which such aggregate supply checks could be designed in future work.
\subsubsection*{What is new compared to existing PoR and transparency-log systems?}

\begin{itemize}[noitemsep]
\item \emph{Multi-domain treasury state machine.} We formalise Bitcoin treasury companies as a state machine over on-chain and off-chain domains, tracking BTC-equivalent exposures rather than only on-chain balances.
\item \emph{Treasury Proof Ledger (TPL).} We define a ledger abstraction that composes proof-of-reserves / proof-of-liabilities style statements, forward-secure logging (PoT), and anchoring to Bitcoin with observer-specific policy views.
\item \emph{Security notions and reductions.} We introduce exposure-soundness, non-equivocation of anchored logs, and policy completeness for deterministic, stateless policies, and prove that restricted versions of these notions are realisable from standard cryptographic primitives together with Bitcoin's ledger assumptions.
\item \emph{Responsible transparency and macroprudential hooks.} We illustrate how TPL supports ``responsible'' disclosure policies for investors, auditors, and regulators, and enables cross-institution supply checks that are not captured by traditional PoR or accountable-logging frameworks.
\end{itemize}

\section{Background}
\label{sec:background}

\subsection{Digital asset treasuries and the transparency gap}

DAT firms offer investors leveraged exposure to crypto assets, but they also inherit Bitcoin's native auditability. Markets therefore expect a level of transparency greater than for traditional assets, especially regarding:

\begin{itemize}[noitemsep]
\item Proof of asset existence,
\item Proof of policy adherence,
\item Proof of movement justification,
\item Proof of correct multi-domain aggregation.
\end{itemize}

Standardized, cryptographic solutions remain missing.

\subsection{Limitations of existing proof-of-reserve methods}

PoR for exchanges is fundamentally incompatible with treasury operations because it requires:

\begin{itemize}[noitemsep]
\item publication of live addresses,
\item continuous doxxing of transaction flows,
\item assumption of a static set of wallets,
\item inability to redact sensitive movements,
\item mismatch with accounting standards (IFRS/GAAP).
\end{itemize}

For these reasons, leading treasury operators have refused to use PoR, and some have called for future zero-knowledge-based mechanisms instead.

\subsection{Proof-of-Transit Timestamping (PoTT)}

PoTT was originally developed as a transport-layer cryptographic receipt scheme for interplanetary Bitcoin communication. Each data transit across a relay generates a signed timestamped receipt, forming a cryptographic chain of custody~\cite{Puente2025}.

In this paper, we adapt PoTT to corporate treasury operations.
We use ``Proof-of-Transit (PoT)'' to denote the resulting simplified variant: PoT receipts are
hash-chained records adapted to treasury flows, not full transport-layer PoTT headers.

\section{Related Work}
\label{sec:related}

\paragraph{Proof-of-reserves and reserve attestations.}
Research on transparency and solvency for Bitcoin custodians has largely focused on exchange-style proof-of-reserves (PoR) and related mechanisms.
Dagher et al.\ introduced \emph{Provisions}, a privacy-preserving PoR construction that allows an exchange to prove that the sum of its customer liabilities is backed by a set of on-chain UTXOs without revealing individual balances~\cite{Dagher2015}.
Chalkias et al.\ generalised this idea with Distributed Auditing Proofs of Liabilities (DAPOL), enabling audits over partitioned customer sets and light-client friendly verification~\cite{Chalkias2018}.
More recent work from the Federal Reserve Bank of St.~Louis discusses succinct, zero-knowledge techniques for balance proofs over dynamic UTXO sets and light-client constructions~\cite{OzmisZkBitcoin2025}.
These systems treat PoR as a cryptographic primitive for a \emph{single} institution and typically reason about a snapshot relation between (assets, liabilities) and a concrete on-chain UTXO set.
They do not attempt to model a multi-domain treasury state, to formulate a conservation law across on-chain holdings, derivatives, encumbrances and fee sinks, or to provide policy-aware views over an explicitly specified state machine.
TPL instead takes such PoR- and PoT-style primitives as building blocks and asks what can be guaranteed at the level of a Bitcoin treasury \emph{ledger}, rather than at the level of a single proof.

\begin{table}[H]
\centering
\scriptsize
\begin{tabularx}{\linewidth}{p{0.19\linewidth}XXXX}
\toprule
 & Standard PoR audits & Provisions / DAPOL-style PoR & CT / accountable logs & TPL (this work) \\
\midrule
Primary goal &
Demonstrate snapshot solvency (assets versus liabilities) for a custodial balance sheet &
Prove reserve sufficiency for a given set of liabilities using privacy-preserving set commitments &
Detect mis-issuance or misbehaviour for a fixed object class (e.g.\ certificates) via append-only logs &
Model and expose a full multi-domain Bitcoin treasury state over time \\
Multi-domain treasury state &
Not modelled beyond the PoR snapshot and its liabilities set &
Not modelled beyond the PoR snapshot and its liabilities set &
Not modelled; logs are per-object rather than per-treasury &
Explicit multi-domain state machine with algebraic conservation across domains \\
Policy-defined stakeholder views &
Typically a single public statement or pass/fail audit report &
Typically a single public or auditor-facing PoR statement &
Fixed query interface (e.g.\ inclusion proofs, gossip-based consistency) &
Explicit policy language defining observer classes and their views over the ledger \\
Supply / exposure consistency wrt cap &
Not connected to global monetary cap or cross-institution exposures &
Not connected to global monetary cap or cross-institution consistency &
Not connected to monetary supply or cross-institution exposures &
Designed to enable cross-domain and cross-institution exposure checks under assumptions \\
Formal deployment-level notions (exposure soundness, policy completeness) &
Not defined at the level of a multi-domain treasury &
Not defined at the level of a multi-domain treasury &
Not defined for monetary-exposure reporting &
Defined and partially realised for treasury-level exposure and policy views \\
\bottomrule
\end{tabularx}
\caption{High-level comparison of TPL with prior lines of work on PoR and transparency logs.}
\label{tab:related-comparison}
\end{table}

\paragraph{Accounting and assurance with PoR.}
In parallel, the accounting and assurance literature has begun to treat PoR as one component of broader governance frameworks.
Lazirko et al.\ propose a ``double-helix'' view in which cryptographic attestations interleave with traditional accounting controls to cover on-chain and off-chain aspects of an entity's balance sheet~\cite{Lazirko2025}.
Vidal-Tom\'as analyses centralised exchange disclosures and argues that PoR should be understood as one tool among several in a wider ecosystem of ``guardians of trust'' and regulatory oversight~\cite{VidalTomas2024}.
Industry-oriented essays, such as Nic Carter's work on PoR standards, similarly emphasise governance, disclosure policies and the limitations of purely cryptographic attestations~\cite{CarterPoR}.
These works articulate important governance desiderata, but they generally do not provide a formal, cryptographically grounded model of treasury state, conservation and policy-based leakage.
TPL can be seen as complementing this line of work: we formalise the underlying state space and security objectives (conservation, exposure soundness, policy completeness, privacy) and show how existing PoR/PoT mechanisms can be composed, under explicit assumptions, to realise a subset of these objectives.

Recent investigative and practitioner accounts emphasise that large Bitcoin positions held via structured vehicles and third-party treasury companies strain existing audit and regulatory frameworks.
A 2025 \emph{Financial Times} investigation characterises Bitcoin treasuries as a ``nightmare'' for auditors and supervisors, in part because encumbrances and cross-entity exposures are difficult to observe and reason about at scale~\cite{FTTreasureNightmare}.
Complementary practitioner work argues that proof-of-reserves practices for exchanges and stablecoin issuers remain heterogeneous and often underspecified from a regulatory point of view~\cite{BitcoinReserveAct,TaxbitStateReserves,Ward2025,Nau2025}.
TPL targets precisely this gap: it does not propose yet another low-level PoR construction, but a treasury-level accountability layer that makes explicit what state is being logged, what views are exposed to which stakeholders and what can be said, under clear assumptions, about the resulting exposure claims.

\paragraph{Bitcoin treasuries and balance-sheet modelling.}
For corporate Bitcoin treasuries specifically, Meister models the interaction between treasury Bitcoin holdings, capital structure and market dynamics, and shows how large positions can create fragility and option-like payoffs~\cite{Meister2025}.
Wen studies survival conditions for digital asset treasury companies under funding and liquidity constraints and the risks of forced liquidation~\cite{Wen2025}.
These works treat the treasury balance sheet and market dynamics in considerable economic detail but largely assume that the disclosure layer is opaque or ad hoc.
TPL is orthogonal: we do not model market dynamics or optimal policies, but we provide a concrete, cryptographically grounded framework for what it means for a Bitcoin treasury to maintain a conserved multi-domain exposure state and to make supply-consistent, policy-specific claims about that state over time.

\paragraph{Secure logging and transparency logs.}
TPL is also related to work on secure logging and transparency logs.
Forward-secure audit logs aim to ensure that once an event has been recorded it cannot later be altered or deleted, even if the logging server is compromised, by chaining authenticated records under evolving keys~\cite{SchneierKelsey1999}.
Transparency log systems such as Certificate Transparency publish append-only, publicly auditable logs of security-critical events anchored to a widely observed substrate~\cite{Laurie2014CT}.
TPL adopts a similar anchored, append-only logging discipline, but the state recorded by the log is not an arbitrary set of events: it is the evolving vector of domain exposures and associated evidence summaries, constrained by a conservation law and accompanied by a policy language that defines the views exposed to different stakeholder classes.
To our knowledge, prior secure-logging work does not define notions analogous to exposure soundness or policy completeness for multi-domain monetary exposures, nor does it consider cross-institutional supply-consistency checks against a fixed monetary cap such as Bitcoin's 21\,million-unit limit.

\paragraph{Cryptographic compliance frameworks for traditional finance.}
A parallel line of work studies how zero-knowledge proofs and related primitives can be used to express regulatory constraints such as AML/KYC rules, capital requirements and balance-sheet privacy for conventional banks and securities intermediaries. Examples include proposal-level frameworks for regulatory reporting with embedded ZKPs, privacy-preserving payment and account systems, and generic compliance languages that compile constraints into cryptographic checks~\cite{Decker2025,FedZK2023,ZKSurvey2024}. These systems typically treat the regulated institution as a generic balance sheet with fiat-denominated assets and liabilities, and focus on expressing regulatory predicates over customer data or transaction flows. In contrast, TPL is specialised to Bitcoin-denominated treasuries: it fixes an explicit multi-domain treasury state machine, models BTC-equivalent exposures and their conservation over time, and provides policy-based, stakeholder-specific views that can be re-verified against Bitcoin anchors.

Taken together, prior work on PoR and solvency protocols, on secure logging and transparency logs, and on cryptographic compliance frameworks for traditional finance does not directly model a Bitcoin treasury as a multi-domain state machine with policy-defined stakeholder views and deployment-level exposure notions. TPL fills precisely this gap by lifting PoR- and PoT-style primitives to the level of a treasury ledger, defining exposure soundness, policy completeness and related properties, and showing how they can be partially realised under clearly stated assumptions.

\section{System Model}

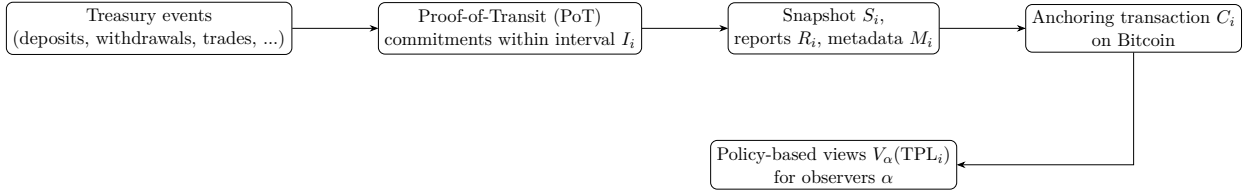
\begin{figure}[t]
  \centering
  \resizebox{\linewidth}{!}{%
\begin{tikzpicture}[>=Stealth, node distance=1.8cm,
                      every node/.style={draw, rounded corners, align=center,
                                        minimum width=3.6cm, minimum height=1cm}]
    \node (events) {Treasury events\\(deposits, withdrawals, trades, ...)};
    \node[right=of events] (pot) {Proof-of-Transit (PoT)\\commitments within interval $I_i$};
    \node[right=of pot] (snap) {Snapshot $S_i$,\\reports $R_i$, metadata $M_i$};
    \node[right=of snap] (anchor) {Anchoring transaction $C_i$\\on Bitcoin};
    \node[below=1.8cm of snap, minimum width=4.8cm] (view) {Policy-based views $V_\alpha(\mathrm{TPL}_i)$\\for observers $\alpha$};

    \draw[->] (events) -- (pot);
    \draw[->] (pot) -- (snap);
    \draw[->] (snap) -- (anchor);
    \draw[->] (anchor.south) |- (view.east);
  \end{tikzpicture}
}
  \caption{High-level lifecycle of the Treasury Proof Ledger (TPL) within a single interval $I_i$.
  Events are recorded by the treasury, committed via PoT, summarised as snapshots and reports,
  and finally anchored to Bitcoin. Policy-based observers derive views $V_\alpha(\mathrm{TPL}_i)$
  from the anchored ledger.}
  \label{fig:tpl-lifecycle}
\end{figure}
\label{sec:model}

\subsection{Treasury domains}

We model a Bitcoin treasury as a set of logical \emph{domains}:

\begin{itemize}[noitemsep]
\item Non-custodial wallet clusters,
\item Institutional custodians,
\item Exchanges and prime brokers,
\item Second-layer nodes (Lightning, rollups),
\item Derivative and collateral positions,
\item Future off-planet domains (Earth--Moon--Mars).
\end{itemize}

Each domain $d$ has a Bitcoin-equivalent exposure $B_d(t)$ at time $t$.
We treat $B_d(t)$ as an algebraic (possibly negative) quantity: $B_d(t) \in \mathbb{Z}$ encodes the net Bitcoin-equivalent
position of domain~$d$ at logical time~$t$, allowing for short or encumbered positions when $B_d(t) < 0$.

Throughout, exposures are normalised to Bitcoin units. In the cryptographic model we treat $B_d(t)$ as tracking Bitcoin-denominated claims and obligations, independent of fiat price movements or mark-to-market accounting. How more complex instruments (for example, derivatives or structured products) are mapped into Bitcoin-equivalent exposures is part of the governance and policy layer: it is encoded in the choice of domains and in policy metadata, not in the TPL primitive itself.

For the purposes of analysis we consider a discrete, totally ordered sequence of event times
$0 = t_0 < t_1 < t_2 < \cdots$, where each $t_k$ is the timestamp of some treasury event
$e_k = (t_k,d_{\mathrm{src}},d_{\mathrm{dst}},v,\mathsf{evid},m)$.
Let $B_d(t_k^-)$ and $B_d(t_k^+)$ denote the exposure of domain~$d$ immediately before and after processing~$e_k$.
Ignoring the fee domain $d_{\mathrm{fee}}$ for the moment, the state evolution of exposures is given by
\begin{equation}
\label{eq:balance-update}
  B_d(t_k^+) =
  \begin{cases}
    B_d(t_k^-) - v, & \text{if } d = d_{\mathrm{src}},\\
    B_d(t_k^-) + v, & \text{if } d = d_{\mathrm{dst}},\\
    B_d(t_k^-),     & \text{otherwise.}
  \end{cases}
\end{equation}
Transfers involving the dedicated fee domain $d_{\mathrm{fee}}$ or non-modelled external counterparties are treated as
net outflows or inflows to the modelled domain set and are handled explicitly when we reason about conservation below.

\paragraph{Worked example.}
As a concrete illustration, consider three domains $d_{\mathrm{cust}}$, $d_{\mathrm{exch}}$ and $d_{\mathrm{fee}}$.
Suppose that at time $t_0$ the exposure vector satisfies
$B_{d_{\mathrm{cust}}}(t_0)=100$, $B_{d_{\mathrm{exch}}}(t_0)=0$ and $B_{d_{\mathrm{fee}}}(t_0)=0$, with all other domains having zero exposure.
A treasury event
$e_1=(t_1,d_{\mathrm{cust}},d_{\mathrm{exch}},10,\mathsf{evid}_1,m_1)$
implements an internal transfer of $10$~BTC from custody to an exchange.
Applying~\eqref{eq:balance-update} yields
$B_{d_{\mathrm{cust}}}(t_1^+)=90$ and $B_{d_{\mathrm{exch}}}(t_1^+)=10$, with all other $B_d(t_1^+)$ unchanged and the sum across $\mathcal{D} \setminus \{d_{\mathrm{fee}}\}$ preserved.
If a subsequent event
\[
e_2=(t_2, d_{\mathrm{exch}}, d_{\mathrm{fee}}, 0.00005, \mathsf{evid}_2, m_2)
\]
records a network-fee payment of $0.00005$ BTC, then
$B_{d_{\mathrm{exch}}}(t_2^+) = 9.99995$,
$B_{d_{\mathrm{fee}}}(t_2^+) = 0.00005$,
and the total over $\mathcal{D} \setminus \{d_{\mathrm{fee}}\}$ decreases by $0.00005$, reflecting the outflow to the fee sink.
This small example illustrates both the local update rule and its consistency with the conservation statement in Theorem~\ref{thm:conservation}.

\subsection{Treasury events}

\begin{definition}[Treasury Event]
An event is a tuple
\[
e = (t, d_{\mathrm{src}}, d_{\mathrm{dst}}, v, \mathsf{evid}, m)
\]
where:
\begin{itemize}[noitemsep]
\item $t$: timestamp,
\item $d_{\mathrm{src}}$, $d_{\mathrm{dst}}$: source and destination domains,
\item $v$: BTC-denominated value,
\item $\mathsf{evid}$: primitive evidence (for example, on-chain transaction identifiers, custody statements, or contractual references),
\item $m$: policy metadata.
\end{itemize}
From each event we derive a digest
\[
h := H(\mathsf{evid} \,\Vert\, v \,\Vert\, m),
\]
which is the value incorporated into the PoT receipt chain.
Note that the timestamp and domain identifiers are not included in this digest; they are instead authenticated via the PoT receipts and any digital signatures over the full event tuple.
This separation is intentional and does not affect our non-equivocation arguments, which always reason about the signed event records rather than the digest in isolation.
\end{definition}

Primitive evidence may include:

\begin{itemize}[noitemsep]
\item raw Bitcoin transactions,
\item custodian-signed statements,
\item exchange API proofs,
\item Lightning commitment states,
\item collateral position attestations.
\end{itemize}

\section{Threat Model and Requirements}
\label{sec:threat}

We consider an adversary that controls the protocol scheduler and may corrupt some of the operational parties. For clarity, we summarise the main assumptions here before formalising them below.

\paragraph{Adversary and trust assumptions.}
\begin{itemize}[noitemsep]
\item The adversary may corrupt the treasury operator and any subset of proof-of-reserves providers, proof-of-transit / logging operators, and off-chain service providers, and may adaptively schedule their actions subject to the liveness bounds below.
\item Bitcoin is modelled as an append-only ledger with $k$-confirmation finality; reorgs deeper than $k$ are treated as failures of the underlying platform assumption and are out of scope.
\item At least one party controlling the keys used to anchor TPL commitments to Bitcoin behaves honestly (or is periodically checked by an independent auditor); otherwise an adversary can log an arbitrary but non-equivocating but economically meaningless history.
\item All cryptographic primitives (hash functions, digital signatures, and any zero-knowledge proofs used inside PoR / PoT mechanisms) satisfy standard collision-resistance, EUF-CMA, soundness, and zero-knowledge properties.
\end{itemize}

\paragraph{Liveness assumptions.}
We assume that honest treasurers and service providers eventually process and anchor all
submitted events.
Concretely, there exist bounds $\Delta_{\mathsf{event}}$, $\Delta_{\mathsf{snap}}$, and
$\Delta_{\mathsf{anchor}}$ such that, in the absence of network partition,
(i) any valid treasury event is incorporated into the internal log within $\Delta_{\mathsf{event}}$,
(ii) snapshot triggers produce PoR snapshots at least once every $\Delta_{\mathsf{snap}}$, and
(iii) TPL commitments are anchored to Bitcoin at least once every $\Delta_{\mathsf{anchor}}$.
Our security experiments allow the adversary to schedule inputs subject to these liveness bounds.
Formally, we call any such schedule \emph{admissible}, and we only define our security experiments on admissible
schedules; adversarial strategies that violate the liveness bounds fall outside the modelled setting.
Conceptually, such violations correspond to the treasury or its service providers ceasing to operate TPL at all;
they break liveness rather than introducing new safety attacks within our model.

The Treasury Proof Ledger is intended to operate in an environment where multiple stakeholders, including public-investors, internal and external auditors, regulators, and commercial counterparties, rely on corporate Bitcoin treasuries as systemic infrastructure.
We therefore make the threat model and requirements explicit.

\subsection{Model and assumptions}
\label{sec:model-assumptions}

We briefly summarise the computational model, operational environment, and assumptions that underlie the security games
and theorems in the rest of the paper.

\paragraph{Computational model and parties.}
All probabilistic algorithms are modelled as PPT (probabilistic polynomial time) in the security parameter~$\lambda$.
The TPL protocol runs alongside the Bitcoin network, which we treat as an append-only ledger with eventual
finality consistent with the liveness bounds in Section~\ref{sec:threat}.
The main parties are treasurers, proof-of-reserves service providers, proof-of-transit providers, anchoring keys,
and external observers or auditors, as defined in Sections~\ref{sec:model} and~\ref{sec:tpl}.
Our security experiments quantify over adversarial schedulers that are restricted to \emph{admissible} schedules in the sense of Section~\ref{sec:threat};
strategies that violate the liveness bounds correspond to operational failure (the system ceasing to run TPL) rather than to additional safety attacks within our model.

\paragraph{Bitcoin and anchoring assumptions.}
We assume that the underlying Bitcoin blockchain provides eventual consistency and $k$-confirmation finality in the usual sense.
Anchoring transactions are required to achieve~$k$ confirmations within the liveness bound~$\Delta_{\mathsf{anchor}}$, and adversaries are not given the ability to cause large-scale chain re-organisations beyond this bound.
These assumptions are standard in cryptographic treatments of Bitcoin-anchored protocols and capture the idea that, once sufficiently confirmed, the anchoring commitments are effectively immutable.

\paragraph{Cryptographic primitives.}
We assume collision-resistant hash functions, a proof-of-reserves (PoR) primitive that satisfies existence and ownership soundness (Definitions~\ref{def:por-existence} and~\ref{def:por-ownership}), and a proof-of-transit (PoT) primitive that satisfies receipt unforgeability in the sense of Definition~\ref{def:pot-unforgeability}.
We also assume secure digital signatures for treasurer and service-provider keys and, when needed in the privacy treatment, zero-knowledge or witness-indistinguishable argument systems for the relevant NP relations (for example, a general-purpose zk-SNARK encoding the PoR and PoT validity predicates together with the conservation law).
All reductions are stated in terms of the standard advantages of adversaries against these primitives.

\paragraph{Organisational assumptions.}
In addition to these cryptographic assumptions, we rely on the minimal non-collusion Assumption~\ref{assump:noncollusion}, which rules out degenerate scenarios in which all PoR and PoT providers, anchoring keys, and external auditors collude with a single adversarial treasury.
If every such role in a deployment were controlled by one strategic actor, the resulting TPL instance could emit internally consistent yet entirely fictitious public views, and exposure-soundness guarantees would become vacuous.
The model also abstracts away from questions of domain completeness and off-chain legal semantics, which must be
addressed by governance and regulatory processes rather than by the cryptographic layer.

\subsection{Security objectives}

We include detailed game-based definitions and protocol syntax, together with proof sketches for the main theorems.
Full reductions and the more technical lemmas are provided in Appendices~\ref{app:neq-reduction}--\ref{app:privacy-reduction} of this version.

\begin{itemize}[noitemsep]
\item \textbf{On-chain existence consistency.} Building on PoR existence soundness at the primitive level, for each snapshot height $h$, the UTXO sets reported for every on-chain domain correspond to actual unspent outputs on Bitcoin (or its commit chains), up to clearly documented reconciliation tolerances.
\item \textbf{Ownership and control soundness.} For each domain, the TPL contains evidence that the treasury operator (or an explicitly named fiduciary) can authorise spends of the reported UTXOs under ordinary network conditions.
\item \textbf{Policy completeness.} Given a fixed disclosure policy $\mathcal{P}$ (for example, ``all cold-wallet holdings and all movements above 1\% of treasury size''), every event that should be disclosed under $\mathcal{P}$ eventually appears in the committed ledger, modulo explicit latency bounds, and in particular no accepted policy-based view derived from the ledger may omit such an event.
\item \textbf{Non-equivocation.} The treasury operator cannot maintain two conflicting TPL histories that are both consistent with the Bitcoin anchoring schedule and both acceptable to honest observers, except with negligible probability.
\item \textbf{Forward integrity and auditability.} Once a TPL state has been anchored, it cannot be altered without producing objectively detectable inconsistencies in the anchoring chain or in the internal hash pointers of the log.
\item \textbf{Privacy compatibility.} Under the defined adversarial model, authorised observers should not be able to infer more about the treasury's internal positions than what is implied by the chosen policy and any auxiliary disclosures. In this paper we formalise only a restricted privacy guarantee for a public-investor policy (Theorem~\ref{thm:pub-privacy}); more general policies and leakage functions remain an open problem.
\end{itemize}
\begin{table}[h]
\centering
\begin{tabularx}{\linewidth}{lX}
\toprule
Goal & Status in this paper \\
\midrule
On-chain existence consistency & Encoded in the conservation law for multi-domain exposure vectors (Theorem~\ref{thm:conservation}) under Assumption~\ref{assump:algebraic-consistency} and assuming PoR existence soundness at the primitive level (Definition~\ref{def:por-existence}). \\
Ownership and control soundness & Provided by the proof-of-reserves interface (Definitions~\ref{def:por-existence}--\ref{def:por-ownership}); assumed as a primitive rather than re-derived. \\
Exposure soundness & Defined formally in Definition~\ref{def:exp-soundness}; we prove a restricted exposure-soundness theorem for closed sets of domains under a faithful policy (Theorem~\ref{lem:restricted-exp-soundness}); the general case is left to future work. \\
Policy completeness & Defined in Definition~\ref{def:policy-completeness}; we obtain completeness for a specific historical-disclosure policy (Lemma~\ref{lem:policy-completeness-history}), while full policy completeness for arbitrary policies is left as a target notion. \\
Non-equivocation & Realised for TPL via the hash-chained ledger and PoT receipts (Theorem~\ref{thm:non-equivocation}). \\
Forward integrity and auditability & A consequence of non-equivocation and anchoring for the append-only ledger (Theorem~\ref{thm:non-equivocation}). \\
Privacy compatibility & Captured via leakage functions and simulation-based privacy; we prove a restricted theorem for a public-investor policy and its leakage profile $L_{\mathsf{pub}}$ (Theorem~\ref{thm:pub-privacy}); richer policies and leakage functions remain open. \\
\bottomrule
\end{tabularx}
\caption{Deployment-level goals for TPL and their status in this paper.}
\label{tab:goals-status}
\end{table}

\subsection{Participants}

We consider the following participants:

\begin{itemize}[noitemsep]
\item The \emph{treasury operator}, controlling policy, signing authority, and the configuration of domains;
\item \emph{Service providers} such as custodians, exchanges, prime brokers, and second-layer operators that hold or manage some of the company's Bitcoin;
\item \emph{Internal stakeholders}, including the board, risk committee, and internal audit;
\item \emph{External stakeholders}, including external auditors, regulators, and public-market investors.
\end{itemize}

\subsection{Adversaries}

Adversaries may include:

\begin{itemize}[noitemsep]
\item A dishonest treasury operator attempting to overstate Bitcoin holdings, hide encumbrances, or misrepresent domain exposures;
\item A compromised or colluding service provider that fabricates balance statements or withholds adverse information;
\item An external attacker who seeks to infer sensitive operational wallet structures or trading strategies from published artefacts;
\item An ``honest-but-curious'' observer (including some regulators) who follows the protocol but aims to learn more than the policy allows.
\end{itemize}

\paragraph{Minimal non-collusion assumption.}
\begin{assumption}[Minimal non-collusion]\label{assump:noncollusion}
For each domain $d \in \mathcal{D}$ we assume that at least one of the following roles remains honest and uncompromised:
(i) the implementation of the PoR primitive and its signing keys for $d$;
(ii) the PoT service provider(s) responsible for issuing receipts that involve $d$;
(iii) the keys used to anchor TPL commitments to Bitcoin on behalf of $d$; or
(iv) an independent external auditor that periodically reconstructs and checks the domain's state.
Scenarios in which all of these roles collude to fabricate both primitive evidence and audit trails are considered
outside the scope of the cryptographic guarantees that we prove, although they may still be bounded by external
legal, regulatory, or operational controls.
This ``at least one honest role per domain'' condition is explicitly organisational and mirrors the standard
``at least one honest log server or monitor'' assumptions used in transparency-log and accountable-logging models.
In practice, role~(iv) is aligned with statutory independence requirements for external audit firms, while role~(iii)
can be instantiated using anchoring keys operated directly by a prudential regulator or by a regulated multi-party HSM
under regulatory oversight.
These governance arrangements are not modelled cryptographically, but they provide concrete ways to justify the
assumption in real deployments.
\end{assumption}
\begin{lemma}[Necessity of non-collusion]
\label{lem:necessity-noncollusion}
If all parties that can influence the inputs to TPL for a given deployment (the treasury operator, all proof-of-reserves and proof-of-transit providers, and all anchoring keys) are controlled by a single adversarial entity, then for any polynomial-time TPL verifier and any target exposure vector consistent with Bitcoin's supply cap there exists a PPT strategy for the adversary that makes the verifier accept while the realised economic exposures of the treasury differ arbitrarily from that target.
In particular, no exposure-soundness property stronger than checking algebraic consistency of the reported balances can be achieved in this fully colluding setting.
\end{lemma}

\begin{proof}[Proof sketch]
In the fully colluding scenario the adversary can choose arbitrary internal books and flows, fabricate consistent PoR snapshots and PoT receipts, and anchor arbitrary commitments while ensuring that all published artefacts are algebraically consistent with the target exposure vector.
Any experiment for exposure soundness that accepts honest executions must also accept this colluding execution, since the verifier only interacts with the fabricated artefacts and not with the underlying economic reality.
\end{proof}

We assume standard cryptographic hardness for hash functions and digital signatures.
In particular, let $H : \{0,1\}^\star \rightarrow \{0,1\}^\kappa$ be the hash function used to commit to ledger prefixes and auxiliary data.

\begin{definition}[Hash collision resistance]
\label{def:hash-coll}
For a PPT adversary $\mathcal{A}$ and security parameter $\lambda$, the experiment $\mathsf{Exp}^{\mathsf{coll}}_{\mathcal{A}}(\lambda)$ runs $\mathcal{A}(1^\lambda)$ to obtain a pair of strings $(x,x')$ and outputs $1$ if $x \neq x'$ but $H(x)=H(x')$.
The collision advantage of~$\mathcal{A}$ is
\[
  \mathsf{Adv}^{\mathsf{coll}}_{\mathcal{A}}(\lambda)
  = \Pr\big[\mathsf{Exp}^{\mathsf{coll}}_{\mathcal{A}}(\lambda)=1\big].
\]
We say that $H$ is collision-resistant if $\mathsf{Adv}^{\mathsf{coll}}_{\mathcal{A}}(\lambda)$ is negligible in $\lambda$ for all PPT adversaries~$\mathcal{A}$.
\end{definition}

\subsection{Security and governance goals}

TPL is designed to satisfy the following goals:

\begin{itemize}[noitemsep]
\item \textbf{Soundness of exposures:} it should be computationally infeasible for a treasury to report total Bitcoin exposure above reality without breaking cryptographic primitives or violating the minimal non-collusion assumption of Section~\ref{sec:threat};
\item \textbf{Completeness of flows:} all material inter-domain movements of Bitcoin must eventually be explainable as sequences of PoT receipts consistent with observed PoR snapshots;
\item \textbf{Policy-correct views:} for any observer class $\alpha$, the published view $\mathrm{View}_\alpha$ must be a deterministic function of the committed ledger and the declared policy $V_\alpha$, with no manual editing of underlying data;
\item \textbf{Operational privacy:} public and semi-public views should reveal only the information about exposures that is authorised by the corresponding policy and leakage function. In this work we only prove a restricted privacy theorem for a specific public-investor policy (Theorem~\ref{thm:pub-privacy}); a full treatment for richer policies and leakage functions is left to future work;
\item \textbf{Accountability and non-repudiation:} once an event has been recorded and anchored in TPL, neither the treasury nor a service provider can plausibly deny its occurrence without contradicting previously published commitments.
\end{itemize}

We formalise exposure soundness and policy completeness as \emph{target} experiments
(Definitions~\ref{def:exp-soundness} and~\ref{def:policy-completeness}): they capture the level of robustness
we would like real deployments to achieve in the long run, rather than the exact guarantees we are able to prove today.
In this paper we prove a \emph{restricted} exposure-soundness theorem (Theorem~\ref{lem:restricted-exp-soundness}) for
closed sets of domains under a faithful exposure policy, together with a simple policy-completeness lemma for a
history-revealing policy (Lemma~\ref{lem:policy-completeness-history}).
A full analysis of exposure soundness and policy completeness for expressive policy languages and open domain sets
is left to future work.

The remainder of the paper turns these informal goals into explicit game-based security notions.
We prove three core integrity properties for an idealised TPL construction: conservation across domains; non-equivocation (and hence forward integrity and non-repudiation) of logged events; and correctness of policy-based views for deterministic, stateless policies. We also give a restricted exposure-soundness lemma for closed sets of domains and a restricted privacy theorem for a public-investor policy.

These goals complement, rather than replace, traditional financial-statement assertions such as existence, completeness, rights and obligations, and presentation under IFRS or US~GAAP.

\subsection{Assumptions and scope}

TPL does not attempt to address all risks faced by digital-asset treasuries.
To make the guarantees of our construction precise, we make the following assumptions explicit.
\paragraph{Summary.}
All guarantees in this paper are proved \emph{only} under the following conditions, which we make explicit upfront:
\begin{itemize}[noitemsep]
\item \textbf{Cryptographic assumptions:} an existence- and ownership-sound proof-of-reserves (PoR) primitive, an unforgeable proof-of-transit (PoT) construction, and a collision- and preimage-resistant hash function with domain separation across its uses.
\item \textbf{System assumptions:} an anchoring substrate that behaves as an append-only log with eventual finality after $k$ confirmations (Assumption~\ref{assump:anchoring-substrate}).
\item \textbf{Interface assumptions:} an algebraically consistent PoR interface that reports exposures consistent with the balance-update rule~\eqref{eq:balance-update}, including explicit flows to and from the fee domain and modelled counterparties (Assumption~\ref{assump:algebraic-consistency}).
\item \textbf{Modelling assumptions:} domain completeness for the set of domains $\mathcal{D}$ (Assumption~\ref{assump:domain-completeness}) so that all economically relevant positions are assigned to some domain, together with correctly configured reporting schedules and policy predicates.
\item \textbf{Trust assumptions:} the minimal non-collusion assumption of Section~\ref{sec:threat}, applied per domain, which rules out full collusion between all PoR/PoT providers, anchoring keys, and auditors for that domain.
\end{itemize}

\begin{assumption}[Domain completeness]
\label{assump:domain-completeness}
For every economically relevant Bitcoin-denominated position held by the treasury at logical time $t$, there exists a domain $d \in \mathcal{D}$ such that this position is reflected in $B_d(t)$ according to the balance-update rule~\eqref{eq:balance-update}.
\end{assumption}

\paragraph{Cryptographic assumptions.}
\emph{Proof-of-reserves soundness.}
We treat the underlying proof-of-reserves mechanism as an external primitive that
outputs snapshots and inclusion proofs with standard existence and ownership soundness
properties (for example, as in prior work on provable reserves~\cite{Dagher2015,Chalkias2018}).
In particular, a dishonest party should not be able to convince verifiers of non-existent
Bitcoin holdings except with negligible probability.

\emph{Hash-function assumptions.}
We assume that the hash function $H$ used to commit to ledger states and PoT receipts
is collision resistant and preimage resistant, and that its different uses
(for example, ledger commitments, PoT receipts, and any PoR-related commitments)
are domain-separated, for instance via distinct prefix tags, so that re-use of
the same primitive does not introduce cross-protocol collisions.

\paragraph{System assumptions.}
We treat the anchoring substrate (e.g., Bitcoin L1) as providing eventual finality after $k$ confirmations:
once an anchor transaction is $k$ blocks deep it is economically infeasible to revert.

\begin{assumption}[Anchoring substrate]\label{assump:anchoring-substrate}
The anchoring substrate provides eventual finality after $k$ confirmations in the sense that once an anchor transaction is
$k$ blocks deep it is economically infeasible to revert, and behaves as an append-only log for the TPL protocol.
For the concrete case in which Bitcoin itself is used as the anchoring substrate, this corresponds to the standard
assumption that an adversary controls strictly less than half of the network's hash power, so that the probability of
reverting a transaction with $k$ confirmations is negligible as a function of~$k$.
Denial-of-service attacks and long-range reorganisations beyond this abstraction are outside our model.
\end{assumption}
We do not model denial-of-service attacks or long-range consensus failures, and we abstract away transaction-fee dynamics except where explicitly noted.
The TPL itself is an append-only log with a well-defined total order on events, maintained by the treasury operator or an appointed service provider.

\paragraph{Proof-of-reserves (PoR) interface.}
We treat the underlying proof-of-reserves primitive as a black-box interface that exposes the current reserves state as a multiset of abstract Bitcoin-denominated exposure items, which may be positive or negative (for example, long and short positions).
Concrete instantiations may range from on-chain UTXOs to off-chain channel states or collateral positions, provided they admit
such a proof-of-reserves interface; fully modelling the semantics of complex off-chain domains is left to future work. For domains with net short or collateralised positions we model exposure items at this interface level as signed margin or collateral records; we treat these as abstract claims about Bitcoin-denominated exposures and do not attempt to capture their full legal or economic semantics.

\begin{assumption}[Algebraic consistency]\label{assump:algebraic-consistency}
For each domain $d \in \mathcal{D}$ and logical time $t$, let $B_d(t)$ denote the algebraic Bitcoin-equivalent exposure
defined by the balance-update rule~\eqref{eq:balance-update}, including explicit flows to and from the fee domain
$d_{\mathrm{fee}}$ and any modelled external counterparties.
We assume that whenever the PoR primitive is invoked for domain $d$ at a time aligned with a TPL snapshot or event,
the exposure reported by the PoR interface for~$d$ coincides with $B_d(t)$ up to this explicit handling of fees and
external flows.
In particular, PoR snapshots and balance proofs do not create or destroy exposure beyond what is accounted for in the
treasury-event sequence and the fee domain.
\end{assumption}

This assumption should be read as a governance and interface requirement rather than as a standalone cryptographic
guarantee: the treasury must configure its proof-of-reserves interface so that the algebraic exposure space used by TPL
reflects the economic notion of exposure it intends to attest to (for example, whether particular derivative positions,
haircuts, or segregated client assets are in or out of scope).

At any time the treasury can run a snapshot procedure over a multiset of unspent outputs $S_{\mathsf{snap}}$ to obtain a public commitment $C_{\mathsf{snap}}$ and auxiliary data that later support membership and balance proofs.
Given a commitment $C_{\mathsf{snap}}$, a description of an exposure claim $(x,a)$ (for example, an address or script and a value $a$), and a proof $\pi$, a deterministic verification algorithm
\[
  \mathsf{PoR.Verify}(C_{\mathsf{snap}}, x, a, \pi) \in \{0,1\}
\]
returns $1$ if and only if the claim is accepted.
We assume two soundness properties: \emph{existence soundness}, meaning that any accepted proof corresponds to a real unspent output of value at least $a$ in $S_{\mathsf{snap}}$, and \emph{ownership soundness}, meaning that such an output is spendable under keys controlled by the treasury.

\begin{definition}[PoR existence soundness]
\label{def:por-existence}
Let $\mathsf{Exp}^{\mathsf{por-exist}}_{\mathcal{A}}(\lambda)$ be the following experiment.
The challenger fixes a multiset $S$ of treasury-controlled Bitcoin positions at some snapshot time, where each $(x,a)\in S$ represents a ledger position $x$ (for example, a specific UTXO or account identifier) with associated Bitcoin value $a$, and runs
the PoR snapshot procedure to obtain a commitment $C_{\mathsf{snap}}$ for $S$.
It gives $1^\lambda$ and $C_{\mathsf{snap}}$ to a PPT adversary~$\mathcal{A}$, which may make arbitrary
(polynomially many) auxiliary queries to the challenger about $S$ that do not reveal secret keys.
Eventually $\mathcal{A}$ outputs a tuple $(x^\star,a^\star,\pi^\star)$.
The experiment outputs $1$ if and only if
$\mathsf{PoR.Verify}(C_{\mathsf{snap}},x^\star,a^\star,\pi^\star)=1$ but $(x^\star,a^\star)\not\in S$.
The \emph{existence-soundness advantage} of~$\mathcal{A}$ is
\[
  \mathsf{Adv}^{\mathsf{por-exist}}_{\mathcal{A}}(\lambda)
  = \Pr\big[\mathsf{Exp}^{\mathsf{por-exist}}_{\mathcal{A}}(\lambda)=1\big].
\]
We say that the PoR primitive satisfies existence soundness if $\mathsf{Adv}^{\mathsf{por-exist}}_{\mathcal{A}}(\lambda)$
is negligible in~$\lambda$ for all PPT adversaries~$\mathcal{A}$.
\end{definition}

\begin{definition}[PoR ownership soundness]
\label{def:por-ownership}
Let $\mathsf{Exp}^{\mathsf{por-own}}_{\mathcal{A}}(\lambda)$ be defined as above except that the experiment first fixes
a set of treasury-held secret keys $K_{\mathrm{treas}}$ and a multiset $S$ of Bitcoin positions spendable under keys in $K_{\mathrm{treas}}$ at snapshot time.
The challenger runs the snapshot procedure to obtain $C_{\mathsf{snap}}$ for $S$ and gives $1^\lambda$ and $C_{\mathsf{snap}}$
to~$\mathcal{A}$, which outputs $(x^\star,a^\star,\pi^\star)$.
The experiment outputs $1$ if and only if
$\mathsf{PoR.Verify}(C_{\mathsf{snap}},x^\star,a^\star,\pi^\star)=1$ but the underlying coin $x^\star$ is not spendable
under any key in $K_{\mathrm{treas}}$.
The \emph{ownership-soundness advantage} of~$\mathcal{A}$ is
\[
  \mathsf{Adv}^{\mathsf{por-own}}_{\mathcal{A}}(\lambda)
  = \Pr\big[\mathsf{Exp}^{\mathsf{por-own}}_{\mathcal{A}}(\lambda)=1\big].
\]
We say that the PoR primitive satisfies ownership soundness if $\mathsf{Adv}^{\mathsf{por-own}}_{\mathcal{A}}(\lambda)$
is negligible in~$\lambda$ for all PPT adversaries~$\mathcal{A}$.
\end{definition}

Existing constructions in the Provisions/DAPOL family instantiate this interface for Bitcoin UTXOs and similar account-based ledgers.

\paragraph{Example instantiations.}
For concreteness, and to demonstrate that our abstract interfaces are non-vacuous, TPL can be instantiated
using existing proof-of-reserves and anchoring mechanisms that satisfy Definitions~\ref{def:por-existence}--\ref{def:por-ownership}
and~\ref{def:pot-unforgeability}.
One natural choice is to use a Provisions-style Merkle-based PoR scheme~\cite{Dagher2015} for snapshots: for each reporting time~$t_i$ the vector of domain exposures $(B_d(t_i))_{d\in\mathcal{D}}$ is encoded as leaves in a Merkle tree, the PoR snapshot set $\mathcal{S}_i$ contains the resulting root and auxiliary data, and inclusion proofs returned by $\mathsf{GenView}$ are Merkle paths checked by $\mathsf{VerifyView}$.
In parallel, PoT receipts can be instantiated by adapting the PoTT construction summarised in Section~\ref{sec:background}: the primitive evidence $\mathsf{evid}_i$ consists of one or more Bitcoin transaction identifiers moving value between controlled UTXOs, and the per-event record $\mathsf{rec}_i$ includes a hash-chain value $R_i$ together with digital signatures by the treasury and any external service provider on $(h_i,d_{\mathrm{src}},d_{\mathrm{dst}},t_i,R_i)$.
Finally, the ledger commitment $C_i = H(\mathrm{TPL}_i)$ can be embedded into Bitcoin either via an \texttt{OP\_RETURN} output or as data committed in a Taproot script path, so that anyone with access to the blockchain can retrieve and verify the sequence of anchors.
Our formal treatment remains independent of the specific choice of PoR and PoT primitives: any instantiation that satisfies the existence-, ownership-, unforgeability-, and collision-resistance assumptions stated above can be plugged into the TPL interface.
\paragraph{Privacy assumptions for restricted results.}
For the restricted privacy theorem for the public-investor policy (Theorem~\ref{thm:pub-privacy}), we assume that the underlying PoR and PoT primitives admit zero-knowledge or witness-indistinguishable proofs of correctness for the public values that appear in TPL (for example, that a PoR commitment really binds to an internal ledger state and that each PoT receipt corresponds to a batch of Bitcoin transactions moving a specified amount of value). We also assume that any digital signatures used in TPL are standard existentially unforgeable schemes whose outputs can be generated by a simulator given the same signing keys as in the real execution (or wrapped inside the same zero-knowledge proof system), so that signatures reveal no auxiliary information beyond what is explicit in their messages and public keys.
Intuitively, we require that the public transcripts derived from these primitives leak no information beyond the values that are already explicit in the corresponding policy-based views and anchoring transactions.
In practice this can be realised, for example, by bundling signatures and PoR/PoT statements inside zero-knowledge proofs of knowledge, or by using schemes whose outputs are statistically close to uniform conditioned on the public keys.
These stronger assumptions ensure that the artefacts exposed through $\mathsf{GenView}$ and the on-chain commitments can be simulated given only the leakage function $L_{\mathsf{pub}}$ and the anchoring transactions, rather than relying purely on existence-and-ownership soundness and collision resistance.

These privacy assumptions are stated at the level of abstract primitives; we deliberately do not commit to a specific instantiation and treat Theorem~\ref{thm:pub-privacy} as an existence-style result that shows such leakage profiles are achievable in principle.

\paragraph{Organisational and governance assumptions.}
TPL is designed to support, not replace, traditional governance and audit processes.
We assume that internal approval workflows (e.g., board resolutions, risk-committee sign-off) are correctly followed whenever they are invoked, and that at least one auditor or oversight body will check consistency between PoR attestations, PoT receipts, and the committed TPL state.
Cases where all internal and external actors collude to fabricate PoR evidence remain outside the scope of purely cryptographic guarantees; exposure soundness is therefore only intended to hold under a minimal non-collusion assumption, namely that for each domain at least one PoR/PoT service provider behaves honestly or is periodically checked by an independent auditor.

\paragraph{Domain completeness and anchoring control.}
TPL models the treasury as a collection of abstract domains $\mathcal{D}$.
We assume that all Bitcoin-denominated positions that management wishes to treat
as part of the treasury are assigned to some modelled domain $d \in \mathcal{D}$.
TPL cannot, by itself, detect assets that are omitted from the domain decomposition;
assessing the completeness of domain coverage remains a traditional audit and
governance responsibility.
We further assume that at least one party controlling the keys used to embed
anchors on Bitcoin behaves honestly or is periodically checked by an independent
auditor.
If a single malicious party controlled all anchoring keys and all internal records,
TPL would still prevent that party from later equivocating about its own committed
history, but it could not stop them from logging a systematically false history
from the outset.

\paragraph{Scope and non-goals.}
We explicitly \emph{do not} address the correctness of valuation policies (for example, how to translate Bitcoin-denominated positions into fiat reporting currencies), nor do we attempt to model market-liquidity risk or price impact.
Similarly, we assume that off-chain legal contracts governing custodial relationships, derivatives, or structured products are honoured by counterparties; TPL's role is to capture cryptographic evidence consistent with those contracts and expose them through policy-aware views.
Finally, we do not attempt to provide differential-privacy-style guarantees; our privacy goal is operational: to ensure that only those events and positions authorised by a given policy become visible to that policy's observers. Intuitively, observers should learn nothing beyond the policy-defined view, its associated Bitcoin anchors, and what is already implied by the public Bitcoin ledger; formalising this via an explicit leakage function and simulation-based privacy definition is left to future work.
\section{The Treasury Proof Ledger}
\label{sec:tpl}

From this point onwards we work with an \emph{idealised} TPL construction that sits on top of the
system and cryptographic assumptions of Section~\ref{sec:threat}.
In particular, we assume that the proof-of-reserves interface satisfies the existence- and ownership-
soundness properties of Definitions~\ref{def:por-existence} and~\ref{def:por-ownership}, that the
proof-of-transit construction satisfies receipt unforgeability (Definition~\ref{def:pot-unforgeability}),
that the Bitcoin anchoring substrate provides eventual finality and the liveness bounds from
Section~\ref{sec:threat}, and that the minimal non-collusion assumption of Section~\ref{sec:threat} holds.
All security theorems in this and subsequent sections are stated with respect to this idealised construction.

\subsection{Ledger structure}

The TPL is a cryptographically committed ledger consisting of:

\begin{enumerate}[noitemsep]
\item \textbf{PoR Snapshots}  
Merkle-committed snapshots of all domain exposures $B_d(t)$.
\item \textbf{PoT Records}  
Receipt chains adapted from PoTT, documenting treasury events.
\item \textbf{Policy Metadata}  
Tags indicating regulatory category, materiality, risk, sensitivity.
\end{enumerate}
The minimal non-collusion assumption of Section~\ref{sec:threat} (Assumption~\ref{assump:noncollusion}) is intentionally limited: it rules out the degenerate case in which all PoR and PoT providers, anchoring keys, and auditors collude to hide missing exposures.

The ledger forms an append-only structure with regular commitment to a public anchoring medium on the Bitcoin blockchain (e.g., via OP\_RETURN commitments or Taproot-embedded commitments).

\subsection{Proof-of-Transit for treasury flows}

Each treasury event
\[
e_i = (t_i, d_{\mathrm{src}}, d_{\mathrm{dst}}, v_i, \mathsf{evid}_i, m_i)
\]
creates a PoT receipt $R_i$.
From the primitive evidence and policy metadata we derive a digest
\[
h_i := H(\mathsf{evid}_i \,\Vert\, v_i \,\Vert\, m_i)
\]
and then define
\[
R_i := H(R_{i-1} \,\Vert\, h_i \,\Vert\, d_{\mathrm{src}} \,\Vert\, d_{\mathrm{dst}} \,\Vert\, t_i).
\]
We initialise the chain with a fixed public seed $R_0 := 0^\lambda$.

For each treasury event $e_i$ we write the full Proof-of-Transit record as
\[
  \mathsf{rec}_i := (t_i, d_{\mathrm{src}}, d_{\mathrm{dst}}, v_i,
                     \mathsf{evid}_i, m_i, h_i, R_i,
                     \sigma^{\mathsf{treas}}_i, \sigma^{\mathsf{prov}}_i),
\]
where $\sigma^{\mathsf{treas}}_i$ and $\sigma^{\mathsf{prov}}_i$ are digital signatures
under the treasury and service-provider keys on the authenticated tuple
$(h_i, d_{\mathrm{src}}, d_{\mathrm{dst}}, t_i, R_i)$.
These receipts and their signatures create a verifiable chain of custody without revealing underlying wallet addresses.

\subsection{State evolution}
\label{subsec:state-evolution}
Let $\mathrm{TPL}_0$ denote the empty ledger at the time a company adopts the scheme.
At each reporting interval $i$, the ledger evolves as
\[
  \mathrm{TPL}_i = \mathrm{TPL}_{i-1} \,\Vert\, (\mathcal{S}_i, \mathcal{R}_i, M_i),
\]
where $\mathcal{S}_i$ is the set of PoR snapshots, $\mathcal{R}_i$ is the set of PoT receipts, and $M_i$ is the multiset
of treasury events appended in the reporting interval corresponding to an appropriate index value.
We distinguish notationally between the per-event PoT chain digest $R_i$ defined in the preceding subsection and the
per-interval set of PoT records $\mathcal{R}_i$: the roman symbol $R_i$ always denotes a single hash-chain value,
whereas the calligraphic symbol $\mathcal{R}_i$ denotes a set of receipts.
The ledger state is summarised by a commitment $C_i = H(\mathrm{TPL}_i)$ that can be anchored periodically to Bitcoin,
a notary chain, or a widely monitored transparency log.

To simplify notation we write $\mathrm{TPL}_i, \mathcal{S}_i, \mathcal{R}_i$ for the ledger prefix and associated artefacts
after reporting interval~$i$, and we use $t_i$ for the logical time associated with that prefix.
To keep notation compact we write $e_k$ for the $k$-th treasury event in the globally ordered sequence and
let $t_k$ denote its timestamp.
Reporting intervals are indexed separately by $i$, with $\mathrm{TPL}_i$ and $C_i$ referring to the ledger prefix
and commitment after interval~$i$ and logical time $t_i$ denoting the right end-point of that interval.
Where convenient we reuse the symbol $i$ for both an event index and a reporting-interval index, but the underlying
objects are always clear from context: $e_i$ denotes an event, whereas $\mathrm{TPL}_i$ and~$C_i$ denote a ledger
prefix and its commitment.

For each domain $d$ and time $t$, the ledger yields a derived exposure $B_d(t)$ obtained by aggregating all events up to time~$t$
together with all intervening PoT receipts and any encumbrance metadata.
Subject to valuation policies, aggregate on-balance-sheet Bitcoin exposure at time $t$ is then
\[
  B_{\mathrm{total}}(t) = \sum_{d \in \mathcal{D}} B_d(t).
\]

\subsection{Protocol syntax and algorithms}

We model the TPL as a tuple of probabilistic polynomial-time algorithms
\[
  \Pi_{\mathsf{TPL}} = (\mathsf{Setup}, \mathsf{AppendEvent}, \mathsf{Snapshot}, \mathsf{Anchor}, \mathsf{GenView}, \mathsf{VerifyView})
\]
that run under a security parameter $\lambda$ and maintain an internal ledger state and associated cumulative commitments.
We write $\mathsf{TPL}$ (without an index) for the abstract protocol and $\mathrm{TPL}_i$ for the ledger prefix after $i$ inputs.
Formally, each ledger state $\mathrm{TPL}_i$ is a prefix of a linear sequence of typed records $(\mathrm{rec}_1,\ldots,\mathrm{rec}_i)$, each of which is authenticated under the cumulative commitment $C_i$.

\begin{description}
  \item[$\mathsf{Setup}(1^\lambda)$:]
  Initialise an empty ledger $\mathrm{TPL}_0 := \epsilon$, set the initial commitment $C_0 := H(\epsilon)$, and generate any signing keys required for treasury and service-provider authentication. Output the public parameters $\mathsf{pp}$ and initial state $(\mathrm{TPL}_0, C_0)$.

  \item[$\mathsf{AppendEvent}(\mathrm{TPL}_{i-1}, e_i)$:]
  Parse the event $e_i$ as $(t_i, d_{\mathrm{src}}, d_{\mathrm{dst}}, v_i, \mathsf{evid}_i, m_i)$,
  where $t_i$ is a timestamp, $d_{\mathrm{src}}$ and $d_{\mathrm{dst}}$ are source and destination
  domains, $v_i$ is the BTC-denominated transfer value, $\mathsf{evid}_i$ is primitive evidence
  (for example, on-chain transaction identifiers, custody statements, or contractual references),
  and $m_i$ is policy metadata.
  Compute a primitive hash $h_i := H(\mathsf{evid}_i \,\Vert\, v_i \,\Vert\, m_i)$ and a PoT receipt
  \[
     R_i := H(R_{i-1} \,\Vert\, h_i \,\Vert\, d_{\mathrm{src}} \,\Vert\, d_{\mathrm{dst}} \,\Vert\, t_i),
  \]
  and form the signed PoT record $\mathsf{rec}_i$ as above (including treasury and service-provider signatures on $(h_i, d_{\mathrm{src}}, d_{\mathrm{dst}}, t_i, R_i)$).
  Extend the ledger by setting $\mathcal{R}_i := \mathcal{R}_{i-1} \cup \{\mathsf{rec}_i\}$, updating domain balances,
  and appending the new record to obtain $\mathrm{TPL}_i$ and commitment $C_i := H(\mathrm{TPL}_i)$.

  \item[$\mathsf{Snapshot}(\mathrm{TPL}_i)$:]
  Using the current ledger prefix, recompute domain exposures $B_d(t_i)$ by folding all PoT receipts and encumbrance metadata since the last snapshot.
  Commit the vector $(B_d(t_i))_{d\in\mathcal{D}}$ using an authenticated data structure (for example, a Merkle tree or vector commitment), and record the resulting digest inside $\mathcal{S}_i$.

  \item[$\mathsf{Anchor}(C_i)$:]
  Construct a Bitcoin transaction embedding $C_i$ in an \texttt{OP\_RETURN} output (or equivalent notary mechanism), broadcast it to the network, and record the transaction identifier and block height in $M_i$ once the anchor reaches the required confirmation depth.

  \item[$\mathsf{GenView}(\mathrm{TPL}_i,\alpha)$:]
  Apply a policy $V_\alpha$ to the ledger prefix to derive a view $\mathrm{View}_\alpha := V_\alpha(\mathrm{TPL}_i)$.
  The policy specifies filters, aggregations, redactions, and timing rules (for example, minimum materiality thresholds or reporting lags) and returns both the derived table of balances and flows and the subset of PoR and PoT artefacts needed to verify it.

  \item[$\mathsf{VerifyView}(\mathrm{View}_\alpha,\alpha,\mathcal{A}_{\mathrm{BTC}})$:]
  Given a candidate view, the public policy description, and access to the set of Bitcoin anchors $\mathcal{A}_{\mathrm{BTC}}$, reconstruct the sequence of commitments $\{C_j\}$ referenced in the view, check that they are all anchored with sufficient confirmations, and verify that the included PoR and PoT artefacts are consistent with the commitments.
  Accept if and only if all checks pass and recomputing $V_\alpha(\mathrm{TPL}_i)$ from the implied ledger digest yields $\mathrm{View}_\alpha$.
\end{description}
\paragraph{Bitcoin verification oracle.}
We model access to the underlying Bitcoin ledger via an abstract oracle $\mathcal{A}_{\mathrm{BTC}}$ that, given a set of transactions and a confirmation depth parameter, checks their validity and inclusion at the required depth.

This interface abstracts over concrete implementation choices for the primitive authenticated data structures and signature schemes, and will be used in the security definitions below.

\subsection{TPL as a state machine}
\label{subsec:tpl-state-machine}

For the security analysis it is convenient to refine the state evolution of Section~\ref{subsec:state-evolution}
and view the TPL as a stateful machine whose state evolves monotonically over a sequence of inputs.
Let $\Sigma$ be the set of internal states and let $\mathcal{I}$ be the set of inputs.
A state $\sigma_i \in \Sigma$ at logical step $i$ has the form
\[
  \sigma_i = (\mathrm{TPL}_i, C_i, \mathsf{aux}_i)
\]
where:
\begin{itemize}[noitemsep]
  \item $\mathrm{TPL}_i$ is the ledger prefix after processing $i$ events,
  \item $C_i := H(\mathrm{TPL}_i)$ is the corresponding commitment, and
  \item $\mathsf{aux}_i$ denotes auxiliary bookkeeping state
        (such as the current vector of domain exposures, the latest snapshot and anchor metadata,
        and any policy catalogue) that can be recomputed from $\mathrm{TPL}_i$ if needed.
\end{itemize}
The initial state $\sigma_0 = (\mathrm{TPL}_0, C_0, \mathsf{aux}_0)$ is produced by
$\mathsf{Setup}(1^\lambda)$.

The input alphabet $\mathcal{I}$ comprises four kinds of inputs:
\begin{itemize}[noitemsep]
  \item \emph{event inputs} $(\mathsf{event}, e_i)$ representing treasury events,
  \item \emph{snapshot triggers} $(\mathsf{snapshot})$,
  \item \emph{anchoring triggers} $(\mathsf{anchor})$, and
  \item \emph{view queries} $(\mathsf{view}, \alpha)$ for a public policy identifier $\alpha$.
\end{itemize}
We write $\delta : \Sigma \times \mathcal{I} \rightarrow \Sigma \times \mathcal{O}$
for the (probabilistic) transition function, where $\mathcal{O}$ is the set of outputs.
Informally, $\delta$ is induced by the algorithms of $\Pi_{\mathsf{TPL}}$ as follows:
\begin{itemize}[noitemsep]
  \item On input $(\mathsf{event}, e_i)$ in state $\sigma_{i-1}$, the machine runs
        $\mathsf{AppendEvent}(\mathrm{TPL}_{i-1}, e_i)$ to obtain $(\mathrm{TPL}_i, C_i)$
        and updates the auxiliary state accordingly.
  \item On a snapshot trigger, the machine runs $\mathsf{Snapshot}(\mathrm{TPL}_i)$,
        appending the resulting snapshot record to $\mathrm{TPL}_i$ and updating $\mathsf{aux}_i$.
  \item On an anchoring trigger, the machine runs $\mathsf{Anchor}(C_i)$
        and updates the anchor metadata in $\mathsf{aux}_i$ when the transaction reaches the required depth.
  \item On a view query $(\mathsf{view}, \alpha)$, the machine runs
        $\mathsf{GenView}(\mathrm{TPL}_i,\alpha)$ to derive $\mathrm{View}_\alpha$
        and may subsequently call $\mathsf{VerifyView}$ to check a candidate view.
        The internal state $(\mathrm{TPL}_i, C_i, \mathsf{aux}_i)$ is not changed by such queries.
\end{itemize}

Thus the TPL induces a (possibly randomised) state machine
\[
  \mathcal{M}_{\mathsf{TPL}} = (\Sigma, \sigma_0, \mathcal{I}, \mathcal{O}, \delta)
\]
whose state sequence $(\sigma_i)_{i\geq 0}$ is a deterministic function of the
initial state and the sequence of inputs, given the random coins of the underlying
algorithms.
Crucially, the ledger component is \emph{append-only}:
\[
  \mathrm{TPL}_0 \;\preceq\; \mathrm{TPL}_1 \;\preceq\; \cdots \;\preceq\; \mathrm{TPL}_i \;\preceq\; \cdots
\]
where $\preceq$ denotes the prefix relation.
This monotonicity underlies the forward-integrity and non-equivocation guarantees
formalised in the next subsection.

\subsection{Security properties}
\label{subsec:security}

We now formalise the core integrity and soundness properties of the TPL
under the system and threat model of Sections~\ref{sec:model}--\ref{sec:threat}
and the protocol interface of the previous subsection.
Throughout this section we assume the admissible schedules from Definition~\ref{def:exp-soundness}; these bounds affect only liveness and not the integrity properties we prove. The conservation, non-equivocation, and view-correctness theorems (Theorems~\ref{thm:conservation}--\ref{thm:view-correctness}) rely only on the cryptographic assumptions of Section~\ref{sec:model-assumptions}. The minimal non-collusion condition in item~(5) of Theorem~\ref{lem:restricted-exp-soundness} is needed only for our restricted exposure-soundness result.
We first give game-based definitions, then relate them to the hash-chained,
anchored-log structure via reduction-style theorems.

\paragraph{Game-based security definitions.}

Throughout we fix a security parameter $\lambda$ and consider adversaries
$\mathcal{A}$ that are probabilistic polynomial-time (PPT) in~$\lambda$.
We write $\mathsf{negl}(\cdot)$ for negligible functions, i.e., functions that
decrease faster than any inverse polynomial in~$\lambda$.
The challenger runs the TPL as the state machine
$\mathcal{M}_{\mathsf{TPL}}$ defined in Section~\ref{subsec:tpl-state-machine},
with honest treasurers and service providers following the protocol.

\begin{definition}[PoT receipt unforgeability]
\label{def:pot-unforgeability}
Let $\mathsf{Exp}^{\mathsf{pot-forge}}_{\mathcal{A}}(\lambda)$ be the following experiment.
The challenger runs $\mathsf{Setup}(1^\lambda)$ to obtain public parameters
$\mathsf{pp}$ and initial state $(\mathrm{TPL}_0,C_0)$ and gives $\mathsf{pp}$ to~$\mathcal{A}$.
The adversary has oracle access to an \emph{append oracle}
$\mathcal{O}_{\mathsf{append}}$ that, on input a well-formed event $e$,
executes $\mathsf{AppendEvent}$ on the current ledger state, updates the state,
and returns the resulting signed PoT record $\mathsf{rec}_i$ (including $R_i$ and the treasury and service-provider signatures).
At the end of the interaction $\mathcal{A}$ outputs a purported PoT chain
$\mathcal{R}^\star = (R^\star_1,\ldots,R^\star_\ell)$ together with a sequence
of associated records $\mathsf{rec}^\star_i = (h^\star_i,d^\star_{\mathrm{src},i},
d^\star_{\mathrm{dst},i},t^\star_i,\sigma^\star_{\mathrm{treas},i},\sigma^\star_{\mathrm{prov},i})$.
The experiment checks that:
\begin{enumerate}[noitemsep]
  \item the sequence $\mathcal{R}^\star$ is internally consistent with the
        public PoT chaining rule and the signature scheme: for all $1 \leq i \leq \ell$ we have
        \[
          R^\star_i = H(R^\star_{i-1} \,\Vert\, h^\star_i \,\Vert\, d^\star_{\mathrm{src},i} \,\Vert\, d^\star_{\mathrm{dst},i} \,\Vert\, t^\star_i),
        \]
        and the treasury and service-provider signatures $\sigma^\star_{\mathrm{treas},i},\sigma^\star_{\mathrm{prov},i}$ on $(h^\star_i,d^\star_{\mathrm{src},i},d^\star_{\mathrm{dst},i},t^\star_i,R^\star_i)$ verify under their respective public keys; and
  \item $\mathcal{R}^\star$ is \emph{not} equal to any prefix of the honestly
        generated sequence of receipts returned by $\mathcal{O}_{\mathsf{append}}$:
        that is, there exists an index $1 \leq j \leq \ell$ such that
        $R^\star_j$ was never output by $\mathcal{O}_{\mathsf{append}}$ in this
        experiment.
\end{enumerate}
If both checks pass the experiment outputs~$1$.
We say that the PoT construction is \emph{unforgeable} if for all PPT adversaries
$\mathcal{A}$ the probability
$\Pr[\mathsf{Exp}^{\mathsf{pot-forge}}_{\mathcal{A}}(\lambda)=1]$
is negligible in $\lambda$.
\end{definition}

We stress that exposure soundness, and the related notion of policy completeness introduced below, are strong target notions:
in the remainder of the paper we prove only a restricted variant for closed sets of domains under faithful policies
(see Theorem~\ref{lem:restricted-exp-soundness}), while the full experiment
$\mathsf{Exp}^{\mathsf{exp-sound}}$ for arbitrary policies and partially participating domains is left as a goal for
future work.

\begin{definition}[Exposure soundness]
\label{def:exp-soundness}
Let $\mathsf{Exp}^{\mathsf{exp-sound}}_{\mathcal{A}}(\lambda)$ be the following
experiment between a challenger and a PPT adversary~$\mathcal{A}$.
\begin{enumerate}[noitemsep]
  \item The challenger runs $\mathsf{Setup}(1^\lambda)$ to obtain
        $(\mathsf{pp},\mathrm{TPL}_0,C_0)$ and gives $\mathsf{pp}$ to~$\mathcal{A}$.
  \item Interacting with honest treasurers and service providers that follow
        the protocol, $\mathcal{A}$ obtains PoR and PoT artefacts, snapshots,
        and anchored commitments, and may adaptively influence the timing and
        sequence of inputs to $\mathcal{M}_{\mathsf{TPL}}$.
        The resulting schedule of inputs must respect the liveness bounds
        $(\Delta_{\mathsf{event}},\Delta_{\mathsf{snap}},\Delta_{\mathsf{anchor}})$
        from Section~\ref{sec:threat}: every valid treasury event is either logged or
        rejected within $\Delta_{\mathsf{event}}$, snapshot triggers produce PoR snapshots
        at least once every $\Delta_{\mathsf{snap}}$, and TPL commitments are anchored
        to Bitcoin within $\Delta_{\mathsf{anchor}}$ of submission. We implicitly restrict all subsequent experiments to such admissible schedules.
  \item Eventually $\mathcal{A}$ outputs an alternative ledger prefix
        $\mathrm{TPL}^\star_i$, a corresponding set of Bitcoin transactions
        $\mathcal{M}^\star$ that it claims anchor commitments for
        $\mathrm{TPL}^\star_i$, and a logical time index $t_i$.
  \item The experiment recomputes domain balances $B_d(t_i)$ from the honest
        ledger $\mathrm{TPL}_i$ and $B_d^\star(t_i)$ from $\mathrm{TPL}^\star_i$,
        using the PoR and PoT artefacts implied by each ledger.
        It checks that all PoR/PoT verifications on $\mathrm{TPL}^\star_i$ and
        its anchors succeed and that $\mathcal{M}^\star$ are valid Bitcoin
        transactions at sufficient confirmation depth embedding the claimed
        commitments.
        It outputs $1$ (meaning $\mathcal{A}$ wins) if
        \[
          \sum_{d\in\mathcal{D}} B_d^\star(t_i) \;>\;
          \sum_{d\in\mathcal{D}} B_d(t_i)
        \]
        while all verification procedures accept.
\end{enumerate}
We say that TPL satisfies \emph{exposure soundness} if, for all PPT adversaries
$\mathcal{A}$, the probability
\[
  \Pr[\mathsf{Exp}^{\mathsf{exp-sound}}_{\mathcal{A}}(\lambda)=1]
\]
is negligible in $\lambda$. As throughout, this experiment is interpreted under the implicit assumptions of domain completeness for $\mathcal{D}$ and algebraic consistency (Assumption~\ref{assump:algebraic-consistency}).
\end{definition}

\begin{definition}[Policy completeness]
\label{def:policy-completeness}
Fix a public policy predicate $P(e,t,\alpha)$ with a latency bound $\Delta_P$ that, given an event $e$, a time $t$, and an observer class $\alpha$, specifies whether $e$ should be visible to $\alpha$ by time $t$.
In the experiment $\mathsf{Exp}^{\mathsf{pol-comp}}_{\mathcal{A}}(\lambda)$
the challenger and adversary first interact as in
$\mathsf{Exp}^{\mathsf{exp-sound}}_{\mathcal{A}}(\lambda)$.
As in Definition~\ref{def:exp-soundness}, the induced schedule of inputs must respect the
liveness bounds $(\Delta_{\mathsf{event}},\Delta_{\mathsf{snap}},\Delta_{\mathsf{anchor}})$
from Section~\ref{sec:threat}; adversarial strategies that violate these bounds fall
outside the modelled setting. This interaction yields a ledger
state $\mathrm{TPL}_i$ at time $t_i$ and the corresponding set of Bitcoin anchors for the prefix.
The adversary then outputs a candidate view $\mathrm{View}^\star_\alpha$ for some
policy $V_\alpha$ implementing~$P$ and an observer class~$\alpha$.
The experiment recomputes the honest ledger prefix $\mathrm{TPL}_i$ and runs
$\mathsf{VerifyView}(\mathrm{View}^\star_\alpha,\alpha,\mathcal{A}_{\mathrm{BTC}})$
using the public anchors and policy description.
If this check rejects, the experiment outputs~$0$.
Otherwise it checks whether there exists an event $e$ such that $P(e,t_i,\alpha)=1$ but
$e$ is absent from $\mathrm{View}^\star_\alpha$.
If such an $e$ is found, the experiment outputs~$1$; otherwise it outputs~$0$.
We say that TPL satisfies \emph{policy completeness} for~$P$ if for all PPT
adversaries $\mathcal{A}$ the probability
$\Pr[\mathsf{Exp}^{\mathsf{pol-comp}}_{\mathcal{A}}(\lambda)=1]$
is negligible in $\lambda$.
\end{definition}

As in Definition~\ref{def:exp-soundness}, all interactions in $\mathsf{Exp}^{\mathsf{pol-comp}}_{\mathcal{A}}(\lambda)$
are implicitly restricted to admissible schedules that respect the liveness parameters
$(\Delta_{\mathsf{event}},\Delta_{\mathsf{snap}},\Delta_{\mathsf{anchor}})$ introduced in Section~\ref{sec:threat};
adversarial strategies that violate these bounds fall outside the modelled setting.

We emphasise that, like exposure soundness, policy completeness is intended as a strong, deployment-level target notion.
In this paper we do not attempt to show that arbitrary real-world deployments of TPL satisfy
$\mathsf{Exp}^{\mathsf{pol-comp}}$; instead we study more restricted settings and leave a full analysis of
policy completeness for expressive policy languages to future work.
As a simple negative example, consider a treasury that actually operates an additional derivatives desk $d_\star$ that is never declared as a domain in $\mathcal{D}$.
The operator can route all risky exposures to $d_\star$ while only logging conservative flows among the declared domains.
By construction, the balance vector on the reported domain set will satisfy conservation and all PoR/PoT checks, so the exposure-soundness experiment of Definition~\ref{def:exp-soundness} will accept, even though the true economic exposure of the treasury is misrepresented.
No cryptographic construction over TPL can prevent this: the failure lies in the reporting perimeter (domain completeness), not in the protocol or its primitives.

We can however establish policy completeness for a simple but non-trivial history-revealing policy, which serves as a sanity check on Definition~\ref{def:policy-completeness}.

\begin{lemma}[Policy completeness for a history-revealing policy]
\label{lem:policy-completeness-history}
Fix an observer class $\alpha = \mathsf{hist}$ and let $P_{\mathsf{hist}}(e,t,\alpha)$ be the predicate that holds
iff event $e$ occurs at or before logical time $t$ and is not filtered out by the policy-language filters for~$\alpha$.
Let $V_{\mathsf{hist}}$ be the deterministic, stateless policy that, on input a ledger prefix $\mathrm{TPL}_i$, outputs
exactly the subsequence of events $e$ with $P_{\mathsf{hist}}(e,t_i,\alpha)=1$ together with the PoR, PoT, and anchoring
artefacts needed for $\mathsf{VerifyView}$ to re-check these events and the corresponding commitment $C_i$.

Assume that the liveness bounds from Section~\ref{sec:threat} hold, that TPL satisfies non-equivocation in the sense of
Definition~\ref{def:non-equivocation}, and that view correctness holds for $V_{\mathsf{hist}}$
(Theorem~\ref{thm:view-correctness}).
Then TPL satisfies policy completeness for $P_{\mathsf{hist}}$ in the sense of Definition~\ref{def:policy-completeness}.
\end{lemma}

\begin{proof}[Proof sketch]
Consider any PPT adversary $\mathcal{A}$ in $\mathsf{Exp}^{\mathsf{pol-comp}}_{\mathcal{A}}(\lambda)$ for
$P_{\mathsf{hist}}$ and $V_{\mathsf{hist}}$.
By liveness, every valid event that should be visible by time $t_i$ appears in the honest ledger prefix $\mathrm{TPL}_i$
with high probability.
By view correctness and non-equivocation, there is at most one view consistent with $\mathrm{TPL}_i$ and the public
anchors that will be accepted by $\mathsf{VerifyView}$ for $V_{\mathsf{hist}}$, namely the one output by an honest
execution of $\mathsf{GenView}$.
Any $\mathrm{View}^\star_\alpha$ that omits an event with $P_{\mathsf{hist}}(e,t_i,\alpha)=1$ while still passing
$\mathsf{VerifyView}$ would therefore contradict either liveness (the event never being logged) or the combination of
non-equivocation and view correctness (two distinct accepted views for the same anchored ledger).
Hence the probability that $\mathsf{Exp}^{\mathsf{pol-comp}}_{\mathcal{A}}(\lambda)$ outputs~$1$ is negligible.
\end{proof}

This definition captures the inability of any efficient adversary to produce a \emph{valid} policy-based view that omits events that should be visible under $P$: the view must first pass $\mathsf{VerifyView}$ before any omission is counted as a win. It is therefore a property of policy-based projections of the ledger, complementary to exposure soundness, which captures whether the underlying ledger itself correctly records the relevant events.

In the remainder of this section we use these experiments as targets for three integrity theorems:
conservation across domains, forward integrity and non-repudiation, and view correctness for policy-based views.
A full reduction showing that arbitrary deployments of TPL satisfy exposure soundness and policy completeness under all
operational assumptions is left to future work; here we restrict attention to these three core properties, which can be
established under standard cryptographic assumptions.

To connect these properties to the exposure-soundness notion of Definition~\ref{def:exp-soundness},
we record a simple restricted lemma that captures how conservation, PoR/PoT soundness, and
non-equivocation jointly constrain attempts to overstate total exposure.

\begin{definition}[Faithful policy]
\label{def:faithful-policy}
A policy $V_\alpha$ is \emph{faithful} if for any two ledgers $\mathrm{TPL}$
and $\mathrm{TPL}'$ and any index~$i$,
whenever the sequences of events that are visible under $V_\alpha$ in
$\mathrm{TPL}_i$ and $\mathrm{TPL}'_i$ coincide, we have
\[
  V_\alpha(\mathrm{TPL}_i) = V_\alpha(\mathrm{TPL}'_i).
\]
\end{definition}

\paragraph{Remark.}
Many natural treasury policies are faithful in this sense.
Intuitively, faithfulness requires that a policy's output be determined solely by the subsequence of events that the policy is
allowed to inspect, rather than by any hidden details of ledger evolution.
This matches the class of policies we intend TPL to support for assurance purposes.

\begin{lemma}
\label{lem:policy-language-faithful}
Every policy expressible in the simple policy language described in Section~\ref{sec:responsible} is faithful in the sense of Definition~\ref{def:faithful-policy}.
\end{lemma}

\begin{proof}[Proof sketch]
Fix an observer class~$\alpha$ and a policy specified by $(\mathsf{Filter}_\alpha,\mathsf{Label}_\alpha,\mathsf{Agg}_\alpha,\Delta_\alpha,\theta_\alpha)$ in the language above.
By construction, the value $V_\alpha(\mathrm{TPL}_i)$ is a deterministic function of the subsequence of events in $\mathrm{TPL}_i$ that satisfy $\mathsf{Filter}_\alpha$, together with their domains and timestamps as seen through $\mathsf{Label}_\alpha$ and the bucketing induced by~$\Delta_\alpha$.
If two ledgers $\mathrm{TPL}$ and $\mathrm{TPL}'$ have the same visible subsequence of events for observer~$\alpha$ up to index~$i$, then after applying the same relabelling, bucketing, materiality filter~$\theta_\alpha$, and aggregation operator~$\mathsf{Agg}_\alpha$ to both, we obtain identical outputs.
Hence $V_\alpha(\mathrm{TPL}_i) = V_\alpha(\mathrm{TPL}'_i)$ whenever the sequences of events that are visible under $V_\alpha$ in $\mathrm{TPL}_i$ and $\mathrm{TPL}'_i$ coincide, which is exactly the requirement of Definition~\ref{def:faithful-policy}.
\end{proof}

\begin{definition}[Closed subset of domains]
\label{def:closed-domains}
Let $\mathcal{D}$ be the global set of domains and let $d_{\mathrm{fee}} \in \mathcal{D}$ denote the distinguished fee domain.
For a time horizon $[0,t_i]$ we say that a subset $\mathcal{D}_0 \subseteq \mathcal{D} \setminus \{d_{\mathrm{fee}}\}$ is \emph{closed over $[0,t_i]$}
if for every treasury event $e_k = (t_k,d_{\mathrm{src}},d_{\mathrm{dst}},v,\mathsf{evid},m)$ with $t_k \le t_i$ and
$\{d_{\mathrm{src}},d_{\mathrm{dst}}\} \cap \mathcal{D}_0 \neq \emptyset$ we have
$d_{\mathrm{src}},d_{\mathrm{dst}} \in \mathcal{D}_0 \cup \{d_{\mathrm{fee}}\}$.
Equivalently, no event up to time $t_i$ transfers value between a domain in $\mathcal{D}_0$ and any external counterparty.
\end{definition}
\paragraph{Example (closed and non-closed domain sets).}
For example, suppose $\mathcal{D}_0$ consists of the cold-storage domain and a set of regulated custodians, and that all exchange balances are treated as external counterparties.
Under a typical treasury policy that never moves funds directly from custodians to exchanges without logging a PoT event, $\mathcal{D}_0$ will be closed over $[0,t_i]$ in the sense of Definition~\ref{def:closed-domains}.
By contrast, if there are direct, unlogged transfers between custodian wallets and external exchanges, then $\mathcal{D}_0$ fails to be closed and our restricted exposure-soundness guarantee no longer applies.

\begin{theorem}[Restricted exposure soundness]
\label{lem:restricted-exp-soundness}
Consider a deployment of TPL and a fixed logical time index~$i$. under the implicit assumptions of domain completeness for $\mathcal{D}$ and algebraic consistency (Assumption~\ref{assump:algebraic-consistency}).
Let $\mathcal{D}_0 \subseteq \mathcal{D} \setminus \{d_{\mathrm{fee}}\}$ be a subset that is \emph{closed over $[0,t_i]$}
in the sense of Definition~\ref{def:closed-domains}.
Suppose there exists a faithful policy $V_{\mathsf{exp}}$ that, on input $\mathrm{TPL}_i$,
outputs either the vector $(B_d(t_i))_{d \in \mathcal{D}_0}$ or their algebraic sum.
Assume that:
\begin{enumerate}[noitemsep]
  \item the underlying PoR primitive is existence- and ownership-sound;
  \item the PoT construction is unforgeable in the sense of Definition~\ref{def:pot-unforgeability};
  \item the hash function $H$ is collision-resistant; and
  \item TPL satisfies non-equivocation in the sense of Definition~\ref{def:non-equivocation}.
  \item the \textbf{minimal non-collusion assumption} of Section~\ref{sec:threat} holds for all domains in $\mathcal{D}_0$.
\end{enumerate}
By Theorem~\ref{thm:non-equivocation}, assumption~(4) in turn holds whenever the hash function and PoT scheme satisfy
the conditions of that theorem.
Then any PPT adversary that wins an exposure-soundness experiment restricted to $\mathcal{D}_0$ at time $t_i$, that is,
produces an alternative anchored ledger prefix $\mathrm{TPL}^\star_i$ for which
\[
  \sum_{d \in \mathcal{D}_0} B^\star_d(t_i)
  \;>\;
  \sum_{d \in \mathcal{D}_0} B_d(t_i)
\]
while all PoR/PoT verification procedures and view verifications for $V_{\mathsf{exp}}$ accept, can be turned into a PPT
adversary that breaks at least one of the above assumptions with non-negligible probability.
In particular, under these assumptions and the conservation property of Theorem~\ref{thm:conservation}, the success
probability of any PPT adversary in this restricted exposure-soundness game, that is, in the experiment $\mathsf{Exp}^{\mathsf{exp\mbox{-}sound},\mathcal{D}_0}_{\mathcal{A}}(\lambda)$ of Appendix~\ref{app:exp-soundness-reduction}, is negligible in~$\lambda$.
\end{theorem}

Intuitively, the minimal non-collusion condition in item~(5) ensures that the net flow of value
between the closed domain set $\mathcal{D}_0$ and the fee domain $d_{\mathrm{fee}}$ that drives
the algebraic conservation argument of Theorem~\ref{thm:conservation} is faithfully reflected
in the PoR snapshots and PoT receipts logged for those domains.
Without at least one honest role per domain, a fully colluding deployment could hide additional
liabilities or fabricate inflows while still passing all local verification checks.

More generally, we do not claim exposure soundness outside closed domain sets and faithful policies: in the unrestricted multi-domain setting, exposure soundness as formalised in Definition~\ref{def:exp-soundness} remains a target notion rather than a proved theorem in this work.

\begin{proof}[Proof sketch]
Fix $\mathcal{D}_0$ and~$i$ as in the statement and let $\mathcal{A}$ be any
PPT adversary for the restricted exposure-soundness game.
Write $E$ for the event that $\mathcal{A}$ succeeds, i.e., that it outputs an
alternative anchored ledger prefix $\mathrm{TPL}^\star_i$ such that
\[
  \sum_{d \in \mathcal{D}_0} B^\star_d(t_i)
  \;>\;
  \sum_{d \in \mathcal{D}_0} B_d(t_i)
\]
while all PoR and PoT verifications and the view verification for
$V_{\mathsf{exp}}$ accept.

Because $\mathcal{D}_0$ is closed over $[0,t_i]$ in the sense of
Definition~\ref{def:closed-domains}, every treasury event that affects a
domain in~$\mathcal{D}_0$ up to time~$t_i$ is internal to
$\mathcal{D}_0 \cup \{d_{\mathrm{fee}}\}$.
By Theorem~\ref{thm:conservation}, the algebraic sum
\[
  S(t) \;:=\; \sum_{d \in \mathcal{D}_0} B_d(t)
\]
is therefore completely determined, for each $t \le t_i$, by the initial
exposures of domains in~$\mathcal{D}_0$ together with the net flow of value
between $\mathcal{D}_0$ and $d_{\mathrm{fee}}$ up to time~$t$.
The minimal non-collusion assumption guarantees that this net flow is fully
reflected in the logged PoR and PoT artefacts for domains in~$\mathcal{D}_0$.

In the honest execution these artefacts induce some canonical value
$S(t_i)$.
On event~$E$, the adversary outputs an alternative ledger prefix
$\mathrm{TPL}^\star_i$ and corresponding artefacts that are all locally
valid, yet yield a strictly larger sum $S^\star(t_i)$ when interpreted under
the same faithful policy $V_{\mathsf{exp}}$.
Consequently, there must exist an earliest time $t_k \le t_i$ at which the
effect of the ledger on $S(t)$ differs between the honest and adversarial
executions.
We now analyse the event at time~$t_k$ and map it to a violation of one of the
assumptions in the theorem.

\begin{itemize}[noitemsep]
  \item \emph{PoR discrepancy.}
        If the first difference arises from a PoR snapshot whose claimed total
        exposure for some domain in~$\mathcal{D}_0$ differs between
        $\mathrm{TPL}_i$ and $\mathrm{TPL}^\star_i$, yet all inclusion proofs
        and verification checks accept, then we can build from $\mathcal{A}$
        an adversary against existence/ownership soundness of the PoR
        primitive.
        This adversary runs $\mathcal{A}$ as a subroutine, extracts the
        mismatched snapshot and its proof from the transcript at time~$t_k$,
        and outputs it as a forgery witnessing a coin that either does not
        exist or is not controlled by the treasury, contradicting assumption~(1).

  \item \emph{PoT discrepancy.}
        If the first difference arises from a transfer between domains in
        $\mathcal{D}_0 \cup \{d_{\mathrm{fee}}\}$ whose PoT receipt appears in
        $\mathrm{TPL}^\star_i$ but not in $\mathrm{TPL}_i$, or vice versa,
        while all PoT verifications accept, then we can build an adversary
        against PoT unforgeability (assumption~(2)).
        This adversary again runs $\mathcal{A}$, locates the earliest such
        discrepant receipt and the preceding honest PoT chain in the ledger,
        and outputs them as a forged extension that is not a prefix of the
        honest chain.

  \item \emph{Commitment/anchor discrepancy.}
        If all PoR snapshots and PoT receipts that affect
        $\mathcal{D}_0$ up to time~$t_i$ are identical in $\mathrm{TPL}_i$
        and $\mathrm{TPL}^\star_i$ and yet $S^\star(t_i) > S(t_i)$, then the
        serialised ledger prefixes must differ in some other way that changes
        the implied balances for domains in~$\mathcal{D}_0$.
        Since all anchors and view verifications are assumed to succeed,
        both ledger prefixes must be consistent with the on-chain commitments
        and the policy $V_{\mathsf{exp}}$.
        In this case we obtain either
        (i) two distinct ledger prefixes that are both consistent with the
        same anchored commitment value, yielding a collision in~$H$ and
        contradicting assumption~(3); or
        (ii) a violation of non-equivocation in the sense of
        Definition~\ref{def:non-equivocation}, contradicting assumption~(4),
        because an observer would be presented with two different, yet
        apparently valid, exposures for~$\mathcal{D}_0$ at time~$t_i$.
\end{itemize}

In all cases, a successful adversary $\mathcal{A}$ for the restricted
exposure-soundness game gives rise to a PPT adversary that breaks at least one
of assumptions~(1)--(4) with non-negligible probability, contradicting those
assumptions.
Hence the success probability of any PPT adversary in the restricted
exposure-soundness game is negligible in~$\lambda$, as claimed.
\end{proof}
\begin{definition}[Non-equivocation]
\label{def:non-equivocation}
Let $\mathsf{Exp}^{\mathsf{neq}}_{\mathcal{A}}(\lambda)$ be the following
experiment between a challenger and a PPT adversary~$\mathcal{A}$.
\begin{enumerate}[noitemsep]
  \item The challenger runs $\mathsf{Setup}(1^\lambda)$ to obtain public parameters
        $\mathsf{pp}$ and an initial state $(\mathrm{TPL}_0,C_0)$ and gives
        $\mathsf{pp}$ to~$\mathcal{A}$, who may interact with honest treasurers
        and service providers as in the previous experiments.
  \item Eventually $\mathcal{A}$ outputs two ledger prefixes
        $\mathrm{TPL}_i$ and $\mathrm{TPL}'_i$ together with two sets of Bitcoin
        transactions $\mathcal{M}$ and $\mathcal{M}'$ that it claims anchor
        commitments for the respective ledgers.
  \item The experiment checks that all transactions in $\mathcal{M}$ and
        $\mathcal{M}'$ are valid Bitcoin transactions at sufficient confirmation
        depth and that both ledgers are syntactically well formed.
  \item It then checks whether there exists an observer class $\alpha$ and
        associated policy $V_\alpha$ such that, letting
        $\mathrm{View}_\alpha := V_\alpha(\mathrm{TPL}_i)$ and
        $\mathrm{View}'_\alpha := V_\alpha(\mathrm{TPL}'_i)$,
        both $\mathsf{VerifyView}(\mathrm{View}_\alpha,\alpha,\mathcal{A}_{\mathrm{BTC}})$
        and $\mathsf{VerifyView}(\mathrm{View}'_\alpha,\alpha,\mathcal{A}_{\mathrm{BTC}})$
        accept, and $\mathrm{View}_\alpha \neq \mathrm{View}'_\alpha$.
        If so, the experiment outputs~$1$; otherwise it outputs~$0$.
\end{enumerate}
We say that TPL satisfies \emph{non-equivocation} if, for all PPT adversaries
$\mathcal{A}$, the probability
\[
  \Pr[\mathsf{Exp}^{\mathsf{neq}}_{\mathcal{A}}(\lambda)=1]
\]
is negligible in~$\lambda$.
\end{definition}
We do not address these challenges here. The restricted theorem below should be read as a proof-of-concept template rather than a deployable privacy guarantee, illustrating how to
phrase simulation-based privacy arguments for TPL.

For a simple public-investor policy we
state below a restricted privacy theorem that illustrates how such a treatment can be instantiated in our framework.

\begin{theorem}[Restricted privacy for a public-investor policy]
\label{thm:pub-privacy}
Fix an observer class $\alpha = \mathsf{pub}$ with policy $V_{\mathsf{pub}}$ as above and define the leakage function
$L_{\mathsf{pub}}$ that, on input $\mathrm{TPL}_i$, outputs for each reporting interval the pair consisting of
(i) the total Bitcoin exposure $B_{\mathrm{tot}}(t)$ across all domains and (ii) the aggregate encumbered exposure
$B_{\mathrm{enc}}(t)$ over all domains marked as encumbered in the policy metadata.
Consider the following two experiments for a PPT adversary~$\mathcal{A}$. We restrict attention here to a single, fixed public-investor policy and to non-adaptive observations at reporting dates: the adversary sees the sequence of views and anchors produced by an honest execution, but cannot adaptively change the policy or reporting schedule.
In the \emph{real} experiment the challenger runs the TPL protocol with honest parties,
restricted to the public-investor policy $V_{\mathsf{pub}}$, and lets $\mathcal{A}$ observe all
protocol messages, the resulting policy-based views $V_{\mathsf{pub}}(\mathrm{TPL}_i)$, and the Bitcoin
transactions used to anchor commitments.
In the \emph{ideal} experiment an ideal functionality first samples an internal execution of TPL,
computes $L_{\mathsf{pub}}(\mathrm{TPL}_i)$ and the associated sequence of anchoring transactions that appear on the
public Bitcoin chain, and then gives only this information to a PPT simulator~$\mathcal{S}$.
The simulator interacts with $\mathcal{A}$ and must produce a simulated view for $\mathcal{A}$.
Assume that:
\begin{enumerate}[noitemsep]
  \item the only artefacts made public about the treasury state are the Bitcoin anchoring transactions and the
        policy-based views $V_{\mathsf{pub}}(\mathrm{TPL}_i)$; and
  \item the underlying PoR and PoT primitives admit zero-knowledge or witness-indistinguishable proofs of correctness for their public outputs with respect to the committed ledger values, and any signatures used in TPL are standard existentially unforgeable schemes whose outputs can be generated by the simulator given the same signing keys as in the real execution (or are wrapped inside the same zero-knowledge proof system), so that these primitives reveal no auxiliary information beyond what is explicit in their outputs and the committed values.
        (These are abstract, idealised PoR/PoT primitives; we do not claim that existing Merkle-based PoR deployments
        are already zero-knowledge or witness-indistinguishable for this leakage profile.)
\end{enumerate}

As a concrete instantiation of assumption~(2), one can, for example, express a Provisions-style Merkle-tree PoR
\cite{Dagher2015} and the hash-chain transition relation for PoT receipts as arithmetic circuits and wrap them in a
general-purpose zk-SNARK or zk-STARK.
The public inputs of the proof are the current commitment $C_i$, the reporting interval, and the aggregate values
leaked by $L_{\mathsf{pub}}(\mathrm{TPL}_i)$, while the witness consists of the underlying wallet balances, encumbrance
flags, and receipt chain.
A zero-knowledge proof system for this statement yields PoR/PoT artefacts whose distribution can be simulated given
only the committed state and $L_{\mathsf{pub}}(\mathrm{TPL}_i)$, matching the leakage profile required by the theorem.

Then there exists a PPT simulator $\mathcal{S}$ such that for every PPT distinguisher~$\mathcal{D}$ the absolute
difference between
the probability that $\mathcal{D}$ outputs~$1$ on $\mathcal{A}$'s view in the real experiment and the probability
that $\mathcal{D}$ outputs~$1$ on $\mathcal{A}$'s view in the ideal experiment is negligible in the security parameter.
Equivalently, TPL is $(L_{\mathsf{pub}}, \mathsf{pub})$-private in the sense that the real and ideal experiments defined above are
computationally indistinguishable for any PPT adversary and distinguisher.
In particular, under these assumptions any information that an efficient public-market adversary learns about the
live treasury wallets from TPL is limited to what is already implied by $L_{\mathsf{pub}}(\mathrm{TPL}_i)$ and the
anchoring transactions themselves.
This theorem should therefore be read as an existence-style result illustrating how simulation-based privacy can be phrased for TPL under strong assumptions, rather than as a guarantee for any particular deployed primitive.
A high-level description of the simulator $\mathcal{S}$ and the associated hybrid argument
is given in the proof below.
A fully detailed simulation-based treatment of privacy for arbitrary policies and richer leakage functions is left
to future work.
\end{theorem}

\subsection{A candidate TPL instantiation}
\label{sec:candidate-instantiation}

The discussion above of assumption~(2) implicitly fixes one natural candidate instantiation for TPL: a Provisions-style Merkle-tree PoR for snapshots, a PoTT-style PoT mechanism that signs ordered receipts and links them in a hash chain, and a general-purpose zk argument system (for example a SNARK or STARK) that proves the joint correctness of the PoR and PoT artefacts and of the derived exposure vector and public-investor view.
In this instantiation the public inputs of the proof are the current commitment $C_i$, the reporting interval, and the aggregate values leaked by $L_{\mathsf{pub}}(\mathrm{TPL}_i)$, while the witness consists of the underlying wallet balances, encumbrance flags, and receipt chain.
The resulting on-chain artefacts have size dominated by logarithmic-size Merkle authentication paths for the UTXOs touched in a given reporting interval together with a small number of succinct proofs.
External verifiers need only check these proofs and recompute a small number of hashes, so their work is polylogarithmic in the number of underlying wallet entries and PoT receipts.
We deliberately do not fix a particular proof system, curve, or parameter set, since these choices are engineering trade-offs that may evolve; instead, our complexity and overhead discussion in Section~\ref{sec:complexity-overhead} provides order-of-magnitude cost estimates for realistic treasuries, and we leave detailed benchmarking of concrete stacks to future implementation work.

Intuitively, policies for which simulation-based privacy is plausible in our framework are those whose published views can be expressed as deterministic functions of a relatively small leakage profile, such as low-dimensional aggregates of domain balances and coarse-grained timing information.
Policies that encode highly path-dependent logic, fine-grained timing or ordering constraints, or rich cross-correlation structure between domains will in general require leakage functions that are too informative to be captured succinctly by $L_\alpha$; for such policies, a strong simulation-based guarantee cannot be expected without substantially weakening the privacy or functionality goals.

Concretely, public-investor policies such as $V_{\mathsf{pub}}$ can be instantiated using proof-of-reserves and proof-of-transit primitives that are wrapped in zero-knowledge or witness-indistinguishable arguments, together with standard digital-signature schemes modelled in the usual way. In practice, this can be achieved, for example, by embedding PoR and PoT statements in SNARK-based proofs and applying the Fiat--Shamir transform in the random-oracle model so that verifiers learn nothing beyond the leakage profile $L_{\mathsf{pub}}$.

Theorem~\ref{thm:pub-privacy} should therefore be read as a proof-of-concept for a single, simple leakage profile $L_{\mathsf{pub}}$: it shows that simulation-based privacy is compatible with our TPL framework for carefully designed public-investor policies, while extending such guarantees to richer policy languages and leakage functions remains an open problem.

\begin{proof}[Proof sketch]
Fix a PPT adversary $\mathcal{A}$ and a PPT distinguisher
$\mathcal{D}$.
Let $\mathsf{Real}_{\mathcal{A},\mathcal{D}}(\lambda)$ denote the
probability that $\mathcal{D}$ outputs~$1$ on $\mathcal{A}$'s view in the
real experiment described in the theorem, and let
$\mathsf{Ideal}_{\mathcal{A},\mathcal{D}}(\lambda)$ denote the corresponding
probability in the ideal experiment with simulator~$\mathcal{S}$.
We must construct $\mathcal{S}$ such that
\[
  \bigl|
    \mathsf{Real}_{\mathcal{A},\mathcal{D}}(\lambda)
    - \mathsf{Ideal}_{\mathcal{A},\mathcal{D}}(\lambda)
  \bigr|
\]
is negligible in~$\lambda$.

In the ideal experiment the functionality samples an internal execution of
TPL, obtains the ledger prefixes $\mathrm{TPL}_i$ and the associated leakage
values $L_{\mathsf{pub}}(\mathrm{TPL}_i)$ together with the sequence of
Bitcoin anchoring transactions, and gives only this information to the
simulator~$\mathcal{S}$.
On input this leakage, $\mathcal{S}$ proceeds in three steps.

\begin{enumerate}[noitemsep]
  \item \emph{Sampling a consistent internal state.}
        For each reporting interval it chooses an arbitrary decomposition of
        the leaked total exposure $B_{\mathrm{tot}}(t)$ and encumbered exposure
        $B_{\mathrm{enc}}(t)$ across an internal collection of synthetic
        domains and wallet addresses.
        This yields a synthetic internal ledger $\mathrm{TPL}'_i$ whose
        aggregate exposures match $L_{\mathsf{pub}}(\mathrm{TPL}_i)$ by
        construction.

  \item \emph{Simulating PoR, PoT, and signatures.}
        By assumption~(2), the PoR, PoT, and signature primitives admit
        zero-knowledge or witness-indistinguishable simulation for their
        public outputs.
        In particular, for each PoR invocation there exists a PPT simulator
        that, given only the public commitment and attested amount, produces a
        transcript that is indistinguishable from a real PoR transcript with
        some witness; similarly for PoT receipts and signatures.
        Simulator~$\mathcal{S}$ invokes these underlying simulators on the
        synthetic internal state to obtain PoR snapshots, PoT receipts, and
        signatures consistent with $\mathrm{TPL}'_i$.
        It then embeds the corresponding commitment values in Bitcoin
        anchoring transactions whose pattern (number and timing) matches the
        schedule provided by the ideal functionality.

  \item \emph{Computing the public-investor view.}
        Finally, $\mathcal{S}$ computes
        $\mathrm{View}^{\mathsf{sim}}_{\mathsf{pub}}
          := V_{\mathsf{pub}}(\mathrm{TPL}'_i)$.
        By the definition of $L_{\mathsf{pub}}$ and the construction in the
        first step, this view depends only on the leaked totals
        $B_{\mathrm{tot}}(t)$ and $B_{\mathrm{enc}}(t)$ and not on the
        particular internal decomposition chosen by~$\mathcal{S}$.
        The simulator feeds to $\mathcal{A}$ the resulting protocol messages,
        anchors, and public views, and forwards $\mathcal{A}$'s output to
        $\mathcal{D}$.
\end{enumerate}

Consider a hybrid execution in which we start from a real execution and
gradually replace the real PoR and PoT artefacts and signatures with the
simulated ones produced as above, keeping the leakage
$L_{\mathsf{pub}}(\mathrm{TPL}_i)$ and the anchor schedule fixed.
By the zero-knowledge or witness-indistinguishability guarantees in
assumption~(2), any PPT distinguisher can detect each such replacement only
with negligible advantage.
Assumption~(1) ensures that, apart from these artefacts, the only information
about the treasury state that is ever made public consists of the Bitcoin
anchoring transactions and the views $V_{\mathsf{pub}}(\mathrm{TPL}_i)$, both
of which are reproduced exactly by~$\mathcal{S}$ from
$L_{\mathsf{pub}}(\mathrm{TPL}_i)$.

It follows by a standard hybrid argument that the overall joint distribution
of the public transcript and $\mathcal{A}$'s view in the ideal execution with
simulator~$\mathcal{S}$ is computationally indistinguishable from that in the
real execution.
Equivalently, for every PPT distinguisher~$\mathcal{D}$ the quantity
$\bigl|
  \mathsf{Real}_{\mathcal{A},\mathcal{D}}(\lambda)
  - \mathsf{Ideal}_{\mathcal{A},\mathcal{D}}(\lambda)
\bigr|$
is negligible in~$\lambda$, which is exactly the privacy claim of the theorem.
\end{proof}
\paragraph{Forward integrity and non-repudiation.}
Informally, TPL enjoys forward integrity if, once a prefix of the ledger has been anchored, no adversary can later change previously committed events without causing a verification failure for any policy-based view.
Any such attack would necessarily produce two anchored prefixes with different projections under some policy and thus constitutes a win in the non-equivocation experiment of Definition~\ref{def:non-equivocation}.
We therefore treat non-equivocation as the main formal vehicle for arguing forward integrity and non-repudiation of logged events.
\begin{proposition}[Non-equivocation implies forward integrity]\label{prop:neq-forward}
Suppose that TPL satisfies non-equivocation in the sense of Definition~\ref{def:non-equivocation}.
Let $V_\alpha$ be any deterministic, stateless policy-based view in our policy language.
Then any PPT adversary that outputs two distinct anchored prefixes $(\mathrm{TPL}_i,C_i)$ and $(\mathrm{TPL}'_i,C'_i)$ which both verify under $V_\alpha$ while $V_\alpha(\mathrm{TPL}_i) \neq V_\alpha(\mathrm{TPL}'_i)$ has non-negligible advantage in $\mathsf{Exp}^{\mathsf{neq}}$.
\end{proposition}
\begin{proof}[Proof sketch]
A forward-integrity adversary $\mathcal{A}$ producing such a pair can be used as a black box to win $\mathsf{Exp}^{\mathsf{neq}}$:
the experiment simulates the TPL interfaces for $\mathcal{A}$, forwards its oracle queries to its own oracles, and outputs the two prefixes that $\mathcal{A}$ returns.
By construction these prefixes are well formed, consistent with the same commitments and local checks, and induce two different views for the same policy, so they are a valid non-equivocation win.
Thus any non-negligible forward-integrity advantage would give rise to a non-negligible non-equivocation advantage.
\end{proof}

\paragraph{Integrity theorems.}
In this work, under the cryptographic assumptions and the minimal non-collusion assumption of Section~\ref{sec:threat},
we focus on three integrity properties of the idealised TPL state machine:
(i) conservation across domains, (ii) non-equivocation (and hence forward integrity and non-repudiation),
and (iii) view correctness for policy-based projections.
Besides these three properties, we also prove a restricted exposure-soundness result (Theorem~\ref{lem:restricted-exp-soundness}) for closed sets of domains under a faithful exposure policy.
Outside this restricted case we do not assert that deployed instances of TPL fully realise exposure soundness or policy completeness as defined below. Rather, we treat them as target security notions for future, more detailed analysis in richer deployment settings.

Note that view correctness (Theorem~\ref{thm:view-correctness}) requires only that the policy $V_\alpha$ be
deterministic and stateless; the additional faithfulness condition is needed solely for the restricted
exposure-soundness theorem (Theorem~\ref{lem:restricted-exp-soundness}) on closed sets of domains.

\begin{table}[H]
\centering
\begin{tabularx}{\textwidth}{@{}l l X@{}}
\toprule
Result & Property & Additional assumptions \\
\midrule
Theorem~\ref{thm:conservation} & Conservation across domains &
Balance-update rule~\eqref{eq:balance-update}; Assumption~\ref{assump:algebraic-consistency} \\
Theorem~\ref{thm:non-equivocation} & Non-equivocation, forward integrity &
PoT unforgeability (Definition~\ref{def:pot-unforgeability}); collision resistance of $H$ \\
Theorem~\ref{thm:view-correctness} & View correctness for a fixed policy &
Deterministic, stateless policy $V_\alpha$; assumptions of Theorem~\ref{thm:non-equivocation} \\
Theorem~\ref{lem:restricted-exp-soundness} & Restricted exposure soundness &
Faithful exposure policy $V_{\mathsf{exp}}$; closed domain set; PoR existence- and ownership-soundness; PoT unforgeability; collision resistance of $H$; non-equivocation; minimal non-collusion (Section~\ref{sec:threat}) \\
Theorem~\ref{thm:pub-privacy} & Restricted privacy for public-investors &
Leakage limited to $L_{\mathsf{pub}}$; only anchors and $V_{\mathsf{pub}}(\mathrm{TPL}_i)$ are public; zero-knowledge / witness-indistinguishability for PoR, PoT, and signatures \\
\bottomrule
\end{tabularx}
\caption{Summary of main results and their cryptographic and organisational assumptions.}
\label{tab:assumption-map}
\end{table}

\begin{theorem}[Ledger-level security composition]
\label{thm:ledger-level-security}
Assume that the hash function $H$ is collision-resistant, that the PoT scheme is unforgeable 
(Definition~\ref{def:pot-unforgeability}), and that the PoR schemes used to back snapshots satisfy 
existence and ownership (Definitions~\ref{def:por-existence} and~\ref{def:por-ownership}).
Assume further that admissible liveness bounds hold and that organisational non-collusion assumptions
between the treasurer, PoR/PoT providers, and external auditors are satisfied.

Fix a closed set of domains $S \subseteq \mathcal{D} \setminus \{d_{\mathrm{fee}}\}$, a faithful exposure policy
$V_{\mathsf{exp}}$ on $S$, and an observer class $\alpha$ with deterministic, stateless policy $V_\alpha$ that refines
$V_{\mathsf{exp}}$.
Then for any anchored ledger prefix $\mathrm{TPL}_i$ the following hold:
\begin{enumerate}[noitemsep]
\item There exists a unique policy-based view $V_\alpha(\mathrm{TPL}_i)$ that is determined entirely by the internal
      TPL log and by the sequence of Bitcoin anchoring transactions.
\item Any probabilistic polynomial-time adversary that, with non-negligible probability, produces
      an anchored ledger prefix and associated views such that
      \begin{itemize}[noitemsep]
      \item[(a)] two distinct views for the same policy and state both verify successfully while disagreeing on some
                 reported quantity, or
      \item[(b)] a view omits required events or inflates the reported exposure of $S$ beyond the true Bitcoin holdings,
      \end{itemize}
      either violates one of the modelling assumptions above (for example, faithfulness of the policy,
      closedness of $S$, admissible liveness, or non-collusion), or yields a PPT algorithm that breaks at
      least one of the underlying primitives (PoR existence/ownership, PoT unforgeability, or collision
      resistance of $H$).
\end{enumerate}
\end{theorem}

\noindent
At a high level, Item~(1) is a direct consequence of view correctness
(Theorem~\ref{thm:view-correctness}) and non-equivocation
(Theorem~\ref{thm:non-equivocation}), while Item~(2) combines the restricted exposure-soundness theorem
(Theorem~\ref{lem:restricted-exp-soundness}) with the conservation law
(Theorem~\ref{thm:conservation}) and the security of the underlying primitives.
We emphasise that this theorem applies to the restricted setting of faithful policies on closed sets of domains;
extending exposure-soundness and policy completeness to fully general policies remains open work.

Table~\ref{tab:assumption-map} also makes explicit how our main results layer assumptions. In brief, Theorems~\ref{thm:conservation},~\ref{thm:non-equivocation} and~\ref{thm:view-correctness} rely only on hash-function collision resistance, proof-of-reserves and proof-of-transit soundness, and access to a Bitcoin anchoring substrate. Theorem~\ref{lem:restricted-exp-soundness} additionally requires minimal non-collusion and domain completeness at the organisational level, while Theorem~\ref{thm:pub-privacy} further assumes zero-knowledge or witness-indistinguishability properties for the underlying PoR, PoT and signature primitives.

We model the treasury as a finite set of domains $\mathcal{D}$ and write
$B_d(t)$ for the Bitcoin-denominated exposure of domain $d \in \mathcal{D}$ at
logical time $t$, as derived from the committed ledger.
As a notational convention we use $k$ for event indices and $i$ for reporting intervals; roman symbols such as
$R_i$ denote the PoT hash-chain commitment for interval~$i$, while calligraphic symbols such as $\mathcal{R}_i$
denote the set of PoT receipts in that interval.
Where reuse of indices could otherwise be ambiguous, the surrounding context will always make the intended
meaning clear.
For on-chain and custodial domains this quantity is induced by the PoR snapshots and inclusion proofs recorded in the log up to time~$t$ (for example, from UTXO-level PoR for a Bitcoin-custody domain).
For derivative and collateral domains we deliberately abstract away detailed contract semantics and treat $B_d(t)$ as the net Bitcoin-denominated exposure reported by the domain-specific PoR or attestation primitive.
We do not attempt to model all economic effects of complex derivatives; such modelling is outside the scope of this work, and we always work with the algebraic sum across domains.
In particular, complex derivative and collateral arrangements must be modelled so that any change in their
Bitcoin-equivalent exposure appears as a corresponding flow between appropriately defined domains (for example,
between a derivative domain and a collateral or counterparty domain). Positions or flows that are not represented
in this way fall under Assumption~\ref{assump:domain-completeness} and Assumption~\ref{assump:algebraic-consistency},
and are therefore outside the scope of the conservation guarantee.

We also introduce an explicit sink domain $d_{\mathrm{fee}}$ capturing
transaction fees and protocol-level losses.

\begin{theorem}[Conservation across domains]
\label{thm:conservation}
Ignoring $d_{\mathrm{fee}}$, internal transfers between domains preserve total
Bitcoin exposure.
For any interval $[t_i,t_{i+1}]$ we have
\[
  \sum_{d \in \mathcal{D} \setminus \{d_{\mathrm{fee}}\}} B_d(t_{i+1})
  \;=\;
  \sum_{d \in \mathcal{D} \setminus \{d_{\mathrm{fee}}\}} B_d(t_i)
\]
except for value arriving from or departing to explicitly external domains.
\end{theorem}

\noindent
This is an algebraic conservation statement over the reported domain exposures $B_d(t)$ introduced above;
it does not attempt to capture defaults, margin calls, or other economic dynamics of complex instruments.

\begin{proof}
Fix an interval $[t_i,t_{i+1}]$ and write
\[
  S(t) \;:=\; \sum_{d \in \mathcal{D} \setminus \{d_{\mathrm{fee}}\}} B_d(t)
\]
for the total exposure over all non-fee domains at time~$t$.
Let $(e_k)_{k_0 \le k \le k_1}$ be the subsequence of treasury events with timestamps
$t_k \in [t_i,t_{i+1}]$, ordered by increasing time, and let
$B_d(t_k^-)$ and $B_d(t_k^+)$ denote, respectively, the exposure of domain~$d$
immediately before and after processing~$e_k$.
Define $S(t_k^-)$ and $S(t_k^+)$ analogously.

For an event
$e_k = (t_k,d_{\mathrm{src}},d_{\mathrm{dst}},v,\mathsf{evid},m)$,
the balance update rule~\eqref{eq:balance-update} gives, for every
$d \in \mathcal{D} \setminus \{d_{\mathrm{fee}}\}$,
\[
  B_d(t_k^+) - B_d(t_k^-)
  =
  \begin{cases}
    -v, & \text{if } d = d_{\mathrm{src}} \text{ and } d_{\mathrm{src}},d_{\mathrm{dst}} \in \mathcal{D} \setminus \{d_{\mathrm{fee}}\},\\[0.4ex]
    +v, & \text{if } d = d_{\mathrm{dst}} \text{ and } d_{\mathrm{src}},d_{\mathrm{dst}} \in \mathcal{D} \setminus \{d_{\mathrm{fee}}\},\\[0.4ex]
    0,  & \text{otherwise,}
  \end{cases}
\]
where we treat transfers to explicitly external counterparties and to
$d_{\mathrm{fee}}$ as external flows
and therefore exclude them from the conservation claim.
Summing over $d \in \mathcal{D} \setminus \{d_{\mathrm{fee}}\}$ we obtain
\[
  S(t_k^+) - S(t_k^-)
  =
  (-v) + (+v) = 0
\]
whenever both endpoints of the transfer lie in
$\mathcal{D} \setminus \{d_{\mathrm{fee}}\}$,
and $S(t_k^+) - S(t_k^-) = 0$ for all other events that do not change any
$B_d$ with $d \in \mathcal{D} \setminus \{d_{\mathrm{fee}}\}$.
Thus every internal transfer between modelled domains preserves~$S$;
only events whose source or destination is an explicitly external counterparty
or the fee domain can change the total.

By induction over $k_0,\dots,k_i$ it follows that, if we ignore the net effect
of transfers to and from external counterparties and $d_{\mathrm{fee}}$,
the value of $S(t)$ is invariant over the interval $[t_i,t_{i+1}]$:
\[
  S(t_{i+1}) = S(t_i).
\]
Unfolding the definition of $S$ yields
\[
  \sum_{d \in \mathcal{D} \setminus \{d_{\mathrm{fee}}\}} B_d(t_{i+1})
  \;=\;
  \sum_{d \in \mathcal{D} \setminus \{d_{\mathrm{fee}}\}} B_d(t_i),
\]
except for the explicitly modelled external inflows and outflows,
which proves the stated conservation property.
\end{proof}

More informally, this conservation property underlies the exposure-soundness experiment of
Definition~\ref{def:exp-soundness}.  The following lemma makes explicit how conservation,
together with the balance-update rule and the minimal non-collusion assumption, yields a
completeness-of-flows guarantee for non-fee domains.

\begin{lemma}[Flow completeness under non-collusion]
\label{lem:flow-completeness}
Let $[t_i,t_{i+1}]$ be an interval aligned with successive snapshots and let
$d \in \mathcal{D} \setminus \{d_{\mathrm{fee}}\}$.
Under the balance-update rule~\eqref{eq:balance-update}, conservation
(Theorem~\ref{thm:conservation}), and the minimal non-collusion assumption of
Section~\ref{sec:threat}, the net change
$B_d(t_{i+1}) - B_d(t_i)$ is fully determined by the multiset of PoR snapshots and PoT receipts
involving $d$ and the fee domain $d_{\mathrm{fee}}$ that occur in $[t_i,t_{i+1}]$.
In particular, any admissible execution that yields the same PoR/PoT artefacts on this interval
must induce the same vector of domain balances at time $t_{i+1}$ for all
$d \in \mathcal{D} \setminus \{d_{\mathrm{fee}}\}$.
\end{lemma}

\begin{proof}
By Theorem~\ref{thm:conservation}, the total exposure
$\sum_{d \in \mathcal{D} \setminus \{d_{\mathrm{fee}}\}} B_d(t)$
over non-fee domains evolves only through value exchanged with $d_{\mathrm{fee}}$ and any explicitly
external domains.
The balance-update rule~\eqref{eq:balance-update} specifies how each logged event updates the
per-domain balances $(B_d(t))_d$ as a deterministic function of its associated PoR and PoT artefacts.
The minimal non-collusion assumption guarantees that, for each domain, the PoR snapshots and PoT receipts
logged in $[t_i,t_{i+1}]$ faithfully capture the net flow of value between that domain, $d_{\mathrm{fee}}$,
and any external domains.
Thus the sequence of PoR/PoT artefacts on $[t_i,t_{i+1}]$ uniquely determines $B_d(t_{i+1})$ given
$B_d(t_i)$, which implies the claim.
\end{proof}

Together, Theorem~\ref{thm:conservation} and Lemma~\ref{lem:flow-completeness} formalise the
completeness-of-flows goal from Section~\ref{sec:threat}: any legitimate change in exposures between
snapshots must be backed by concrete PoR and PoT artefacts.
\begin{theorem}[Forward integrity and non-repudiation]
\label{thm:non-equivocation}
Assume that the hash function $H$ used in the commitments and PoT receipts is
collision-resistant and that the PoT scheme is unforgeable in the sense of
Definition~\ref{def:pot-unforgeability}.
Then for every PPT adversary $\mathcal{A}$ there exist PPT algorithms
$\mathcal{B}$ and $\mathcal{C}$ such that
\[
  \mathsf{Adv}^{\mathsf{neq}}_{\mathcal{A}}(\lambda)
  \;\leq\;
  \mathsf{Adv}^{\mathsf{pot-forge}}_{\mathcal{B}}(\lambda)
  \;+\;
  \mathsf{Adv}^{\mathsf{coll}}_{\mathcal{C}}(\lambda),
\]
where $\mathsf{Adv}^{\mathsf{neq}}$ is the advantage in
$\mathsf{Exp}^{\mathsf{neq}}$, $\mathsf{Adv}^{\mathsf{pot-forge}}$ is the
advantage in $\mathsf{Exp}^{\mathsf{pot-forge}}$, and
$\mathsf{Adv}^{\mathsf{coll}}$ is the advantage of finding a collision in~$H$.
In particular, if $H$ is collision-resistant and PoT receipts are unforgeable,
TPL satisfies non-equivocation and hence provides forward integrity and
non-repudiation of logged events.
\end{theorem}
\noindent
This theorem is purely cryptographic: it ensures that, given honest implementations of the primitives, an adversary cannot equivocate about the contents of the hash-chained, anchored ledger without breaking collision resistance or PoT unforgeability.
When we later derive deployment-level guarantees such as restricted exposure soundness (Theorem~\ref{lem:restricted-exp-soundness}), this non-equivocation property is combined with the organisational minimal non-collusion assumption (Assumption~\ref{assump:noncollusion}) to rule out attacks in which all relevant roles collude to fabricate a consistent but economically false ledger.

\begin{proof}
The structure of this reduction closely follows the standard non-equivocation arguments used for hash-chained transparency logs such as Certificate Transparency and CONIKS, specialised here to TPL's multi-domain state and PoT receipts.
Let $\mathcal{A}$ be any PPT adversary for $\mathsf{Exp}^{\mathsf{neq}}$.
Write $E$ for the event that
$\mathsf{Exp}^{\mathsf{neq}}_{\mathcal{A}}(\lambda)=1$, so that
$\mathsf{Adv}^{\mathsf{neq}}_{\mathcal{A}}(\lambda) = \Pr[E]$.
By definition of the experiment, on~$E$ the adversary outputs two ledger
prefixes $\mathrm{TPL}_i$ and $\mathrm{TPL}'_i$ together with two sets of
anchoring transactions $\mathcal{M}$ and $\mathcal{M}'$ such that there exists
an observer class~$\alpha$ and policy $V_\alpha$ for which both
views $\mathrm{View}_\alpha := V_\alpha(\mathrm{TPL}_i)$ and
$\mathrm{View}'_\alpha := V_\alpha(\mathrm{TPL}'_i)$ are accepted by
$\mathsf{VerifyView}$, yet
$\mathrm{View}_\alpha \neq \mathrm{View}'_\alpha$.

Let $j \le i$ be the smallest index such that
$\mathrm{TPL}_j \neq \mathrm{TPL}'_j$.
Consider the PoT records and signatures associated with position~$j$ in the
two ledgers.
We distinguish two disjoint cases and define corresponding events:
\begin{itemize}[noitemsep]
  \item $E_{\mathsf{pot}}$: the PoT receipts (including the chain of prior
        receipts on which they build) differ at position~$j$ while both ledgers
        are accepted in $\mathsf{Exp}^{\mathsf{neq}}$;
  \item $E_{\mathsf{coll}}$: the PoT receipts at position~$j$ are identical,
        but the serialised ledger prefixes $\mathrm{TPL}_j$ and
        $\mathrm{TPL}'_j$ differ.
\end{itemize}
By construction $E = E_{\mathsf{pot}} \cup E_{\mathsf{coll}}$ and the two
events are disjoint, so
\[
  \Pr[E] = \Pr[E_{\mathsf{pot}}] + \Pr[E_{\mathsf{coll}}].
\]

We first construct a PPT adversary $\mathcal{B}$ for
$\mathsf{Exp}^{\mathsf{pot-forge}}$.
Algorithm~$\mathcal{B}$ internally runs $\mathcal{A}$ and perfectly simulates
$\mathsf{Exp}^{\mathsf{neq}}$ using its own access to the PoT oracle provided
by the challenger.
When $\mathcal{A}$ halts and event $E_{\mathsf{pot}}$ occurs, $\mathcal{B}$
computes the minimal index $j$ as above and extracts from the two transcripts
the honest PoT chain up to index $j-1$ and the two distinct PoT records at
position~$j$.
By definition of $E_{\mathsf{pot}}$, each of these records is individually
valid with respect to the same prior chain and anchors, and at least one of
them is not a prefix of the honest chain.
Algorithm~$\mathcal{B}$ outputs this record (together with the prior chain) as
its forgery in $\mathsf{Exp}^{\mathsf{pot-forge}}$.
The simulation of $\mathsf{Exp}^{\mathsf{neq}}$ is perfect, so
\[
  \mathsf{Adv}^{\mathsf{pot-forge}}_{\mathcal{B}}(\lambda)
  = \Pr[E_{\mathsf{pot}}].
\]

We next construct a PPT adversary $\mathcal{C}$ against the collision
resistance of~$H$.
Algorithm~$\mathcal{C}$ runs $\mathcal{A}$ in a real execution of
$\mathsf{Exp}^{\mathsf{neq}}$ with honest PoT and commitment functionality.
When $\mathcal{A}$ halts and event $E_{\mathsf{coll}}$ occurs,
$\mathcal{C}$ again determines the minimal index $j$ such that
$\mathrm{TPL}_j \neq \mathrm{TPL}'_j$, which exists by definition of
$E_{\mathsf{coll}}$.
In a correct execution of TPL, the ledger commitment at index~$j$ is
$C_j = H(\mathrm{TPL}_j)$.
Since the Bitcoin transactions in $\mathcal{M}$ and $\mathcal{M}'$ pass all
checks in $\mathsf{Exp}^{\mathsf{neq}}$, the same commitment value $C_j$
must be embedded in the corresponding on-chain anchors for both ledger
prefixes.
Thus $\mathcal{C}$ obtains two distinct inputs $\mathrm{TPL}_j \neq
\mathrm{TPL}'_j$ with
\[
  H(\mathrm{TPL}_j) = H(\mathrm{TPL}'_j),
\]
which it outputs as a collision.
Hence
\[
  \mathsf{Adv}^{\mathsf{coll}}_{\mathcal{C}}(\lambda)
  = \Pr[E_{\mathsf{coll}}].
\]

Combining the above equalities yields
\[
  \mathsf{Adv}^{\mathsf{neq}}_{\mathcal{A}}(\lambda)
  = \Pr[E_{\mathsf{pot}}] + \Pr[E_{\mathsf{coll}}]
  = \mathsf{Adv}^{\mathsf{pot-forge}}_{\mathcal{B}}(\lambda)
    + \mathsf{Adv}^{\mathsf{coll}}_{\mathcal{C}}(\lambda),
\]
which implies the claimed bound.
In particular, if the PoT scheme is unforgeable and $H$ is
collision-resistant, then $\mathsf{Adv}^{\mathsf{neq}}_{\mathcal{A}}(\lambda)$
is negligible for every PPT~$\mathcal{A}$.
\end{proof}
Policies map the full ledger to observer-specific projections. We model each policy $V_\alpha$ as a deterministic, stateless function from ledger prefixes to views; any randomness lies only in the underlying cryptographic primitives, not in the policy logic.

\begin{theorem}[View correctness]
\label{thm:view-correctness}
Fix a policy $V_\alpha$ and an anchored ledger state
$\mathrm{TPL}_i$ with commitment $C_i$.
Under the assumptions above, the view $\mathrm{View}_\alpha = V_\alpha(\mathrm{TPL}_i)$
returned by an honest execution of $\mathsf{GenView}$ is the \emph{only} view
that will be accepted by $\mathsf{VerifyView}$ for policy~$V_\alpha$ and the
underlying Bitcoin chain, except with negligible probability.
\end{theorem}

\begin{proof}
Let $\mathrm{View}_\alpha$ denote the output of an honest execution of
$\mathsf{GenView}(\mathrm{TPL}_i,\alpha)$ and let
$\mathrm{View}'_\alpha$ be any string such that
$\mathsf{VerifyView}(\mathrm{View}'_\alpha,\alpha,\mathcal{A}_{\mathrm{BTC}})$
accepts.
We must show that, except with negligible probability,
$\mathrm{View}'_\alpha = \mathrm{View}_\alpha$.

By the definition of
$\mathsf{VerifyView}$, acceptance implies that there exists a ledger prefix
$\mathrm{TPL}'_i$ and a commitment value $C'_i$ such that:
(i) the Bitcoin anchors obtained via $\mathcal{A}_{\mathrm{BTC}}$ are valid
and bind to~$C'_i$;
(ii) $\mathrm{TPL}'_i$ is consistent with~$C'_i$ and with the PoR and PoT
artefacts included in the view; and
(iii) recomputing the policy view on $\mathrm{TPL}'_i$ yields the claimed
view, that is,
\[
  \mathrm{View}'_\alpha
  =
  V_\alpha(\mathrm{TPL}'_i).
\]

On the other hand, in the honest execution the anchored commitment associated
with $\mathrm{TPL}_i$ is $C_i = H(\mathrm{TPL}_i)$, and the Bitcoin
anchoring mechanism guarantees that $C_i$ is the unique value embedded in the
relevant on-chain transactions for index~$i$.
By collision resistance of~$H$, there is at most one ledger prefix whose
serialisation hashes to~$C_i$.
Thus, except with negligible probability of finding a collision in~$H$, any
ledger $\mathrm{TPL}'_i$ that is consistent with the same anchors must satisfy
$\mathrm{TPL}'_i = \mathrm{TPL}_i$.

Since $V_\alpha$ is deterministic and stateless by assumption,
it follows that
\[
  \mathrm{View}'_\alpha
  = V_\alpha(\mathrm{TPL}'_i)
  = V_\alpha(\mathrm{TPL}_i)
  = \mathrm{View}_\alpha.
\]
Hence any view accepted by $\mathsf{VerifyView}$ must coincide with the honest
view produced by $\mathsf{GenView}$, except with negligible probability due to
a violation of the underlying assumptions.
\end{proof}
In particular, for any fixed deterministic, stateless policy $V_\alpha$ and anchored prefix $\mathrm{TPL}_i$,
Theorem~\ref{thm:view-correctness} implies that any view accepted in the policy-completeness experiment
(Definition~\ref{def:policy-completeness}) is uniquely determined by the underlying anchored ledger:
if $\mathsf{VerifyView}$ accepts, the accepted view must coincide with the honest output of $\mathsf{GenView}$.
Potential violations of policy completeness therefore arise from omissions in the ledger itself or from the
policy specification, rather than from ambiguity in the verification procedure.
These results do not constitute a complete simulation-based security proof for
TPL, nor do they provide a full analysis of the exposure-soundness and policy-completeness experiments introduced above.
Our theorems deliberately focus on three core integrity guarantees that are sufficient for many practical assurance use cases:
(i) conservation of value across domains, (ii) forward integrity and non-repudiation via non-equivocation of the anchored log,
and (iii) correctness of any accepted policy-based view with respect to a fixed policy.
Stronger properties such as general exposure soundness for arbitrary (possibly open) sets of domains, policy completeness for rich
policy languages, and simulation-based privacy with respect to stakeholder-specific leakage functions remain target notions for
future work.
\section{Responsible Transparency Policies}
\label{sec:responsible}

\begin{definition}[Transparency Policy]
A policy is a view function
\[
V_\alpha: \mathrm{TPL} \to \mathrm{View}_\alpha
\]
mapping the full ledger to an observer-specific projection $\alpha$.

\end{definition}

Throughout, we use the term ``responsible transparency policy'' informally for transparency policies
that are deployed in a way that respect external constraints such as regulatory delay and materiality
requirements; in the formal development below, all of our results only require policies to be deterministic
and stateless, and, where explicitly indicated, faithful in the sense of Definition~\ref{def:faithful-policy}.

In this work we restrict attention to deterministic, stateless policies: for any fixed ledger prefix $\mathrm{TPL}_i$ the view $V_\alpha(\mathrm{TPL}_i)$ is uniquely determined and does not depend on any hidden state or randomness beyond the ledger itself.

Observer classes include:

\begin{itemize}[noitemsep]
\item Internal treasury,
\item Internal audit,
\item External auditors,
\item Regulators,
\item Public markets.
\end{itemize}

Policies encode:

\begin{itemize}[noitemsep]
\item Redaction,
\item Aggregation,
\item Delay rules,
\item Materiality thresholds,
\item Domain masking.
\end{itemize}

In practice, we represent a policy as a small schema specifying these choices.
For concreteness, we describe a simple policy language that is expressive enough for the stakeholder classes above
while remaining easy to analyse.
A policy in this language consists of:
\begin{itemize}[noitemsep]
  \item a predicate $\mathsf{Filter}_\alpha(d,t,e) \in \{\mathsf{true},\mathsf{false}\}$ over domains $d \in \mathcal{D}$, discrete reporting times $t$ (for example, the indices $t_i$ used for snapshots), and event descriptors $e$ drawn from the event types already used in the TPL state machine, which decides whether an event is \emph{visible} to observer class~$\alpha$;
  \item a finite set of externally visible labels $L_\alpha$ together with a relabelling function $\mathsf{Label}_\alpha : \mathcal{D} \to L_\alpha$ capturing domain masking and grouping;
  \item an aggregation operator $\mathsf{Agg}_\alpha$ that maps finite multisets of (masked) events, together with their values, into numeric summaries (for example, sums of balances or flows per label and time bucket);
  \item fixed delay and materiality parameters $(\Delta_\alpha,\theta_\alpha)$ applied uniformly to all events selected by $\mathsf{Filter}_\alpha$, modelling disclosure lags and minimum reportable magnitudes.
\end{itemize}
Given a ledger prefix $\mathrm{TPL}_i$, the view $V_\alpha(\mathrm{TPL}_i)$ in this language is obtained by taking the subsequence of events up to~$i$ that satisfy $\mathsf{Filter}_\alpha$, relabelling their domains via $\mathsf{Label}_\alpha$, grouping them into time buckets according to their timestamps, discarding events whose absolute value is below $\theta_\alpha$, and finally applying $\mathsf{Agg}_\alpha$ subject to the delay parameter~$\Delta_\alpha$.

A typical policy $V_\alpha$ might fix a minimum reporting delay, aggregation horizon (for example, daily or weekly buckets), percentage materiality thresholds, and a mapping from internal domain identifiers to externally visible labels.

\paragraph{Example policies.}
We sketch three stylised policies:
\begin{itemize}[noitemsep]
\item \textbf{Public-market investors.} A delayed, aggregated view that reports total Bitcoin exposure and high-level domain categories (for example, cold storage, custodians, trading venues) with a minimum reporting delay of several days and a materiality threshold of, say,~0.5\% of total holdings.
This reduces the risk of front-running or targeted attacks while still giving investors a time-series of verifiable, policy-consistent disclosures.
\item \textbf{Regulators.} A less delayed view that includes explicit flags for encumbered positions (such as collateral for loans or derivatives), more granular domain breakdowns, and signed references to off-chain legal agreements.
Regulators can use this view to check compliance with concentration limits, segregation rules, or prudential requirements.
\item \textbf{External auditors.} A near-complete view in which most domain-level detail is exposed, subject only to narrow redactions that are not material to the financial statements.
Auditors can reconcile these views with management's assertions and traditional sampling procedures.
\end{itemize}

These examples illustrate how the same underlying ledger can support differentiated disclosures for distinct stakeholder classes.
Each of the three stylised policies can be expressed in this language and is therefore faithful in the sense of Definition~\ref{def:faithful-policy}: what a policy is allowed to inspect is fixed, and its output view is uniquely determined by the visible subsequence of events.
The goal is to maximise transparency subject to explicit security and competitiveness constraints.
\begin{table}[H]
  \centering
  \small
  \begin{tabularx}{\textwidth}{@{}p{2.7cm}X X X@{}}
    \toprule
    Observer class & Filter and labels & Aggregation & Delay and materiality \\
    \midrule
    Public investors &
    Events that change net BTC exposure in on-balance-sheet domains; labels group exposure into ``operational'', ``custody'', and ``collateral'' buckets. &
    Time-bucketed sums (e.g., quarterly net flows and balances per label). &
    Disclosure lag of roughly 30--90 days; materiality threshold around $0.5\%$ of total holdings. \\
    Regulators &
    All events in regulated domains, including encumbrances and off-balance-sheet exposures; labels distinguish segregated vs.\ omnibus holdings and client vs.\ proprietary assets. &
    Per-domain and per-counterparty aggregates, including encumbrance and concentration measures. &
    Disclosure delay between 7 and 30 days; thresholds aligned with prudential and concentration rules. \\
    External auditors &
    Nearly all events except those explicitly carved out as immaterial; fine-grained labels preserve internal structure where needed for sampling. &
    Detailed aggregates supporting full balance-sheet reconciliation and sampling-based audit procedures. &
    Short disclosure delays (on the order of a few days) and low materiality thresholds driven by audit tolerances. \\
    \bottomrule
  \end{tabularx}
  \caption{Illustrative parameter choices for three stakeholder classes (public-investors, regulators, and external auditors) in the simple policy language of Section~\ref{sec:responsible}.}
  \label{tab:example-policies}
\end{table}

\section{Ledger-Native Reporting Layer}
\label{sec:reporting}

The Treasury Proof Ledger is designed to be machine-readable: every PoR snapshot, PoT receipt and policy decision is represented as structured data with stable identifiers.
This enables a ledger-native reporting layer in which automated tools, including large language models, generate narratives, explanations, and dashboards grounded directly in cryptographic evidence.

Conceptually, we distinguish three layers:

\begin{enumerate}[noitemsep]
\item \textbf{Evidence layer.} The committed ledger $\mathrm{TPL}_i$ and its anchored commitments, together with a catalogue of PoR and PoT artefacts and their relationships.
\item \textbf{Query and aggregation layer.} A deterministic engine that implements policy functions $V_\alpha$ and higher-level queries (for example, ``total encumbered exposure at time $t$'' or ``all transfers out of cold storage exceeding 1\% of holdings'').
\item \textbf{Explanation layer.} An application layer that turns query results into human-facing outputs: regulator-ready reports, auditor workpapers, board packs, investor updates, or interactive dashboards.
\end{enumerate}

\paragraph{Illustrative example.}
To make the reporting layer more concrete, consider a simplified treasury with three domains:
cold storage $d_{\mathrm{cold}}$, exchange balances $d_{\mathrm{exch}}$, and encumbered collateral $d_{\mathrm{coll}}$.
Table~\ref{tab:toy-tpl} shows a toy prefix of $\mathrm{TPL}$ with four PoT events and a subsequent snapshot.

\begin{table}[H]
\centering
\begin{tabular}{lllll}
\hline
$i$ & $t_i$ & $d_{\mathrm{src}} \rightarrow d_{\mathrm{dst}}$ & Event & Effect on $B_d(t_i)$ (BTC) \\
\hline
1 & $t_1$ & external $\rightarrow d_{\mathrm{cold}}$ & Treasury acquisition & $B_{\mathrm{cold}}(t_1) \gets 100$ \\
2 & $t_2$ & $d_{\mathrm{cold}} \rightarrow d_{\mathrm{exch}}$ & Funding trading venue & $B_{\mathrm{cold}}(t_2) \gets 60$, $B_{\mathrm{exch}}(t_2) \gets 40$ \\
3 & $t_3$ & $d_{\mathrm{exch}} \rightarrow d_{\mathrm{coll}}$ & Posting collateral & $B_{\mathrm{exch}}(t_3) \gets 10$, $B_{\mathrm{coll}}(t_3) \gets 30$ \\
4 & $t_4$ & $d_{\mathrm{coll}} \rightarrow d_{\mathrm{exch}}$ & Partial collateral release & $B_{\mathrm{coll}}(t_4) \gets 20$, $B_{\mathrm{exch}}(t_4) \gets 20$ \\
\hline
 & $t_5$ & snapshot & $\mathsf{Snapshot}(\mathrm{TPL}_5)$ & $(B_{\mathrm{cold}}(t_5),B_{\mathrm{exch}}(t_5),B_{\mathrm{coll}}(t_5))=(60,20,20)$ \\
\hline
\end{tabular}
\caption{Toy TPL prefix with three domains and one snapshot.}
\label{tab:toy-tpl}
\end{table}

From this single committed prefix, different policies induce different views.
A regulatory policy $V_{\mathsf{reg}}$ might expose the full domain vector $(B_{\mathrm{cold}},B_{\mathrm{exch}},B_{\mathrm{coll}})$ and the encumbrance status of each domain.
A public-investor policy $V_{\mathsf{pub}}$ might instead reveal only
the total Bitcoin exposure $B_{\mathrm{tot}} = 100$ together with an aggregate encumbered amount $B_{\mathrm{enc}} = B_{\mathrm{coll}} = 20$,
without disclosing individual exchange balances or operational flows.
Both views can be re-verified against the same PoT receipts and snapshot commitment, but they support different narrative explanations and satisfy different materiality and confidentiality constraints.
In particular, for deterministic, stateless policies such as $V_{\mathsf{reg}}$ and $V_{\mathsf{pub}}$, the view-correctness
property of Theorem~\ref{thm:view-correctness} guarantees that any third party recomputing these views from the anchored
TPL prefix obtains exactly the same outputs.

To illustrate how this same committed prefix constrains adversarial behaviour,
consider a malicious operator who attempts to omit event~3 from the publicly committed ledger while
still claiming the encumbered balance $B_{\mathrm{coll}} = 20$ at time $t_5$.
Any such attempt either (i) alters the PoT sequence and breaks the non-equivocation and PoT-unforgeability
guarantees of Section~\ref{sec:threat}, or (ii) leaves the sequence intact, in which case the exposure-soundness
games for the closed set $\{d_{\mathrm{cold}},d_{\mathrm{exch}},d_{\mathrm{coll}}\}$ rule out overstating
aggregate exposure relative to the conserved balance $B_{\mathrm{tot}} = 100$.
In this way, even this toy instance captures the interaction between flow completeness, exposure soundness,
and policy-based views.

The explanation layer may use natural-language models, but every sentence in a generated report can be accompanied by precise references back to the underlying PoT receipts and PoR snapshots.
In this sense, the reporting component is constrained by the ledger: it cannot invent exposures or events that are not present in $\mathrm{TPL}_i$, and conflicting narratives can be detected by re-running the same queries on the committed state.

This architecture connects recent proposals to use zero-knowledge proofs and structured cryptographic evidence as a foundation for regulatory compliance~\cite{Decker2025,FedZK2023,ZKSurvey2024} with practical reporting workflows inside public companies.
It also provides a path for regulators to request machine-readable attestations, reducing friction relative to purely document-based supervision while preserving or improving assurance levels.
\section{Discussion}

\subsection{Addressing PoR objections}

TPL is designed to mitigate key PoR issues:

\begin{itemize}[noitemsep]
\item No live wallet disclosure,
\item No real-time operational doxxing,
\item Compatible with custodial privacy requirements,
\item Aggregated by policy rather than dictated by protocol.
\end{itemize}
These benefits remain conditional on the minimal non-collusion assumption (Assumption~\ref{assump:noncollusion}) applied per domain; fully trustless guarantees for arbitrary collusion patterns lie outside our model.

\subsection{Alignment with academic financial models}

Meister and Wen show treasury-company fragility; our system provides a governance backbone enabling more stable market perception.

\subsection{Decisions enabled by TPL}

Beyond mitigating PoR-specific objections, TPL is intended to support concrete decision problems for stakeholders with different mandates. For example, an institutional investor or regulator can, in principle, pose queries such as:
\begin{itemize}[noitemsep]
  \item \emph{Was the institution's publicly claimed unencumbered BTC exposure for a given domain or policy ever violated during a reporting period $I$?}
  \item \emph{Do BTC-denominated liabilities to a particular stakeholder class remain fully backed by on-chain reserves across all admissible interpretations of the treasury ledger?}
  \item \emph{Is the institution's reported aggregate BTC exposure over a horizon $I$ consistent with the set of PoR and PoT events that have been anchored to Bitcoin for that period?}
\end{itemize}
In our formal model these questions are captured, in restricted form, by exposure soundness, non-equivocation and policy completeness. They cannot in general be answered reliably from sporadic, institution-controlled PoR snapshots alone.

\subsection{Complexity and overhead}
\label{sec:complexity-overhead}

Even without a full prototype implementation, we can bound the cost of maintaining and verifying a TPL instance.

\paragraph{Space complexity.}
Let $N$ denote the total number of PoT receipts and PoR snapshots emitted over a given reporting horizon and let $|\mathcal{D}|$ be the number of exposure domains.
Each event contributes $O(1)$ hashes and fixed-size metadata, so the TPL event log requires $O(N)$ storage overall.
Maintaining the current exposure vector $(B_d(t))_{d \in \mathcal{D}}$ requires $O(|\mathcal{D}|)$ space, since each domain stores a single signed integer exposure.
If commitments are 256-bit hashes and typical event metadata is a few hundred bytes, then even for $N$ in the low millions the total storage footprint remains well within the capacity of standard corporate logging infrastructure, and the asymptotic overhead of PoR commitments is additive in the size of the underlying PoR scheme.

\paragraph{Verification cost.}
Recomputing a commitment $C_i$ from an alleged prefix $\mathrm{TPL}_i$ of length $N_i$ requires $O(N_i)$ hash evaluations.
Verifying a policy-based view for a given period amounts to (i) re-evaluating the relevant exposures, (ii) checking the PoR proofs attached to each snapshot in the view, and (iii) checking inclusion of PoT receipts and anchor commitments.
If a particular view touches $N_{\mathsf{view}}$ events, $S_{\mathsf{view}}$ PoR snapshots and $R_{\mathsf{view}}$ PoT receipts, then the total verification time is
$O\!\left(N_{\mathsf{view}} + \mathsf{cost}_{\mathsf{PoR}}(S_{\mathsf{view}}) + \mathsf{cost}_{\mathsf{PoT}}(R_{\mathsf{view}})\right)$,
where $\mathsf{cost}_{\mathsf{PoR}}$ and $\mathsf{cost}_{\mathsf{PoT}}$ denote the cost of verifying the underlying PoR and PoT primitives.
Importantly, observers rarely need to scan the full history: typical investor, auditor or regulatory views can be implemented so that they touch only a small fraction of the total event stream, especially when combined with range proofs that aggregate PoT receipts between snapshots.

\paragraph{Anchoring cost.}
Anchoring every $M$ events amortises Bitcoin fees across multiple commitments.
For example, anchoring monthly with a single \texttt{OP\_RETURN} transaction per anchor yields $O(T/M)$ anchors over a $T$-month horizon, with total on-chain data of a few kilobytes.
Treasuries can choose finer-grained anchoring at the cost of higher fees and more frequent on-chain interactions.

\subsection{Illustrative evaluation and case study}

To ground these bounds in concrete numbers, consider a stylised publicly listed company that holds
$10{,}000$~BTC in its corporate treasury, split across $|\mathcal{D}|=5$ exposure domains:
four \emph{internal} domains (cold storage, collateral, operational hot wallets, and pending settlements) and
one \emph{external} domain representing segregated client accounts at a regulated custodian or exchange.
Assume an average of $N \approx 20{,}000$ PoT events and $12$ PoR snapshots per year, with monthly anchoring to Bitcoin.
Internally, the treasury desk and risk team review daily policy-based views;
an external auditor receives a signed quarterly view; and the securities regulator receives
a lagged, public-investor view at the same quarterly cadence.

Over a representative three-month period this yields on the order of $N/4 \approx 5{,}000$ PoT events and three PoR snapshots.
Each end-of-month snapshot aggregates the previous month's PoT receipts and is anchored once to Bitcoin.
Table~\ref{tab:tpl-quarter} sketches one such quarter: the internal team checks fine-grained movement-justification views,
the custodian checks that its liabilities match the external-domain balances, the auditor replays the quarter from the TPL
prefix, and a macroprudential observer can, in principle, combine the public-investor views from multiple such treasuries.

\begin{table}[H]
\centering
\small
\begin{tabularx}{\linewidth}{@{}l X X X@{}}
\toprule
Month & Events logged & Bitcoin anchor & Primary checkers \\
\midrule
Jan &
$\approx 1{,}600$ PoT receipts; 1 PoR snapshot &
1 transaction with an \texttt{OP\_RETURN} commitment &
Treasury desk and internal risk (daily and monthly views); custodian (end-of-month reconciliation) \\
Feb &
$\approx 1{,}700$ PoT receipts; 1 PoR snapshot &
1 transaction with an \texttt{OP\_RETURN} commitment &
Treasury; external auditor prepares the quarterly replayable view \\
Mar &
$\approx 1{,}700$ PoT receipts; 1 PoR snapshot &
1 transaction with an \texttt{OP\_RETURN} commitment &
Treasury; securities regulator and public-investors receive a lagged quarterly view re-verifiable against Bitcoin \\
\bottomrule
\end{tabularx}
\caption{Illustrative three-month TPL operation for a single corporate treasury with one external custodian.}
\label{tab:tpl-quarter}
\end{table}

\begin{table}[H]
\centering
\scriptsize
\begin{tabularx}{\linewidth}{lXX}
\toprule
Quantity & Stylised value & Expression / comments \\
\midrule
Total PoT events per year & $N \approx 20{,}000$ & Company-scale treasury operations under moderate activity assumptions \\
PoT events per Bitcoin anchor & $\approx N/12 \approx 1{,}700$ & Monthly anchoring with events distributed roughly uniformly over time \\
Anchors over a $T$-year horizon & $12T$ & One \texttt{OP\_RETURN} anchoring transaction per month \\
Total on-chain data over $T$ years & $\approx 12T \cdot s$ bytes & $s$ denotes the payload size per anchor (e.g., $s \in [80,200]$ bytes) \\
Approximate fee range over $T$ years & $[12T \cdot f_{\min} \cdot s,\, 12T \cdot f_{\max} \cdot s]$ satoshis & For fee rates $f_{\min}, f_{\max}$ sat/vByte chosen by the treasury \\
\bottomrule
\end{tabularx}
\caption{Summary resource footprint for the stylised TPL deployment described in Section~\ref{sec:complexity-overhead}.}
\label{tab:tpl-summary}
\end{table}

A minimal prototype that realises this scenario can be implemented as a state-machine engine that ingests PoR snapshots, PoT receipts, and policy events, maintains the domain exposure vector and policy metadata, and periodically derives anchoring commitments and policy-based views.
Such a prototype, coupled with a thin reporting layer that exposes parameterised views and verification routines, would be sufficient to empirically validate the asymptotic storage and verification costs discussed above; we outline one possible architecture in Appendix~\ref{app:prototype}.

If each event record (including hashes, identifiers, encumbrance flags and signatures) is encoded in at most
$500$~bytes and each snapshot in at most $1$~kB, then a single year of operation produces on the order of
$N \cdot 500 \text{ bytes} + 12 \cdot 1\text{ kB} \approx 10 \text{ MB}$
of raw TPL data. At the reporting granularity of quarterly views, this corresponds to roughly $2.5$~MB per quarter, well within the scale of ordinary corporate log-management systems.
Even if we multiply this by a factor of three to account for database indexes, redundancy and auxiliary logs,
a decade of history remains well below a gigabyte, comfortably within the storage capabilities of standard
corporate logging and archiving systems.

On the verification side, checking a one-year prefix $\mathrm{TPL}_i$ for internal audit purposes requires
$O(N)$ hash evaluations, here on the order of $2\cdot 10^4$ invocations of $H$ plus signature checks on the
corresponding PoT receipts.
Generating and verifying a policy-based quarterly view that touches, say, $5{,}000$ events amounts to recomputing
aggregates over those events and checking the relevant Merkle proofs and anchors, again linear in the number of events
visible under the policy.
On commodity hardware, such workloads are easily compatible with interactive dashboards and near-real-time drill-down
by auditors and regulators.

Finally, the anchoring footprint of this deployment is modest.
Monthly commitments over ten years correspond to $120$ Bitcoin transactions with a single \texttt{OP\_RETURN} output each,
for a total on-chain data volume of at most a few kilobytes.
Even under conservative fee assumptions, the direct monetary cost of anchoring remains negligible relative to the size
of the underlying treasury, while providing a stable, publicly verifiable timeline for all subsequent views and reports. For example, even assuming a fee of $100$~sat/vbyte for a $200$-vbyte anchor transaction (about $0.0002$~BTC per anchor), monthly anchoring over a decade consumes only about $0.02$--$0.03$~BTC in total, orders of magnitude smaller than a $10{,}000$~BTC position. Overall, for a medium-sized treasury of this scale, the operational and on-chain overheads of TPL are negligible compared to the assets under management, yet they enable regulators, auditors and public-investors to reason concretely about treasury solvency and aggregate systemic exposure.
\subsection{Bitcoin anchoring and protocol considerations}

Our security arguments rely on an explicit model of how TPL commits to Bitcoin.

\paragraph{Reorganisations and confirmations.}
We assume that an anchor reaching $k$ confirmations is final; We choose $k$ such that the probability of a deeper reorganisation is negligible in the security parameter~$\lambda$ over the time horizon of interest. anchors with depth $< k$ are treated as provisional and may be rolled back if a reorganisation occurs.
In practice, treasuries can require a conservative $k$ when publishing externally-audited reports, while allowing internal processes to work with lower depths for timeliness.

\paragraph{Time and ordering.}
The ledger maintains a logical event index and timestamp for each receipt and snapshot.
Bitcoin block timestamps are used only as coarse external reference points; the total order of events is defined by the TPL itself.
This avoids ambiguities arising from block-time variability and makes replays or re-orderings detectable via inconsistencies between the TPL sequence and anchored commitments.

\paragraph{Fees and protocol losses.}
As noted above, transaction fees and protocol-level losses are accounted for explicitly via the sink domain $d_{\mathrm{fee}}$.
This makes the conservation theorem exact and prevents silent leakage of value through operational costs.

\subsection{Supply-consistency and aggregate claims}

While TPL is defined at the level of a single treasury, the same machinery gives a natural abstraction for
reasoning about \emph{aggregate claims} on Bitcoin across multiple institutions.
Let $\mathcal{L} = \{\mathrm{TPL}^{(1)},\dots,\mathrm{TPL}^{(m)}\}$ denote a family of independently operated
TPL instances, for example corresponding to a set of ETFs, listed corporates and custodians.
For each ledger $\mathrm{TPL}^{(j)}$ and reporting date $t$, a policy-based view induces a well-defined
lower bound $E^{(j)}(t)$ on the institution's net Bitcoin exposure that is re-verifiable against Bitcoin
anchors under the exposure soundness and non-equivocation properties defined in this paper. These guarantees remain conditional on the minimal non-collusion assumption (Assumption~\ref{assump:noncollusion}) for each participating ledger.

We can capture this at a minimal formal level by an aggregate experiment.

\begin{definition}[Aggregate supply-consistency experiment]
\label{def:aggregate-supply}
Fix a family $\mathcal{L} = \{\mathrm{TPL}^{(1)},\dots,\mathrm{TPL}^{(m)}\}$ and a reporting date $t$.
For each $j \in \{1,\dots,m\}$, let $\mathcal{D}_0^{(j)}$ be the subset of domains that a given
policy $V_{\mathsf{exp}}^{(j)}$ treats as in-scope for Bitcoin exposure at time~$t$, and assume that
$\mathcal{D}_0^{(j)}$ is closed over $[0,t]$ in the sense of Definition~\ref{def:closed-domains}.
An \emph{aggregate supply-consistency experiment} outputs
\[
  E_{\mathrm{tot}}(t) \;=\; \sum_{j=1}^{m} E^{(j)}(t)
\]
where each $E^{(j)}(t)$ is the exposure estimate induced by $V_{\mathsf{exp}}^{(j)}$ on the
anchored ledger prefix $\mathrm{TPL}^{(j)}_t$.
The experiment \emph{accepts} if all local verification procedures and view verifications for
$V_{\mathsf{exp}}^{(j)}$ accept for every $j$, and \emph{rejects} otherwise.
\end{definition}

Intuitively, the aggregate experiment considers a family of treasuries that all expose TPL instances and
reports the sum of their individual, policy-based exposure estimates at a given date.

\begin{lemma}[Soundness of aggregate exposure bounds]
\label{lem:aggregate-supply}
Suppose that for each $\mathrm{TPL}^{(j)} \in \mathcal{L}$ the conditions of
Theorem~\ref{lem:restricted-exp-soundness} hold for some exposure policy $V_{\mathsf{exp}}^{(j)}$
and closed domain set $\mathcal{D}_0^{(j)}$.
Assume further that, in the real world, the in-scope domains $\{\mathcal{D}_0^{(j)}\}_{j=1}^m$ jointly
cover all actual on-chain BTC holdings of the participating institutions without double-counting, and
that Bitcoin's protocol-level supply cap is $S_{\max}=21{,}000{,}000$~BTC.
Then in any execution in which the aggregate supply-consistency experiment
(Definition~\ref{def:aggregate-supply}) accepts and all restricted exposure-soundness experiments for
$\{\mathrm{TPL}^{(j)}\}_{j=1}^m$ hold, the aggregate exposure estimate
$E_{\mathrm{tot}}(t)$ satisfies
\[
  E_{\mathrm{tot}}(t) \;\leq\; S_{\max}.
\]
Equivalently, any PPT adversary that produces an accepting execution with
$E_{\mathrm{tot}}(t) > S_{\max}$ with non-negligible probability can be turned into a PPT adversary
that violates restricted exposure soundness for at least one $\mathrm{TPL}^{(j)}$ or the coverage
assumption on $\{\mathcal{D}_0^{(j)}\}_{j=1}^m$.
\end{lemma}

\begin{proof}[Proof sketch]
Under the coverage assumption, the union of in-scope domains
$\{\mathcal{D}_0^{(j)}\}_{j=1}^m$ accounts for all BTC held by the participating institutions.
By Theorem~\ref{lem:restricted-exp-soundness}, each $E^{(j)}(t)$ is a sound exposure estimate for
$\mathcal{D}_0^{(j)}$ in the sense that, except with negligible probability, it does not strictly
overstate the true BTC holdings of those domains while all local checks accept.
Summing over $j$ therefore yields an aggregate estimate $E_{\mathrm{tot}}(t)$ that, except with
negligible probability, does not strictly exceed the total number of BTC actually held by the
participating institutions.
Since the Bitcoin protocol enforces a global supply cap $S_{\max}$, this total is itself bounded
by $S_{\max}$, giving $E_{\mathrm{tot}}(t) \leq S_{\max}$.
Any accepting execution with $E_{\mathrm{tot}}(t) > S_{\max}$ must therefore correspond either to a
violation of restricted exposure soundness for some $\mathrm{TPL}^{(j)}$ or to a failure of the
coverage assumption.
\end{proof}
Because the Bitcoin protocol enforces a fixed maximum supply of $21{,}000{,}000$~BTC, any configuration with
$E_{\mathrm{tot}}(t)$ approaching or exceeding this limit flags either (i) inconsistent or overly aggressive
use of policies, (ii) hidden rehypothecation or double-counting of collateral across ledgers, or
(iii) outright misrepresentation of holdings.
As of late~2025, public data sources report a circulating supply of roughly $19.95$~million~BTC, so
even modest overstatement of exposures by large intermediaries could, in principle, push $E_{\mathrm{tot}}(t)$
close to the protocol cap.

We stress that a full macroprudential treatment of this cross-institutional problem, including strategic partial
participation by adversaries who selectively deploy or omit TPL instances, is outside the scope of
this work.
However, the formal interface that TPL provides, namely anchored, forward-secure logs and verifiable, policy-based
views is deliberately structured so that regulators, standard-setters and independent researchers can
build such aggregate supply-consistency checks on top, using TPL instances as auditable primitives.

\subsection{Forward compatibility}

TPL extends naturally to:

\begin{itemize}[noitemsep]
\item Lightning and rollups,
\item Tokenized treasuries,
\item Bitcoin-backed instruments,
\item Earth--Moon--Mars multi-domain treasuries.
\end{itemize}

\subsection{Novelty relative to prior work}

The TPL differs from existing proof-of-reserves and assurance proposals along several dimensions.
First, it targets long-horizon corporate and public-sector treasuries, where management is a
strategic actor subject to governance and regulatory constraints, rather than short-lived exchange
solvency snapshots.
Second, it treats the treasury as a collection of heterogeneous exposure domains (on-chain,
custodial, derivative and collateral) that must be reasoned about within a single, policy-aware
ledger.
Third, it separates an append-only evidence layer from a machine-checked reporting layer that
exports policy-based stakeholder views, rather than only raw proofs.
Finally, it is deliberately Bitcoin-only at the monetary layer: other assets and instruments enter
only as exposures of the treasury, not as independent monetary substrates.

From a cryptographic point of view, TPL can be understood as a domain-specific accountable logging
and transparency system, in the spirit of secure-logging and transparency-log constructions such as
Certificate Transparency.
It shares with these generic systems the use of append-only, hash-chained commitments and
non-equivocation guarantees, but differs in two main respects.
First, the state space is explicitly structured around multi-domain balance vectors and exposure
flows that satisfy a conservation law.
Second, the primary outputs are policy-based views for heterogeneous stakeholders rather than raw
append-only logs.
The formal treatment of conservation (Theorem~\ref{thm:conservation}), domain-level exposure
soundness (Theorem~\ref{lem:restricted-exp-soundness}) and their interaction with PoR/PoT soundness
in a multi-domain treasury model is, as far as we are aware, not present in prior work on
accountable logging or proof-of-reserves, and appears to be specific to the Bitcoin treasury setting
considered here.

From the vantage point of prior work on exchange-style PoR, Provisions/DAPOL-style schemes, and
Certificate-Transparency-style accountable logs (cf.~Table~\ref{tab:related-comparison}), we are not
aware of any construction that simultaneously
(i) models a full multi-domain treasury state with an explicit conservation law,
(ii) defines deployment-level exposure-soundness and policy-completeness notions over that state, and
(iii) provides verifiable, policy-based stakeholder views that are sound with respect to both the
underlying PoR/PoT primitives and the anchored hash-chain.
In this sense the main novelty of TPL is definitional and compositional: it lifts existing PoR and
transparency mechanisms into a unified, Bitcoin-specific framework that captures the deployment-level
properties regulators and treasury managers care about.
\subsection{Limitations and open questions}

Our proposal has several limitations that bound its applicability and clarify the scope of the current guarantees.

Operationally, the TPL does not eliminate the need for trust in corporate governance, policy choice and
correct modelling of the treasury's business.
Many behaviours remain outside the reach of cryptography: senior management can mis-set policies,
fail to operate the system, or selectively disclose views to external stakeholders.
Our security notions therefore apply to \emph{what is logged and anchored}, not to every aspect of
corporate conduct.

Economically, the ledger can only represent positions that are included in the domain decomposition.
Hidden encumbrances, off-balance-sheet commitments, or complex rehypothecation chains that are not
captured as domains or flows remain outside TPL's security envelope, and the domain-completeness
assumption is essential.
In particular, analyses of hidden encumbrances or rehypothecation that route value through instruments
omitted from the domain decomposition lie outside the guarantees we prove.

On the cryptographic side, the deployment-level notions of exposure soundness and policy completeness
(Definitions~\ref{def:exp-soundness} and~\ref{def:policy-completeness}) are defined in full generality,
but we only prove a subset of the resulting goals.
In this version we establish three core integrity theorems for an idealised TPL instance
(Conservation, non-equivocation/forward integrity, and view correctness; Theorems~\ref{thm:conservation},
\ref{thm:non-equivocation} and~\ref{thm:view-correctness}) and two restricted deployment-level
results: a restricted exposure-soundness theorem for closed sets of domains under faithful policies
(Theorem~\ref{lem:restricted-exp-soundness}) and a restricted privacy theorem for a simple public-investor
policy (Theorem~\ref{thm:pub-privacy}).
The general case of exposure soundness and policy completeness under strategic behaviour, richer policy
languages, and partial deployment across institutions remains an open problem for future work.

On the privacy side, while the design aims to minimise leakage to what is already implied by anchors,
policy-based views and authorised disclosures, we do not yet give a full, simulation-based privacy
treatment for an expressive class of policies.
Theorem~\ref{thm:pub-privacy} should be viewed as an existence result for one natural leakage profile
for public-investors, rather than as a complete characterisation of privacy in TPL-like systems.

\paragraph{Implementation and evaluation.}
An important next step is to implement a prototype TPL engine and empirically evaluate its practicality.
A minimal deployment would expose a programmable interface for ingesting treasury events, constructing and anchoring ledger commits, and generating policy-based views for internal and external observers.
On top of this engine, one can benchmark end-to-end throughput, anchoring latency, storage overhead, and verification cost for both internal auditors and external regulators, under realistic event workloads and regulatory policies.
We leave such an implementation and its experimental evaluation, including comparisons with existing proof-of-reserves deployments and conventional audit workflows, to future work, noting that a prototype of this form would provide strong evidence that the TPL abstraction is not only conceptually sound but also viable in production environments.
In this version we specify protocol syntax and property-based security games, and we formally prove three core integrity properties: conservation across domains, forward integrity and non-repudiation, and view correctness under standard hash-function and anchoring assumptions; a complete simulation-based analysis and a fuller treatment of privacy and strategic behaviour remain future work.
In particular, modelling how strategic actors might shape transaction patterns or reporting schedules, and how TPL interacts with broader governance processes, seems essential for understanding which classes of fraud become significantly more detectable in practice.

We have presented the \emph{Treasury Proof Ledger} (TPL), a Bitcoin-anchored logging framework that
organises proof-of-reserves snapshots, proof-of-transit receipts, and policy metadata into a single,
append-only ledger for Bitcoin treasuries.
Within an explicit system and threat model we specified the protocol syntax and state machine, defined
game-based security notions for treasury transparency, and proved that an idealised TPL construction
achieves three core integrity properties under standard assumptions on hash functions, digital signatures,
the soundness of the underlying PoR and PoT primitives, Bitcoin anchoring, and the minimal non-collusion assumption of Section~\ref{sec:threat}: conservation across domains, forward integrity and non-repudiation of logged events
via non-equivocation, and verifiable policy-based views for deterministic, stateless policies.

Beyond these integrity guarantees we also formalised stronger, policy-dependent goals,
including exposure soundness and policy completeness.
In this paper we prove a restricted exposure-soundness theorem for closed sets of domains under a
faithful exposure policy, and a restricted privacy theorem for a simple public-investor policy,
showing how a leakage function that reveals only aggregate total and encumbered exposure can be
used to formulate a simulation-based privacy guarantee under additional assumptions on the
underlying PoR and PoT primitives.
Identifying concrete PoR, PoT, and signature schemes that realise this leakage profile is left to future work.
These results make the scope of the current formal guarantees explicit and provide a template for
extending the framework to more expressive policies and adversarial behaviours.

As noted in Section~\ref{sec:threat}, the full reductions and additional technical lemmas are provided in the appendices of this version.

Overall, TPL should be viewed as a cryptographic foundation rather than a complete reporting stack.
A natural next step is to instantiate and benchmark TPL with concrete proof-of-reserves schemes at
scale, to develop richer policy languages and tooling for practitioners, and to integrate TPL with
existing accounting, audit, and governance processes.
A full simulation-based privacy treatment and a comprehensive analysis of exposure soundness and
policy completeness in realistic deployment settings (including strategic behaviour and partial
participation) remain open problems for future work.

We leave a full prototype and empirical evaluation to future work; this paper's contribution is the formal framework and security analysis.

\appendix

\section{Cross-institution supply-consistency checks}
\label{sec:supply-consistency}

Beyond single-treasury guarantees, TPL is intended to support system-wide checks that aggregate the exposures of multiple
treasuries and compare them against Bitcoin's protocol-level supply.

Let $T^{(1)},\dots,T^{(n)}$ denote $n$ distinct treasuries, each with its own deployed TPL instance and anchored ledger.
Write $B_{\mathrm{tot}}^{(j)}(t)$ for the total Bitcoin exposure reported by treasury $T^{(j)}$ at logical time~$t$ under a
faithful exposure policy on a closed set of domains, and let $\mathsf{CirculatingBTC}(t)$ denote the circulating Bitcoin
supply at time~$t$ as determined from the underlying Bitcoin chain (for example, total mined minus provably burned coins).

\begin{definition}[Cross-institution $\varepsilon$-supply-consistency]
\label{def:cross-supply-consistency}
Fix a tolerance parameter $\varepsilon \geq 0$ capturing acceptable measurement, timing, and modelling error.
We say that the family of TPL instances $\{T^{(j)}\}_{j=1}^n$ is \emph{$\varepsilon$-supply-consistent at time~$t$} if
\[
  \sum_{j=1}^n B_{\mathrm{tot}}^{(j)}(t)
  \;\leq\;
  \mathsf{CirculatingBTC}(t) + \varepsilon.
\]
\end{definition}

This inequality can be interpreted as a system-wide sanity check: the aggregate Bitcoin-denominated claims of the monitored
treasuries do not exceed the protocol's circulating supply by more than~$\varepsilon$.

\begin{proposition}[TPL-based supply-consistency sanity check]
\label{prop:supply-consistency}
Assume that, for each $j \in \{1,\dots,n\}$, the corresponding TPL instance satisfies the conditions of
Theorem~\ref{thm:ledger-level-security} for some faithful exposure policy on a closed set of domains that determines
$B_{\mathrm{tot}}^{(j)}(t)$.
Assume further that the procedure used to compute $\mathsf{CirculatingBTC}(t)$ from the Bitcoin chain is correct up to
additive error at most $\varepsilon$.

Then, in this restricted model, any probabilistic polynomial-time adversary that outputs anchored ledger prefixes
and associated exposure views for the family $\{T^{(j)}\}_{j=1}^n$ and a time~$t$ such that
\[
  \sum_{j=1}^n B_{\mathrm{tot}}^{(j)}(t)
  \;>\;
  \mathsf{CirculatingBTC}(t) + \varepsilon
\]
with non-negligible probability either:
\begin{enumerate}[noitemsep]
\item exploits a modelling or governance failure (for example, violating the faithfulness or closed-domain
      assumptions, or double-counting the same underlying Bitcoin across multiple treasuries), or
\item yields a PPT algorithm that breaks at least one of the underlying primitives
      (PoR existence/ownership, PoT unforgeability, or collision resistance of~$H$)
      for some $T^{(j)}$.
\end{enumerate}
\end{proposition}

\noindent
Intuitively, if each individual TPL instance cannot inflate its reported exposure beyond the underlying Bitcoin it
controls (by restricted exposure-soundness), then a cross-institution violation of the global inequality must come
either from collusive or inconsistent modelling choices across institutions or from an attack on the underlying
cryptographic mechanisms.
A regulator or systemic-risk monitor can therefore interpret sustained violations as strong evidence of either
governance problems or cryptographic compromise.

As a concrete illustration, suppose three treasuries $T^{(1)},T^{(2)},T^{(3)}$ report total Bitcoin exposures
$B_{\mathrm{tot}}^{(1)}(t)=10{,}000$, $B_{\mathrm{tot}}^{(2)}(t)=15{,}000$, and $B_{\mathrm{tot}}^{(3)}(t)=5{,}000$ at some
reference time~$t$ under faithful exposure policies on closed domain sets, and the circulating-supply oracle reports
$\mathsf{CirculatingBTC}(t) \approx 19{,}600{,}000$~BTC. For any reasonable tolerance parameter $\varepsilon$ (for example,
accounting for modest timing and measurement error), the inequality of Definition~\ref{def:cross-supply-consistency}
is comfortably satisfied. By contrast, if a family of large treasuries collectively claimed exposures within a few
percent of, or even above, the protocol-level circulating supply, then under the assumptions of
Proposition~\ref{prop:supply-consistency} this would be a strong signal that either the underlying primitives have been
compromised or that economically significant positions are being double-counted or omitted.

In practice, a supervisor can treat the tolerance~$\varepsilon$ as an explicit policy parameter that absorbs known
sources of unavoidable discrepancy (such as delayed reporting, rounding conventions, or incomplete coverage of small
institutions) while still flagging sustained or growing violations as indicators of systemic-risk build-up.

\section{Implementation Notes and Prototype Sketch}
\label{app:prototype}

A prototype implementation of TPL can follow the state-machine structure in Section~\ref{sec:tpl} while remaining
compatible with existing treasury, custody, and audit systems:

\begin{itemize}[noitemsep]
\item \textbf{State-machine engine.} A core service ingests PoR snapshots, PoT receipts, and policy events, validates
      signatures, updates the domain exposure vector (including the fee domain $d_{\mathrm{fee}}$), and appends
      authenticated events to an append-only ledger.
\item \textbf{Anchoring module.} A background process batches commitments $C_i$ for recent prefixes of the ledger,
      constructs Bitcoin anchoring transactions (for example, via \texttt{OP\_RETURN} or Taproot commitments), and
      tracks confirmation depth to expose the $\mathcal{A}_{\mathrm{BTC}}$ interface used in our security games.
\item \textbf{Policy engine and reporting API.} A policy layer implements the language from Section~\ref{sec:responsible},
      realising $\mathsf{GenView}$ and $\mathsf{VerifyView}$ for a small catalogue of stakeholder policies
      (board, regulator, public-investors).
      It exposes parameterised reporting endpoints that generate policy-based views and their associated proofs.
\item \textbf{Storage backend.} The event log and derived state can be stored in a log-structured store or relational
      database with indexes on time, domain, and policy tags, allowing efficient generation of period-based views
      and drill-down queries.
\item \textbf{Pluggable PoR/PoT adapters.} Adapters connect TPL's abstract PoR and PoT interfaces to concrete schemes,
      such as Merkle-tree PoR in the style of Provisions and an accountable PoT service with hash-chained receipts.
\end{itemize}

\paragraph{Instantiation A: Merkle-tree PoR, PoT, and zk-SNARK views.}
One concrete instantiation of the abstract interfaces in Sections~\ref{sec:model}--\ref{sec:responsible} can be built from
off-the-shelf components:

\begin{itemize}[noitemsep]
\item \emph{PoR.} For each domain that holds on-chain Bitcoin or custodial liabilities, the treasurer runs a Merkle-tree
      PoR scheme in the style of Provisions~\cite{Dagher2015}, committing to the set of relevant UTXOs or customer
      liabilities.
      The Merkle roots are included in the $\mathsf{evid}$ field of snapshot events and contribute to the computation of
      $B_d(t)$ for the affected domains.
\item \emph{PoT.} Each treasury event $e_i$ carries a PoT receipt obtained from an accountable logging service that
      maintains a hash-chained sequence of receipts.
      The message hashed into the chain includes the event metadata $(t_i,d_{\mathrm{src}},d_{\mathrm{dst}},v)$, references
      to the relevant PoR roots, and a pointer to the previous receipt.
      Periodically, the logging service anchors the head of the hash chain into Bitcoin via an \texttt{OP\_RETURN} or
      Taproot commitment; the anchored commitments are exactly the $C_i$ used in our security games.
\item \emph{Policy proofs.} For policies that require hiding internal flows or domain-level balances, the treasurer can
      generate succinct zero-knowledge proofs using a general-purpose zk-SNARK.
      The public inputs to the zk-SNARK consist of the Bitcoin anchoring transactions and PoR/PoT roots,
      while the private witness consists of the internal ledger state.
      The statement proved is that the policy-based view $V_\alpha(\mathrm{TPL}_i)$ is a correct function of a ledger
      state consistent with the anchored commitments and PoR/PoT artefacts.
\end{itemize}

With modern zk-SNARK systems, the resulting proofs can be kept small (on the order of kilobytes) and can be verified in
milliseconds to seconds on commodity hardware, even for views that aggregate many underlying events.
The precise choice of PoR, PoT, and zk-SNARK scheme affects concrete parameters but does not change the TPL abstraction
or the security definitions.

\paragraph{Toy deployment and cost sketch.}
To give a concrete sense of scale, consider a listed company that:
\begin{itemize}[noitemsep]
\item maintains $5$ exposure domains (for example cold storage, two exchanges, collateral, and derivatives),
\item records on the order of $10^3$ balance-affecting treasury events per day into the PoT stream,
\item obtains PoR snapshots for all domains once per day, and
\item anchors a commitment $C_i$ to Bitcoin once per day.
\end{itemize}
In such a deployment, the TPL log would accumulate roughly $3.6 \times 10^5$ events and $365$ PoR snapshots per year.
Assuming a few hundred bytes per event and kilobytes per snapshot, the raw log size is on the order of a few hundred
megabytes per year, well within the capabilities of standard enterprise storage.
A daily anchoring cadence yields $365$ anchoring transactions per year; even at conservative feerates this is negligible
compared to typical treasury transaction volumes.
Verification workloads are similarly modest: reconstructing a policy-based view for a reporting period reduces to hashing
and signature checks over a few thousand events plus verification of the associated PoR and PoT artefacts.
These back-of-the-envelope figures are not a substitute for a full performance evaluation, but they indicate that the
framework can be instantiated with realistic reporting frequencies without imposing prohibitive storage or fee costs.
Even a minimal prototype with these components is sufficient to evaluate throughput, storage growth, and verification
latency under realistic event volumes, complementing the asymptotic bounds and stylised case study presented in
Section~\ref{sec:reporting} and the Discussion section.

\section{Glossary of Key Terms}

\begin{description}
  \item[Policy-based view] A selectively disclosed projection of the TPL log generated according to a specified policy.
  \item[Exposure soundness] Integrity property ensuring that domain balances implied by the log are internally consistent.
  \item[Forward integrity] Property ensuring past states cannot be modified after commitment.
  \item[Non-equivocation] Property ensuring a single consistent log prefix for all parties.
  \item[PoR] Proof-of-Reserves snapshot of Bitcoin holdings.
  \item[PoT] Proof-of-Transit receipt chained in the TPL log.
\end{description}

\section{Terminology Table}
\begin{tabular}{p{0.25\linewidth} p{0.65\linewidth}}
\textbf{Term} & \textbf{Meaning} \\
\hline
Policy-based view & Selective disclosure view derived from a ledger prefix $\mathrm{TPL}_i$ according to a policy $V_\alpha$. \\
Exposure soundness & Inability to overstate total BTC exposure across domains while all PoR, PoT, and anchoring checks succeed. \\
Forward integrity & Past states cannot be modified once committed. \\
Non-equivocation & A single consistent log prefix across all verifiers. \\
PoR snapshot & Periodic proof-of-reserves attestation of Bitcoin holdings. \\
PoT receipt & Hash-chained proof-of-transit record. \\
\end{tabular}

\section{Notational Summary}
\begin{tabular}{p{0.25\linewidth} p{0.65\linewidth}}
\textbf{Symbol} & \textbf{Definition} \\
\hline
$\mathcal{D}$ & Set of treasury domains. \\
$e$ & Treasury event $(t,d_{\mathrm{src}},d_{\mathrm{dst}},v,\mathsf{evid},m)$ and derived digest $h = H(\mathsf{evid} \,\Vert\, v \,\Vert\, m)$. \\
$\mathrm{TPL}_i$ & Ledger prefix after $i$ abstract inputs (events, snapshot triggers, anchors). \\
$C_i$ & Ledger commitment $H(\mathrm{TPL}_i)$. \\
$B_d(t)$ & Exposure of domain $d$ at time $t$. \\
$H$ & Collision-resistant hash function. \\
\end{tabular}

\section{Non-Equivocation Game and Reduction Details}

\begin{figure}[t]
  \centering
  \resizebox{\linewidth}{!}{%
\begin{tikzpicture}[>=Stealth, node distance=1.8cm,
                      every node/.style={draw, rounded corners, align=center,
                                        minimum width=4.8cm, minimum height=1cm}]
    \node (cons) {Conservation across domains\\(Theorem~\ref{thm:conservation})};
    \node[right=of cons] (neq) {Non-equivocation of TPL\\(Theorem~\ref{thm:non-equivocation})};
    \node[right=of neq] (view) {View correctness for policies\\(Theorem~\ref{thm:view-correctness})};

    \node[below=1.8cm of neq] (exp) {Restricted exposure soundness\\(Theorem~\ref{lem:restricted-exp-soundness})};
    \node[below=1.8cm of view] (priv) {Restricted privacy for public-investors\\(Theorem~\ref{thm:pub-privacy})};

    \draw[->] (cons) -- (neq);
    \draw[->] (neq) -- (view);
    \draw[->] (cons.south) |- (exp.west);
    \draw[->] (neq.south) -- (exp.north);
    \draw[->] (view.south) -- (priv.north);
    \draw[->] (exp.east) -- (priv.west);
  \end{tikzpicture}
}
  \caption{Roadmap of the main reduction arguments.
  Conservation establishes that the total exposure across domains is invariant under honest operation.
  Non-equivocation and view correctness ensure that policy-based views are uniquely determined by the anchored ledger.
  Restricted exposure soundness shows that any adversary who increases exposure must break at least one underlying primitive,
  and the restricted privacy theorem builds on this for a specific public-investor policy.}
  \label{fig:proofs-roadmap}
\end{figure}
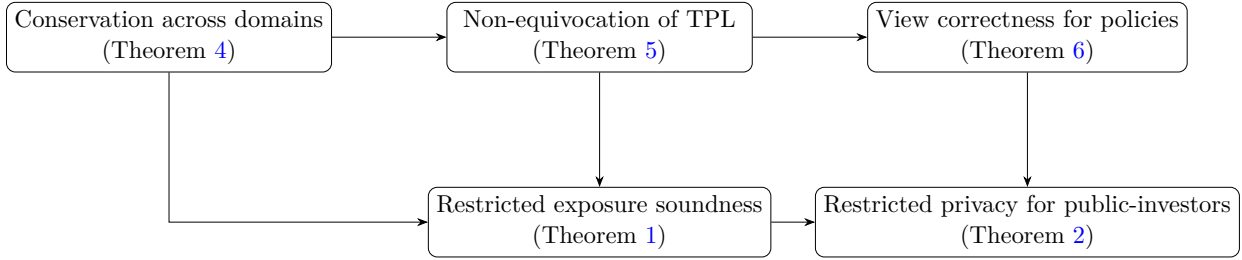
\label{app:neq-reduction}

In this appendix we restate the non-equivocation property using an explicit oracle game and sketch the standard reduction to proof-of-transit unforgeability and hash collision resistance.

\subsection*{Game $\mathsf{Exp}^{\mathsf{neq}}_{\mathcal{A}}(\lambda)$}

For a PPT adversary $\mathcal{A}$ and security parameter $\lambda$, the experiment $\mathsf{Exp}^{\mathsf{neq}}_{\mathcal{A}}(\lambda)$ proceeds as follows.
The challenger runs $\mathsf{Setup}(1^\lambda)$ to obtain public parameters $\mathsf{pp}$, an initial ledger prefix $\mathrm{TPL}_0$, and its commitment $C_0$.
The adversary receives $\mathsf{pp}$ and oracle access to the public TPL interface:
\begin{itemize}
  \item $\mathcal{O}_{\mathsf{AppendEvent}}(e)$ appends an abstract input $e$ to the log, extending $\mathrm{TPL}_i$ to $\mathrm{TPL}_{i+1}$ and returning any public artefacts.
  \item $\mathcal{O}_{\mathsf{Snapshot}}()$ triggers a PoR snapshot according to the construction in Section~\ref{sec:tpl}.
  \item $\mathcal{O}_{\mathsf{Anchor}}()$ publishes an anchor for the current ledger commitment to the underlying blockchain once the anchoring policy allows it.
  \item $\mathcal{O}_{\mathsf{GenView}}(\alpha,t)$ returns the view $V_\alpha(\mathrm{TPL}_i)$ for observer class $\alpha$ at (logical) time~$t$.
\end{itemize}
At the end of its interaction $\mathcal{A}$ outputs two tuples $(\mathrm{TPL}_i,C_i)$ and $(\mathrm{TPL}'_i,C'_i)$, together with an observer class~$\alpha$.
The experiment checks that both prefixes are well-formed, that $C_i$ and $C'_i$ are valid anchored commitments, and that both induced views for~$\alpha$ are accepted by the verification algorithm, yet $V_\alpha(\mathrm{TPL}_i) \neq V_\alpha(\mathrm{TPL}'_i)$.
If all checks pass, the experiment outputs $1$.
The non-equivocation advantage of~$\mathcal{A}$ is
\[
  \mathsf{Adv}^{\mathsf{neq}}_{\mathcal{A}}(\lambda)
  = \Pr\big[\mathsf{Exp}^{\mathsf{neq}}_{\mathcal{A}}(\lambda)=1\big].
\]

\subsection*{Reduction outline}

The proof of Theorem~\ref{thm:non-equivocation} can be seen as a standard reduction argument.
Assume there exists a PPT adversary $\mathcal{A}$ with non-negligible non-equivocation advantage.
We construct PPT adversaries $\mathcal{B}$ and $\mathcal{C}$ such that
\[
  \mathsf{Adv}^{\mathsf{neq}}_{\mathcal{A}}(\lambda)
  \leq
  \mathsf{Adv}^{\mathsf{pot-forge}}_{\mathcal{B}}(\lambda)
  +
  \mathsf{Adv}^{\mathsf{coll}}_{\mathcal{C}}(\lambda)
  + \mathsf{negl}(\lambda).
\]
Adversary $\mathcal{B}$ plays the PoT-forgery game and internally simulates the TPL oracles for~$\mathcal{A}$ using its own access to the PoT signing oracle.
When $\mathcal{A}$ outputs a winning pair of prefixes, $\mathcal{B}$ identifies the first index at which the sequences of PoT receipts diverge.
If this divergence yields a valid receipt chain that is not a prefix of the honest chain, $\mathcal{B}$ outputs it as a PoT forgery.
Adversary $\mathcal{C}$ plays the hash-collision game for $H$ and again simulates the TPL interface for~$\mathcal{A}$.
If the PoT receipts agree on both prefixes but the ledger contents differ at some minimal index, then, by construction of the commitments, $\mathcal{C}$ obtains two distinct inputs that hash to the same value and outputs them as a collision for $H$.
Standard arguments then show that if both the underlying PoT scheme and the hash function are secure, the non-equivocation advantage must be negligible.

\section{Restricted Exposure-Soundness Reduction}
\label{app:exp-soundness-reduction}

In this appendix we give a full reduction proof of Theorem~\ref{lem:restricted-exp-soundness}.
Recall that the unrestricted exposure-soundness notion is captured by
Definition~\ref{def:exp-soundness}, whose experiment
$\mathsf{Exp}^{\mathsf{exp\mbox{-}sound}}_{\mathcal{A}}$ allows the adversary to manipulate
all domains in~$\mathcal{D}$.
In the restricted setting of Theorem~\ref{lem:restricted-exp-soundness} we fix a non-empty
subset $\mathcal{D}_0 \subseteq \mathcal{D}\setminus\{d_{\mathrm{fee}}\}$ that is closed over~$[0,t_i]$
and a faithful policy $V_{\mathsf{exp}}$ that either reveals the vector
$(B_d(t_i))_{d\in\mathcal{D}_0}$ or its algebraic sum.
We consider the obvious restriction of the exposure-soundness experiment to $\mathcal{D}_0$,
which we denote by $\mathsf{Exp}^{\mathsf{exp\mbox{-}sound},\mathcal{D}_0}_{\mathcal{A}}(\lambda)$.
The only difference from Definition~\ref{def:exp-soundness} is that, in the final check, the
experiment tests
\[
  \sum_{d\in\mathcal{D}_0} B_d^\star(t_i) \;>\;
  \sum_{d\in\mathcal{D}_0} B_d(t_i)
\]
rather than summing over all $d\in\mathcal{D}$.

\begin{lemma}
\label{lem:restricted-exp-reduction}
Assume that the PoR primitive is existence- and ownership-sound, the PoT construction
satisfies Definition~\ref{def:pot-unforgeability}, and $H$ is collision-resistant.
Assume also that the minimal non-collusion assumption of Section~\ref{sec:threat} holds
for every domain in $\mathcal{D}_0$.
Then for every PPT adversary $\mathcal{A}$ there exist PPT adversaries
$\mathcal{B}_{\mathsf{por}}$, $\mathcal{B}_{\mathsf{pot}}$, and $\mathcal{B}_{H}$ such that
\[
  \mathsf{Adv}^{\mathsf{exp\mbox{-}sound},\mathcal{D}_0}_{\mathcal{A}}(\lambda)
  \;\leq\;
  \mathsf{Adv}^{\mathsf{por}}_{\mathcal{B}_{\mathsf{por}}}(\lambda)
  + \mathsf{Adv}^{\mathsf{pot\mbox{-}forge}}_{\mathcal{B}_{\mathsf{pot}}}(\lambda)
  + \mathsf{Adv}^{\mathsf{coll}}_{\mathcal{B}_{H}}(\lambda)
  + \mathsf{negl}(\lambda),
\]
where $\mathsf{Adv}^{\mathsf{por}}$ denotes the combined existence/ownership-soundness
advantage of a PoR adversary, $\mathsf{Adv}^{\mathsf{pot\mbox{-}forge}}$ the advantage in
the PoT-forgery game of Definition~\ref{def:pot-unforgeability}, and
$\mathsf{Adv}^{\mathsf{coll}}$ the collision-finding advantage against~$H$.
\end{lemma}

\begin{proof}
Fix a closed set $\mathcal{D}_0$ and a time $t_i$.
Let $\mathcal{A}$ be any PPT adversary for the restricted exposure-soundness game and
write $E$ for the event that
$\mathsf{Exp}^{\mathsf{exp\mbox{-}sound},\mathcal{D}_0}_{\mathcal{A}}(\lambda)=1$.
By Definition~\ref{def:closed-domains} and Theorem~\ref{thm:conservation}, for every
$t\leq t_i$ the algebraic sum
\[
  S(t) \;:=\; \sum_{d\in\mathcal{D}_0} B_d(t)
\]
is uniquely determined by the initial exposures of domains in~$\mathcal{D}_0$ and the
net flow of value between $\mathcal{D}_0$ and $d_{\mathrm{fee}}$ up to time~$t$.
The minimal non-collusion assumption guarantees that this net flow is faithfully
reflected in the logged PoR snapshots and PoT receipts for domains in
$\mathcal{D}_0 \cup\{d_{\mathrm{fee}}\}$.

In an honest execution these artefacts induce some canonical value $S(t_i)$.
On event~$E$, the adversary outputs an alternative ledger prefix $\mathrm{TPL}^\star_i$
together with anchors and artefacts such that
\[
  S^\star(t_i) \;:=\; \sum_{d\in\mathcal{D}_0} B_d^\star(t_i) \;>\; S(t_i)
\]
while all local PoR, PoT, and anchoring checks accept and the policy-based view for
$V_{\mathsf{exp}}$ is verifying.
Consider the earliest logical time $t_k\leq t_i$ at which the transcript of artefacts
affecting $\mathcal{D}_0\cup\{d_{\mathrm{fee}}\}$ differs between $\mathrm{TPL}_i$ and
$\mathrm{TPL}^\star_i$.
There are three exhaustive and mutually exclusive possibilities, depending on the
type of the first discrepancy.

\paragraph{Case 1: PoR discrepancy.}
At time $t_k$ the first difference is a PoR snapshot for some domain in~$\mathcal{D}_0$.
By construction, the snapshot in $\mathrm{TPL}^\star_i$ must be accepted by all PoR
verification procedures, and the induced balances for $\mathcal{D}_0$ must eventually
lead to $S^\star(t_i) > S(t_i)$.
We construct an adversary $\mathcal{B}_{\mathsf{por}}$ for the PoR existence/ownership
soundness game as follows.
$\mathcal{B}_{\mathsf{por}}$ internally runs $\mathcal{A}$ and perfectly simulates
the TPL interface, except that whenever the first discrepant PoR snapshot would be
produced, it instead uses the PoR challenger to obtain the corresponding commitment
and proof.
Because the remainder of the simulation is honest and $t_k$ is the first point of
divergence, the view of $\mathcal{A}$ in this simulation is (up to negligible
differences coming from the PoR game interface) identical to its view in the
original experiment.
Whenever $\mathcal{A}$ wins, $\mathcal{B}_{\mathsf{por}}$ outputs the mismatched
snapshot and its proof as a PoR forgery.
By definition of the PoR game, this snapshot either certifies the existence of
coins that do not exist on-chain or asserts ownership of coins not controlled by
the treasury.
Thus
\[
  \Pr[E \wedge \text{PoR discrepancy at }t_k]
  \;\leq\;
  \mathsf{Adv}^{\mathsf{por}}_{\mathcal{B}_{\mathsf{por}}}(\lambda) + \mathsf{negl}(\lambda).
\]

\paragraph{Case 2: PoT discrepancy.}
Suppose the first discrepancy corresponds to a PoT receipt for a transfer between
domains in $\mathcal{D}_0\cup\{d_{\mathrm{fee}}\}$ that appears in
$\mathrm{TPL}^\star_i$ but not in $\mathrm{TPL}_i$, or vice versa, again with all
PoT verifications accepting.
We build an adversary $\mathcal{B}_{\mathsf{pot}}$ for the PoT-forgery experiment
of Definition~\ref{def:pot-unforgeability}.
$\mathcal{B}_{\mathsf{pot}}$ simulates the TPL interface for~$\mathcal{A}$ using its
own access to the PoT signing and verification oracles.
When $\mathcal{A}$ produces a winning pair $(\mathrm{TPL}_i,\mathrm{TPL}^\star_i)$,
$\mathcal{B}_{\mathsf{pot}}$ identifies the earliest time $t_k$ where the sequences
of PoT receipts differ and outputs the corresponding receipt together with the
preceding honest PoT chain as its forgery.
By construction, this receipt is valid but does not extend the honest PoT chain as
a prefix, contradicting PoT unforgeability.
Hence
\[
  \Pr[E \wedge \text{PoT discrepancy at }t_k]
  \;\leq\;
  \mathsf{Adv}^{\mathsf{pot\mbox{-}forge}}_{\mathcal{B}_{\mathsf{pot}}}(\lambda)
  + \mathsf{negl}(\lambda).
\]

\paragraph{Case 3: Commitment or hashing discrepancy.}
Finally, suppose there is no discrepancy in PoR snapshots or PoT receipts for
$\mathcal{D}_0\cup\{d_{\mathrm{fee}}\}$, yet $S^\star(t_i) > S(t_i)$.
Then the serialised ledger prefixes (including the policy metadata that affects
exposures) must differ in some other way that changes the implied balances while
all verifications still pass.
Because commitments to ledger prefixes are embedded in Bitcoin transactions as
$H(\cdot)$-based digests, there must exist two distinct inputs to~$H$ that produce
the same digest.
We define an adversary $\mathcal{B}_H$ against collision resistance of~$H$ that
simulates the TPL interface for $\mathcal{A}$ and, upon observing a successful
pair $(\mathrm{TPL}_i,\mathrm{TPL}^\star_i)$ with agreed PoR/PoT artefacts but
different implied balances, outputs the two distinct serialisations that hash to
the same value.
Thus
\[
  \Pr[E \wedge \text{no PoR/PoT discrepancy}] \;\leq\;
  \mathsf{Adv}^{\mathsf{coll}}_{\mathcal{B}_H}(\lambda) + \mathsf{negl}(\lambda).
\]

Combining the three cases and using a union bound yields the claimed inequality for
\[
  \mathsf{Adv}^{\mathsf{exp\mbox{-}sound},\mathcal{D}_0}_{\mathcal{A}}(\lambda).
\]
\end{proof}

Lemma~\ref{lem:restricted-exp-reduction} immediately implies
Theorem~\ref{lem:restricted-exp-soundness}, since by assumption
$\mathsf{Adv}^{\mathsf{por}}_{\mathcal{B}_{\mathsf{por}}}$,
$\mathsf{Adv}^{\mathsf{pot\mbox{-}forge}}_{\mathcal{B}_{\mathsf{pot}}}$, and
$\mathsf{Adv}^{\mathsf{coll}}_{\mathcal{B}_H}$ are negligible in~$\lambda$.

\medskip

\section{Restricted Privacy Simulation}
\label{app:privacy-reduction}

We now provide a more detailed simulation-based argument for
Theorem~\ref{thm:pub-privacy}.
For clarity we restate the two experiments from the theorem.
In the \emph{real} experiment the challenger runs TPL with honest parties,
restricted to the public-investor policy $V_{\mathsf{pub}}$, and gives the
adversary~$\mathcal{A}$ the full protocol transcript together with the anchored
Bitcoin transactions.
In the \emph{ideal} experiment an ideal functionality first samples an internal
execution of TPL, computes the leakage $L_{\mathsf{pub}}(\mathrm{TPL}_i)$ and the
corresponding anchoring transactions, and gives only this information to a PPT
simulator~$\mathcal{S}$.
The simulator interacts with $\mathcal{A}$ and must produce a view that is
indistinguishable from the real transcript.

Let $\mathsf{Real}_{\mathcal{A},\mathcal{D}}(\lambda)$ denote the probability
that a PPT distinguisher~$\mathcal{D}$ outputs~$1$ on the view of $\mathcal{A}$
in the real experiment, and let
$\mathsf{Ideal}_{\mathcal{A},\mathcal{D}}(\lambda)$ be the analogous probability
in the ideal experiment.
We must show that, under the assumptions of Theorem~\ref{thm:pub-privacy}, there
exists a PPT simulator~$\mathcal{S}$ such that for all PPT $\mathcal{A}$ and
$\mathcal{D}$,
\[
  \bigl|\mathsf{Real}_{\mathcal{A},\mathcal{D}}(\lambda)
        - \mathsf{Ideal}_{\mathcal{A},\mathcal{D}}(\lambda)\bigr|
  \;\leq\;
  \mathsf{negl}(\lambda).
\]

\paragraph{Simulator construction.}
The simulator~$\mathcal{S}$ is given the leakage
$L_{\mathsf{pub}}(\mathrm{TPL}_i)$ and the sequence of anchoring transactions
that appear on the public Bitcoin chain.
Recall that $L_{\mathsf{pub}}$ consists of the sequence of reporting intervals
together with, for each interval, the total Bitcoin exposure
$B_{\mathrm{tot}}(t)$ across all domains and the aggregate encumbered exposure
$B_{\mathrm{enc}}(t)$ with respect to the policy metadata.
Using this information, $\mathcal{S}$ reconstructs a ``dummy'' internal
execution of TPL that is consistent with the leakage:
it samples synthetic internal domain balances and PoR/PoT transcripts whose
aggregate values match the leaked totals, subject only to the algebraic
constraints imposed by the conservation property and the policy metadata.
Because the leakage does not fix the internal structure of exposures within or
across domains, there is a large space of such dummy executions.

Whenever the real protocol would generate a PoR snapshot or PoT receipt that is exposed to $\mathcal{A}$, the simulator uses the zero-knowledge or witness-indistinguishability properties of the corresponding proof systems (assumption~(2) in Theorem~\ref{thm:pub-privacy}) to sample equally distributed proofs relative to the dummy execution.
For signatures, the simulator either uses the same signing keys as in the real protocol or wraps the signing and verification procedures inside the same zero-knowledge proof system, so that the resulting signature outputs are distributed as in an honest execution and leak no further information beyond their messages and public keys.
More concretely, for each primitive the simulator invokes the corresponding ideal simulator guaranteed by the ZK/WI definition (for the PoR/PoT proofs) or the signing oracle (for signatures), using as public inputs the values determined by the dummy execution and the leakage.
Because these simulators and oracles produce outputs that are computationally indistinguishable from real proofs and signatures, the resulting transcript is indistinguishable from one generated by an honest execution of TPL sharing the same leakage.

\paragraph{Hybrid argument.}
We formalise the above intuition by a sequence of hybrid experiments
$\mathsf{H}_0,\dots,\mathsf{H}_r$.
Hybrid $\mathsf{H}_0$ is the real experiment.
Hybrid $\mathsf{H}_r$ coincides with the ideal experiment using the simulator
$\mathcal{S}$ described above.
Each intermediate hybrid replaces one additional use of a real PoR, PoT, or
signature generation algorithm with the output of the corresponding simulator
for that primitive, while keeping the underlying leakage $L_{\mathsf{pub}}$
fixed.
Because (by assumption) each primitive is zero-knowledge or
witness-indistinguishable for its public output, the view of any PPT distinguisher
changes by at most a negligible amount at each step:
for every $j\in\{0,\dots,r-1\}$ there exists a PPT reduction from a distinguisher
between $\mathsf{H}_j$ and $\mathsf{H}_{j+1}$ to breaking the privacy of one of
the underlying primitives.
Applying the triangle inequality yields
\[
  \bigl|\Pr[\mathcal{D}(\mathsf{H}_0)=1] - \Pr[\mathcal{D}(\mathsf{H}_r)=1]\bigr|
  \;\leq\;
  \sum_{j=0}^{r-1}
    \bigl|\Pr[\mathcal{D}(\mathsf{H}_j)=1] - \Pr[\mathcal{D}(\mathsf{H}_{j+1})=1]\bigr|
  \;\leq\;
  \mathsf{negl}(\lambda),
\]
where the final inequality follows from the negligible privacy advantages of the
underlying primitives and the fact that $r$ is polynomial in~$\lambda$.

\paragraph{Conclusion.}
By construction, the distribution of $\mathsf{H}_r$ is exactly that of the
ideal experiment for the leakage function $L_{\mathsf{pub}}$ and class
$\alpha=\mathsf{pub}$.
Hence
\[
  \bigl|\mathsf{Real}_{\mathcal{A},\mathcal{D}}(\lambda)
        - \mathsf{Ideal}_{\mathcal{A},\mathcal{D}}(\lambda)\bigr|
  \;=\;
  \bigl|\Pr[\mathcal{D}(\mathsf{H}_0)=1] - \Pr[\mathcal{D}(\mathsf{H}_r)=1]\bigr|
  \;\leq\; \mathsf{negl}(\lambda),
\]
as required.
This completes the proof of Theorem~\ref{thm:pub-privacy} under the stated
assumptions.

\end{document}